\documentclass[3p]{elsarticle}
\usepackage{lineno}
\usepackage[justification=centering]{caption}
\usepackage{booktabs}
\usepackage{enumerate}
\usepackage{amsmath,amssymb,amsfonts}
\usepackage{graphicx}
\usepackage[utf8]{inputenc}	
\usepackage{rotating}
\usepackage{array}
\usepackage[toc]{appendix}		
\numberwithin{equation}{section}
\usepackage{thmtools,thm-restate}	
\usepackage{indentfirst} 
\usepackage{tikz}
\usepackage{subfig}
\usepackage{float}
\usepackage{color}	
\usepackage{accents}

\usepackage[ruled, linesnumbered]{algorithm2e}
\usepackage{makecell}
\modulolinenumbers[5]
\allowdisplaybreaks

\makeatletter
\renewcommand\subsection{\@startsection{subsection}{2}{\z@}%
    {10pt \@plus 3\p@ \@minus 2\p@}
    {8pt \@plus 3\p@ \@minus 2\p@}  
    {
    \normalfont\bfseries }}
\makeatother

\usepackage{etoolbox}
\makeatletter
\AtBeginDocument{\def\paragraph{\def\@period{}%
    \vskip 12pt\@startsection{paragraph}{4}{\z@}{6\p@ \@plus \p@}{-5\p@}{\normalfont \itshape}%
}}
\makeatother

\makeatletter
\def\ps@pprintTitle{%
 \let\@oddhead\@empty
 \let\@evenhead\@empty
 \def\@oddfoot{}%
 \let\@evenfoot\@oddfoot}
\makeatother

\newtheorem{theorem}{Theorem}[section]

\newtheorem{lemma}[theorem]{Lemma}
\newtheorem{remark}[theorem]{Remark}
\newtheorem{corollary}[theorem]{Corollary}
\newproof{proof}{Proof}
\newtheorem{definition}[theorem]{Definition}

\newcommand{\play}{P}
\newcommand{\CB}{\mathcal{GCB}_n^{X_A, X_B}}
\newcommand{\CCB}{\textit{Const-}\mathcal{CB}_n^{X_A, X_B}}
\newcommand{\FA}{ F_{A^*_i}}
\newcommand{\FB}{ F_{B^*_i}}
\newcommand{\FAn}{ F_{A^n_i}}
\newcommand{\FBn}{ F_{B^n_i}}
\newcommand{\AW}{ A^W_{\gam ,i}}
\newcommand{\BW}{ B^W_{\gam ,i}}
\newcommand{\AS}{ A^S_{\gam ,i}}
\newcommand{\BS}{ B^S_{\gam ,i}}
\newcommand{\LdaA}{\lambda^*_A}
\newcommand{\LdaB}{\lambda^*_B}
\newcommand{\barep}{\bar{\varepsilon}}

\newcommand{\prob}{\mathbb{P}}
\newcommand{\IU}{{\rm IU}^{\gam }}
\newcommand{\wmin}{\underaccent{\bar}{w}}
\newcommand{\wmax}{\bar{w}}
\newcommand{\Gmin}{\underaccent{\bar}{\gamma}}
\newcommand{\Gmax}{\bar{\gamma}}
\newcommand{\Lmin}{\underaccent{\bar}{\lambda}}
\newcommand{\Lmax}{\bar{\lambda}}

\newcommand{\LB}{\mathcal{GLB}_n^{X_A, X_B}}
\newcommand{\CLB}{\textit{Const-}\mathcal{LB}_n^{X_A, X_B}}
\newcommand{\Ex}{\mathbb{E}}
\newcommand{\Sn}{\mathcal{S}_n^{\eqref{eq:Equagamma}}}
\newcommand{\Ostar}{\Omega_A(\gam )}
\newcommand{\Y}{\mathcal{Y}^{\zeta}}
\newcommand{\X}{\mathcal{X}^{\zeta}}
\newcommand{\Ymu}{\mathcal{Y}^{\mu^R}}
\newcommand{\Xmu}{\mathcal{X}^{\mu^R}}
\newcommand{\Ynu}{\mathcal{Y}^{\nu^R}}
\newcommand{\Xnu}{\mathcal{X}^{\nu^R}}
\newcommand{\gam}{\gamma^*}
\newcommand{\Del}{\Delta_{\gam }(\zeta, \varepsilon)}
\newcommand{\Delmu}{\Delta_{\gam }(\mu^R, \varepsilon)}
\newcommand{\Delnu}{\Delta_{\gam }(\nu^R, \varepsilon)}
\newcommand{\de}{{\rm d}}
\newcommand{\mO}{\mathcal{O}}
\newcommand{\e}{e}
\newcommand{\hC}{\hat{C}}
\newcommand{\hep}{\hat{\varepsilon}}
\newcommand{\tep}{\tilde{\varepsilon}}
\newcommand{\logep}{\max \left\{1, \ln(\varepsilon^{-1}) \right\}}
\newcommand{\logone}{\max \left\{1, \ln(\varepsilon_1^{-1}) \right\}}
\newcommand{\logtwo}{\max \left\{1, \ln(\varepsilon_2^{-1}) \right\}}
\newcommand{\logbar}{\max \left\{1, \ln(\barep^{-1}) \right\}}
\newcommand{\Lip}{\mathcal{L}_{\zeta}}
\makeatletter
\newlength\tmp@\newlength\t@mp
\newcommand{\comp}[3]
  {\mathop{ \settowidth\tmp@{$\displaystyle\mathop{#1}^{#3}_{#2}$}
  \hbox to \tmp@{\hss \settowidth\t@mp{$\displaystyle #1$}\setlength\t@mp{.45\t@mp}
  $\displaystyle\mathop{#1}^{\hspace\t@mp #3}_{\hspace{-\t@mp}#2}$
  \hss} }}
\makeatother

\biboptions{authoryear}

\begin{document}

\begin{frontmatter}
\title{Approximate Equilibria in Generalized Colonel Blotto and Generalized Lottery Blotto Games}
\author[uga]{Dong Quan Vu\corref{cor1}}
\ead{dong-quan.vu@inria.fr}

\author[uga,mpi]{Patrick Loiseau}
\ead{patrick.loiseau@inria.fr}

\author[saf]{Alonso Silva}
\ead{alonso.silva-allende@safrangroup.com}


\address[uga]{Université Grenoble Alpes, Inria, CNRS, Grenoble INP, LIG, 38000 Grenoble, France}
\address[mpi]{Max Planck Institute for Software Systems (MPI-SWS), Campus E1 5, D-66123, Saarbr\"ucken, Germany}
\address[saf]{Safran Tech,  Signal and Information Technologies, 78117 Châteaufort, France}
%
%

\cortext[cor1]{Corresponding author}

\begin{abstract} 
	In the Colonel Blotto game, two players with a fixed budget simultaneously allocate their resources across $n$ battlefields to maximize the aggregate value gained from the battlefields where they have the higher allocation. Despite its long-standing history and important applications, the Colonel Blotto game still lacks a complete Nash equilibrium characterization in its most general form where players are asymmetric and battlefields’ values are heterogeneous across battlefields and different between the two players---this is called the \emph{Generalized Colonel Blotto game}. In this work, we propose a simply-constructed class of strategies---the independently uniform strategies---and we prove that they are approximate equilibria of the Generalized Colonel Blotto game; moreover, we characterize the approximation error according to the game's parameters. We also consider an extension called the \emph{Generalized Lottery Blotto game}, with stochastic winner-determination rules allowing more flexibility in modeling practical contests. We prove that the proposed strategies are also approximate equilibria of the Generalized Lottery Blotto~game.
\end{abstract}


\end{frontmatter}


\section{Introduction}
\label{sec:Intro}

The \emph{Colonel Blotto game} is one of the most well-known resource allocation games. Its description is very simple: two players, each having a fixed amount of resources (called budget), compete over a finite number $n$ of battlefields. Each battlefield is evaluated by the players with a certain value. Players simultaneously allocate their resources toward the battlefields and each player's payoff is her aggregate gains from all the battlefields. In each battlefield, the winner, who is simply the one that has the higher allocation, gains the corresponding value and the loser gains zero---this is called the winner-takes-all rule---; in battlefields with tie allocations, the value is shared between the players with a predetermined tie-breaking rule (e.g., sharing equally between them). 

Throughout its long-standing history since its first introduction by \citet{borel1921}, the Colonel Blotto game has attracted interest from different research communities for its potential to elegantly model a large range of practical situations. As a canonical example, consider the advertising competition between two firms, in a duopoly setting, that need to allocate their advertising budgets to several common markets (corresponding to battlefields)---(see e.g., \citet{fu2019multimarket} for several real-world examples). As a first approximation we can assume that, in each market, the firm having a higher advertising expenditure eventually becomes the dominant firm in that market (i.e., it holds the total share of the market). Both firms' objective is to make allocation decisions to maximize their aggregate shares from all markets. The Colonel Blotto game can also be used to model problems in military logistics, see e.g., \citet{gross1950,grosswagner}; in politics (where political parties distribute their budgets to compete over voters), see e.g., \citet{kovenock2012,laslier2002distributive,myerson1993incentives,roberson2006}; in cybersecurity (where attack/defense resources are distributed over sensitive targets), see e.g., \citet{chia2012,schwartz2014}; in online advertising (where marketing campaigns allocate ads broadcasting time over web users), see e.g., \citet{masucci2014,masucci2015}; in telecommunication (where network service providers distribute and lease their spectrum to users), see e.g.,~\citet{hajimirsaadeghi2017dynamic}. In many of these applications, it is often the case that the number of battlefields under consideration is very large.\footnote{For example, in advertising competitions between pharmacy companies, they deploy their medical representatives' effort to present/advertise new products to persuade (a large number of) doctors to prescribe their drugs (see \citet{fu2019multimarket} for more details). There, the doctors correspond to the markets over which the firms compete.}

The main focus of the literature on the Colonel Blotto game is its Nash equilibria (henceforth, simply referred to as equilibria). Completely characterizing and computing an equilibrium of the Colonel Blotto game, however, is a notoriously difficult problem. A standard approach used in the literature is to first find candidate equilibrium marginal distributions corresponding to players' allocations toward battlefields---they are called the optimal univariate distributions of the game---; and then construct an $n$-variate joint distribution (of these univariate distributions) whose realizations satisfy the budget constraints. Constructing this $n$-variate distribution, with this particular coupling constraint, is the main challenge in studying equilibria of the Colonel Blotto game. So far, there are results available only under restricting assumptions on the game's parameters such as assuming players symmetry or battlefields homogeneity, and assuming that the battlefields evaluations are identical for both players---see our detailed discussion on the related work below. In several applications of the Colonel Blotto game (e.g., in the example of advertising competitions), such assumptions are not satisfied in practice; studying more general variants of the Colonel Blotto game is therefore of prominent importance.

In this paper, we consider the most general version of the Colonel Blotto game; where the evaluations of the battlefields' values can be heterogeneous across battlefields and different between the two players, and the players' budgets can be asymmetric---we refer to it as the \emph{Generalized Colonel Blotto game} (hereinafter, GCB game). The GCB game was first formulated and studied by \citet{kovenock2020generalizations}.\footnote{Note that we consider a tie-breaking rule that is more general than that of \citet{kovenock2020generalizations} (see Section~\ref{sec:BlottoFormulation}).} Yet, \emph{the characterization of equilibria in the \emph{GCB} game remains an open question}---even the existence of an equilibrium has not been proved or disproved in general cases; and the key question remains open: how to play strategically in the GCB game to obtain good guarantees on payoffs? In this work, we also study the Nash equilibrium but we take a different angle: instead of looking for an exact equilibrium, we focus on approximate equilibria, i.e., strategies such that the extra payoff a player can gain by unilaterally deviating from them is bounded by a small term (relative to the scale of the players' total payoffs). \emph{Our first contribution is to identify a class of approximate equilibria of the \emph{GCB} game}, called the \emph{independently uniform strategies}, henceforth denoted as IU strategies.\footnote{We explain the name IU in Section~\ref{sec:IU_Strategy} and give the formal definition of approximate equilibria in Section~\ref{sec:ApproximateBlotto}.} These strategies ensure that the budget constraints are satisfied. Moreover, we construct the IU strategies such that, although their corresponding marginals do not form a set of optimal univariate distributions of the game, they approximate well-known ones. Based on this, we characterize the approximation error of this solution according to the games' parameters and show that it is negligible when the number of battlefields is sufficiently large (it quickly decreases as the number of battlefields increases). The IU strategies are also simple and efficiently computable even in large-scale~problems; which provides desired scalability in~practice.

Going beyond the Colonel Blotto game, we then consider the \emph{Lottery Blotto game}, which relaxes the winner-takes-all assumption. In the Lottery Blotto game, each player only gains a fraction of their value in each battlefield; alternatively, we can interpret it as each player winning a battlefields' value with a certain probability depending on the players' allocations on that battlefield (which can be non-zero even for the player with smaller allocation). For instance, in some advertising competitions, the sales in each market may be shared between the two firms based on their advertising expenditures, but with no firm having total dominance; these situations are discussed in \citet{duffy2015stochastic,friedman1958,kovenock2019full} as motivating examples for different variants of the Lottery Blotto game. The Lottery Blotto game model may also prove useful in other areas such as political contests for voters' attention, research and development activities, or radio-wave transmission with noises. 

In this paper, we specifically consider the Lottery Blotto game that results from replacing the winner-determination-rule in the GCB game by a (generic) contest success function---we refer to it as the \emph{Generalized Lottery Blotto game} (henceforth, GLB game). Contest success functions, studied profoundly in the rent-seeking literature (see e.g., \citet{corchon2007theory,skaperdas1996}), are functions that take the players' allocations as inputs and output the probability of winning a battlefield. The definition of a contest success functions that we adopt (see Section~\ref{sec:LotteryFormulation}) includes the winner-takes-all rule as a special case, so that the GCB game is a particular case of the GLB game. Similar to the GCB game, the equilibrium characterization is an open question for the GLB game in its general setting, with exceptions limited to a few special cases---see our detailed discussion of the related work below. \emph{Our second main contribution is to prove that the {\rm IU} strategy is also an approximate equilibrium of the {\rm GLB} game} with an approximation error that decreases quickly as the number of battlefields increases and the corresponding contest success functions converge pointwise to that of the GCB game. As an illustration, we consider the family of GLB games with ratio-form contest success functions---an important class that is often studied in the literature. We analyze the IU strategies with two of the most well-known cases of ratio-form contest success functions: the power form and the logit form; in these cases, we obtain more precise results on the convergence of the approximation~error.

\subsection{Related Work}

As stated above, the GCB game was introduced by \citet{kovenock2020generalizations}, who provide a set of distributions that are optimal univariate distributions of the GCB game. They then indicate a sufficient condition\footnote{The set of battlefields are partitioned such that two battlefields are in the same partition if they have the same (normalized) values; the sufficient condition on the attainability of equilibria requires a sufficient number of battlefields in each partition.} for these optimal univariate distributions to constitute an equilibrium that only covers a restricted range of games; and they also show a necessary condition without which there is no equilibrium satisfying such a set of~univariate distributions. On the other hand, most works in the literature focus on the more restricted \emph{constant-sum Colonel Blotto game}, where both players assign the same value to each battlefield.\footnote{It is trivial to see that in this case, the summation of players' payoffs always equals the summation of battlefields' values. We sometimes refer to the GCB game (without further assumptions) as the non-constant-sum Colonel Blotto game to distinguish it from the constant-sum version.} Even in this simpler version, complete equilibrium characterization results are available only under restrictive assumptions. When players have symmetric budgets, equilibria are constructed by \citet{borel1938} in the constant-sum Colonel Blotto game involving three battlefields and by \citet{gross1950,grosswagner} in the constant-sum Colonel Blotto game containing any number of battlefields (see also \citet{laslier2002,laslier2002distributive,thomas2017} for a modern presentation of this solution). For the constant-sum Colonel Blotto game with asymmetric budgets, equilibria characterization remains an open question in general; the exceptions are the following restricted cases: the games with only two battlefields (\citet{macdonell2015}), the games with any number of battlefields but homogeneous values (\citet{roberson2006}), and the games where there exists a sufficient number of battlefields of each possible value (\citet{schwartz2014}). Note that our results for the GCB game can be trivially adapted to the constant-sum version, i.e., the IU strategy is also an approximate equilibrium in this version. Moreover, for the constant-sum Colonel Blotto game, we show the additional result that the IU strategy is also an approximate max-min strategy (with the same approximation error).

To the best of our knowledge, only very few works in the Colonel Blotto game literature mention approximate results, and in very different settings compared to the GCB game we consider. \citet{weinstein2005} briefly discusses approximate equilibria (with a fixed approximation error) of the game with only 3 battlefields in a variant of the constant-sum Colonel Blotto game with the majority objective.\footnote{In the game with majority objective, a player wins the whole game and gains a positive payoff only if the aggregate values (or the number) of battlefields won by her exceed a given threshold, e.g. 50\% (see also \citet{kvasov2007contests,kovenock2012,laslier2005}).} Also for the majority objective, \citet{Behnezhad19a} propose a polynomial time approximation scheme for $(u,p)$-maximin strategies.\footnote{Strategies guaranteeing that a player gains a payoff at least $u$ with a probability at least $p$.} A strategy construction similar to the IU strategies can be found in \citet{vu18a} for the (constant-sum) discrete Colonel Blotto game (i.e., where the budgets and every allocation are required to be integers) with asymmetric budgets and heterogeneous battlefields. Due to the discrete condition, their analysis has essential differences to our work.\footnote{Particularly, their asymptotic results involve a double limits of the number of battlefields and the ratio of players' budgets; moreover, the convergence of the players' payoffs in their work does not have the difficulties of continuous allocation encountered in our work.}

A few works have considered other extensions and related versions of the Colonel Blotto game, though with a significantly different flavor than ours. In particular, \citet{boix-adsera2020} study a Colonel Blotto game between more than two players with symmetric budgets; \citet{kovenock2020generalizations,myerson1993incentives} completely find equilibria of the General Lotto game---a relaxed version of the Colonel Blotto game where budget constraints are only required to hold in expectation; \citet{Ahmadinejad16a,behnezhad2018battlefields,Behnezhad19a,Behnezhad17a,hart2008,hortala2012,vu18a} focus on the discrete Colonel Blotto game; \citet{powell2009,rinott2012} consider sequential Colonel Blotto games and \citet{adamo2009blotto,kovenock2011blotto,paarporn2019characterizing} study Colonel Blotto games with incomplete information.  Note that the GCB game can also be considered as an instance of larger classes such as multiple-battlefields conflicts (see e.g. \citet{kovenock2012conflicts}) and multi-item contests (see e.g., \citet{fu2019contests,kvasov2007contests,robson2005}) where battlefields' outcomes and players' payoffs are defined by more general~functions.

The idea of a game formulation similar to the GLB game (with generally defined contest success functions) appeared previously in \citet{kovenock2012conflicts}, where it was not, however, explicitly defined. More importantly, \citet{kovenock2012conflicts} only study with details a particular case where players' gains in each battlefield follow the Tullock contest success function (termed after \citet{tullock1980} and also called the lottery contest success function).\footnote{That is the GLB game where two players, called A and B, commonly evaluate each battlefield $i$ with a value $w_i$; if players allocates $x^A_i, x^B_i$ to battlefield $i$ then Player A gains $x^A_i w_i/ (x^A_i + x^B_i)$ and Player B gains $x^B_i w_i/ (x^A_i + x^B_i)$ from this battlefield.} This Lottery Blotto game with the Tullock function is also the instance that attracted the most attention in the literature: \citet{friedman1958} investigates the pure equilibrium of its constant-sum version, \citet{robson2005} generalizes these results in the case where the Tullock function is modified with a multiplicative constant representing advantages of a certain player, \citet{kovenock2019full} characterize a best-response mechanism of the non-constant-sum version, \citet{duffy2015,kim2019existence} prove the existence and partially characterize the equilibria in the majority version and \citet{kim2018lottery} provide similar results in a version with an infinite number of players. Note that the terms ``lottery Blotto game" or ``lottery Colonel Blotto game" are sometimes used in the literature (e.g., by \citet{kovenock2019full,kovenock2012conflicts}) to refer to this particular game variant that is essentially simpler than the GLB that we study. Coincidentally, the term ``Generalized Lottery Blotto" is also used by \citet{osorio2013} to refer to a game version with a slight generalization of the Tullock function (also studied by \citet{shubik1981systems}); however, only numerically computed approximate-results of the equilibrium are proposed and no tractable close-form solution is provided in general cases where battlefields' values are asymmetric across players. Note that Lottery Blotto games with Tullock function and its generalizations are included in the class of Lottery Blotto games with ratio-form contest success functions that we study in this work. Finally, a sequential Blotto-type game with generic contest success functions also appears in \citet{klumpp2019dynamics}, but it concerns the majority rule; moreover, results are only obtained under a sufficiently-concave assumption on the contest success functions.

\subsection{Roadmap and Notation}

The remainder of this paper is organized as follows. Section~\ref{sec:GamesFormulation} gives the formal definitions of the GCB and GLB game. Although the GLB game model is essentially more general, we first focus on the GCB game as it is a more classical game. Section~\ref{sec:preliminaries} provides preliminary results needed in the rest of the analysis. In Section~\ref{sec:ApproximateBlotto}, we define the IU strategy and show that any IU strategy is an approximate equilibrium of the GCB game. In Section~\ref{sec:LotteryApproximation}, we study how the IU strategy is also an approximate equilibrium of the GLB game. We  conclude in Section~\ref{Conclu}. Finally, detailed proofs of all lemmas and theorems are given in Appendix. 

Throughout the paper, we use bold symbols (e.g., $\boldsymbol{x}$) to denote vectors and subscript indices to denote its elements (e.g.,
$\boldsymbol{x} = (x_1, x_2, \ldots , x_n)$). The notation $[n]$ denotes the set $\{1,2, \ldots , n\}$, for any $n \in \mathbb{N} \backslash \{0\}$. We often use the letter $\play$ to denote a player and use $-\play$ to indicate her opponent in the games. $R^n_{\ge 0}$ denotes the set of all $n$-tuples whose elements are non-negative ($R_{\ge 0}:=R^1_{\ge 0}$). We denote the Euler's number by~$\e$. For any random variable $X$, we use $F_{X}$ and $\Ex X$ to denote its corresponding cumulative density function (abbreviated by CDF) and its expectation respectively; for an event $E$, we denote the probability that it happens by $\prob(E)$. We use the asymptotic notation $\mO$ with its standard definition and $\tilde{\mO}$ as a variant of $\mO$ where logarithmic terms are ignored (see \ref{sec:appen_preliminary} for formal definitions). We sum up the notations used in this work in Table~\ref{table:notation} (\ref{sec:appen_preliminary}).

\section{Games Formulation}
\label{sec:GamesFormulation}
In this section, we define the two games that are our main focus: in Section \ref{sec:BlottoFormulation}, we introduce the Generalized Colonel Blotto game (henceforth, GCB game) and its restricted version---the constant-sum Colonel Blotto game; in Section  \ref{sec:LotteryFormulation}, we present the Generalized Lottery Blotto game (henceforth, GLB game), as an extension of the GCB~game.

\subsection{The Generalized Colonel Blotto Game}
\label{sec:BlottoFormulation}
We consider the following one-shot, complete information game between two players A and B.  Each player has a fixed amount of resources (called the \emph{budgets}), denoted $X_A$ and $X_B$, respectively. Without loss of generality, we assume that \mbox{$0< X_A \le X_B$}. Players simultaneously allocate their resources across $n$ \emph{battlefields} \mbox{($n \ge 3$)}. Each battlefield $i \in [n]$ is embedded with two parameters $w^A_i, w^B_i > 0$, corresponding to the \emph{values} at which Player A and Player B respectively assess this battlefield. 
A \emph{pure strategy} of player ${\play} \in \left\{ A, B \right\}$ is a vector $\boldsymbol{x}^{\play} = {\left( {x_i^\play} \right)_{i \in [n]}} \in \mathbb{R}_ {\ge 0} ^n$ that satisfies the budget constraint $\sum\nolimits_{i = 1}^n {x_i^{\play} \le {X_{\play}}}$. In each battlefield $i$, when player ${\play}$ allocates strictly more than her opponent, she gains completely her embedded values $w^{\play}_i$ while the opponent gains $0$. In case of a tie, i.e., if~$x^A_i = x^B_i$, then Player A receives $\alpha w^A_{i}$ and Player B receives $(1-\alpha) w^B_{i}$, where \mbox{$\alpha \in \left[ 0, 1 \right]$} is a fixed parameter. Each player's payoff is the summation of values she gains from all battlefields; formally, for any pure strategy profile $(\boldsymbol{x}^A, \boldsymbol{x}^B)$, the payoffs of players A and B are \mbox{$\Pi_A(\boldsymbol{x}^A, \boldsymbol{x}^B) = \sum \nolimits _{i=1}^n {w^A_i\cdot \beta_A \left( {x_i^A,x_i^B} \right)}$} and \mbox{$\Pi_B(\boldsymbol{x}^A, \boldsymbol{x}^B) = \sum \nolimits _{i=1}^n {w^B_i \cdot \beta_B \left( {x_i^A,x_i^B} \right)}$} respectively; here, $\beta_A$ and $\beta_B$ (henceforth, we called them the Blotto functions) are functions defined as follows:
\begin{equation}
    \beta_A\left( {x,y} \right) = \left\{ \begin{array}{l}
    1 \text{ , if } x >y\\
    \alpha \text{ , if } x =y \\
    0 \text{ , if } x < y
    \end{array} \right. \quad
    \textrm{ and } \quad
    \beta_B\left( {x,y} \right) = \left\{ \begin{array}{l}
    1 \text{ , if } y >  x\\
    1-\alpha \text{ , if } y = x \\
    0 \text{ , if } y < x
    \end{array} \right., \textrm{ for all } x, y \in \mathbb{R}_ {\ge 0}. \label{eq:betafunction}
\end{equation}
\begin{definition}
\label{def:BlottoGame}
    \textbf{A Generalized Colonel Blotto game} with $n$ battlefields and budgets $X_A, X_B$, denoted by ${\CB}$, is the game defined above; in particular, the strategy set of player $\play \in \{A,B\}$ is \mbox{$\{\boldsymbol{x}^\play \in \mathbb{R}^n_{\ge 0}: \sum \nolimits_{i=1}^n{x^{\play}_i \le X_{\play} } \}$} and her payoff is $\Pi_{\play} (\boldsymbol{x}^A, \boldsymbol{x}^B)$ when players A and B play the pure strategies $\boldsymbol{x}^A$ and $\boldsymbol{x}^B$ respectively. 
\end{definition}
To lighten the notation, we only include the subscript $n$ (the number of battlefields) and $X_A$, $X_B$ (the values of players' budgets) in the notation ${\CB}$ and omit the other parameters; in particular, the values $\alpha$ and $w^A_i, w^B_i$ for $i \in [n]$. In this game, a \emph{mixed strategy} is a joint distribution on the allocations of all battlefields, such that any drawn pure strategy of a player is an $n$-tuple that satisfies her budget constraint. We reuse the notations $\Pi_A\left(s^A, s^B \right)$ and $\Pi_B\left(s^A, s^B \right)$ to denote the payoffs of players A and B when they play the mixed strategies $s^A$ and~$s^B$, respectively.
Note that the definition of $\CB$ above allows asymmetry in players' budgets and heterogeneity in battlefields values; moreover, it allows battlefield values to differ between the two players. Furthermore, the defined payoff functions can be understood as if we randomly break the tie (if it happens) such that Player A wins battlefield $i$ with probability $\alpha$ while Player B wins it with probability $(1-\alpha)$. This includes all the classical tie-breaking rules considered in the literature; for instance, the rule of giving the whole value to Player B used by \citet{roberson2006, schwartz2014} corresponds to $\alpha=0$; the 50-50 rule used by \citet{Ahmadinejad16a,Behnezhad17a,kovenock2020generalizations} corresponds to~\mbox{$\alpha=1/2$}.

In this paper, we also often work with the \emph{normalized values} of the battlefields defined as \mbox{$v^A_i:= {w^A_i}/{W_A}$} and \mbox{$v^B_i:= {w^B_i}/{W_B}$}, where \mbox{$W_A:=\sum \nolimits_{j=1}^n w^A_j$} and \mbox{$W_B:=\sum \nolimits_{j=1}^n w^B_j$} for $i \in [n]$. We trivially observe that $v^{\play}_i \in \left[0, 1\right]$ for all $i$ and that $\sum \nolimits_{j=1}^n {v^{\play}_j} = 1$. Most of our analysis relies on an additional assumption that the battlefields' values are bounded away from zero and infinity; particularly, we work with the following condition---namely, Assumption $(A0)$---: 
\begin{equation*}
     \exists \wmin, \wmax >0:  \wmin \leq w^\play_i  \leq \wmax, \forall i \in [n], \forall \play \in \{A, B\}. \tag{Assumption $(A0)$} \label{eq:A0}
\end{equation*}
\ref{eq:A0} is a fairly mild assumption that is satisfied in most of (if not all) practical applications. As a direct consequence, the normalized values satisfy
\begin{equation}
    \frac{\wmin}{n \wmax} \le v^{\play}_i \le \frac{\wmax}{n \wmin}, \quad \forall i \in [n], \forall \play \in \{A, B\}. \label{eq:bound_v^p_i}
\end{equation}

Finally, we note that most works in the literature (the only exception, to our knowledge, being the work of \citet{kovenock2020generalizations}) focus only on the \emph{constant-sum} Colonel Blotto game where players have the same evaluations on battlefields' values. The game $\CB$ given in Definition~\ref{def:BlottoGame} is more general; hence all our results for $\CB$ can be straightforwardly applied to this constant-sum version as well. However, for the purpose of comparing with the literature and because we can show stronger results in this special case, it is useful to also formally define the constant-sum game variant as follows.
\begin{definition}
\label{def:constantsumGame}
   A \textbf{constant-sum Colonel Blotto game}, denoted by $\CCB$, is a game that has the same formulation as the game $\CB$ but with the additional condition that \mbox{$w^A_i= w^B_i, \forall i \in [n]$}.
\end{definition}
As a trivial corollary of this additional condition, in $\CCB$, players also have common normalized valuation on battlefields, i.e., \mbox{$v^A_i \!= \! v^B_i$} for all $i \in [n]$ and the players' maximum payoffs are equal, i.e., \mbox{$W_A=W_B$}.

\subsection{Contest Success Functions and the Generalized Lottery Blotto Game}
\label{sec:LotteryFormulation}

In this section, we present the Generalized Lottery Blotto game (GLB game) that extends the model of the GCB game. This new game is based on the notion of contest success functions (henceforth, CSFs), that we introduce below before defining the game model.

Contest success functions are functions that quantify the winning probability in \emph{contests} (also called \emph{rent-seeking} competitions) where several players compete for a single prize by exerting resources/efforts. CSFs can be defined for any number of players (see e.g., a general definition by~\citet{skaperdas1996}), but in this work, we focus only on the case of two players.
\begin{definition}
    \label{def:CSF_general}
    \mbox{$\mathcal{\zeta}_A: \mathbb{R}^2_{\ge 0} \to \mathbb{R}$} and \mbox{$\mathcal{\zeta}_B: \mathbb{R}^2_{\ge 0} \to \mathbb{R}$} is a pair of contest success functions (\textbf{CSF}s) if and only if the following two conditions are satisfied:
    \begin{itemize}
        \item[$(C1)$] $\zeta_A(x,y),\zeta_B(x,y) \ge 0$ and $\zeta_A(x,y) + \zeta_B(x,y)= 1$, $\forall x,y\ge 0$. 
        \item[$(C2)$] $\zeta_{A}(x, y)$ (resp. $\zeta_{B}(x, y)$) is non-decreasing in $x$ (resp. in $y$) and non-increasing in $y$ (resp. in $x$).
    \end{itemize}
\end{definition}

Intuitively, the function $\zeta_A$ (resp. $\zeta_B$) maps any pair of players' invested resources to the probability that Player A (resp. Player B) wins the prize. Condition~$(C1)$ indicates that the outputs of any pair of the CSFs always satisfy the condition of a probability distribution. On the other hand, Condition~$(C2)$ states that a player's winning probability increases (or at least stays the same) when she increases her effort and decreases (or at least stays the same) when her opponent increases her effort. We note that Definition~\ref{def:CSF_general} allows a more general definitions of the CSFs (in two-player cases) compared to the definition given by \citet{clark1998contest,hirshleifer1989conflict,skaperdas1996} that contains other assumptions.\footnote{For example, \citet{skaperdas1996} defines $\zeta_A$, $\zeta_B$ with an axiom of anonymity; they also require that any player who puts a strictly positive amount of resources has a strictly positive probability of winning the prize; \citet{clark1998contest} considers the CSFs additionally satisfying the Choice Axiom. These are technical conditions needed for proving their results and we omit them here lest they unnecessarily limit our scope of study.} While many of the CSFs considered in the literature are continuous functions, we do not include continuity in Definition~\ref{def:CSF_general} to keep the generality. Importantly, the Blotto functions $\beta_A, \beta_B$ of the game $\CB$ (i.e., the winner-takes-all rule) satisfy Conditions $(C1)$ and $(C2)$, hence $\beta_A, \beta_B$ are CSFs. Besides these functions, some examples of other CSFs considered in the literature~are: 
\begin{enumerate}[(a)]
    \item $\zeta_A(x,y) = x/(x+y)$ and $\zeta_B(x,y) = y/(x+y)$, proposed by \citet{tullock1980};
    \item $\zeta_A(x,y) = \max\left\{ \min \left\{\frac{1}{2} \! +\! C(x\!-\!y),1 \right\},0 \right\} $ and $\zeta_B(x,y) = 1- \zeta_A(x,y)$, proposed by~\citet{che2000difference}, where $C>0$ is a fixed parameter;
    \item $\zeta_A(x,y) = \frac{1}{2} - \frac{y-x}{2y}$ if $x \le y$ and $\zeta_A(x,y) = \frac{1}{2} + \frac{x-y}{2x}$ if $x \ge y$; and $\zeta_B(x,y)=1-\zeta_A(x,y)$, proposed by \citet{alcalde2007tullock}.
\end{enumerate}
%



Building on the notion of CSFs and the GCB game, we now define a new game model based on the following~idea: in a game $\CB$, we view each battlefield as a contest between players where the prize is the battlefield's value and players' effort correspond to their allocations; by doing this, each pair of CSFs defines an instance of a game where the probability of winning a battlefield follows them accordingly. 
\begin{definition}
\label{def:LotteryGame}
    Let $\zeta = (\zeta_A,\zeta_B)$ be a pair of CSFs, a \textbf{Generalized Lottery Blotto game} with $n$ battlefields and budgets $X_A, X_B$, denoted $\LB(\zeta)$, is the game with the same players A and~B and the same strategy sets as in $\CB$; but where payoffs are given, for any pure strategy profile $(\boldsymbol{x}^A,\boldsymbol{x}^B)$,~by
    \begin{equation*}
        \Pi_A^{\zeta}(\boldsymbol{x}^A, \boldsymbol{x}^B) = \sum \nolimits _{i=1}^n {w^A_i\cdot \zeta_A \left( {x_i^A,x_i^B} \right)} \qquad \textrm{and } \qquad \Pi_B^{\zeta}(\boldsymbol{x}^A, \boldsymbol{x}^B) = \sum \nolimits _{i=1}^n {w^B_i \cdot \zeta_B \left( {x_i^A,x_i^B} \right)}. 
    \end{equation*}
\end{definition}

The GLB game model is more flexible than that of the GCB game, as it allows choosing the CSFs that define the players' payoffs for each specific practical situation. Throughout the paper, to refer to a GCB game $\CB$ that has the same parameters $n,X_A,X_B, w^A_i$, $w^B_i, \forall i \in [n]$ as a GLB game $\LB$, we call $\CB$ the \emph{corresponding game} of $\LB$ and vice versa. Intuitively, the players' payoffs in the $\LB$ game can be seen as the expected payoffs in the corresponding GCB game with respect to the following random process determining the winner in any battlefield $i$: Player A wins with probability $\zeta_A(x^A_i,x^B_i)$ and Player B wins with probability $\zeta_B(x^A_i,x^B_i)$ if they allocate $x^A_i$ and $x^B_i$ respectively. Similar to the game $\CB$, players' payoffs in the $\LB$ game are also monotonic with respect to the allocations in a~battlefield (due to Condition~$(C2)$). Note that, to derive our results for $\LB$, we will also use \ref{eq:A0} introduced above.

Besides the GLB game with generally defined CSFs, we additionally consider the games with CSFs that belong to a special class called the \emph{ratio-form} CSFs. These are the CSFs that are studied the most profoundly in the literature. We will use the games with these ratio-form CSFs to illustrate the results obtained in the GLB game. 
\begin{definition}
    \label{def:CSF_ratioform}
    Functions \mbox{$\mathcal{\zeta}_A, \mathcal{\zeta}_B: \mathbb{R}^2_{\ge 0} \to \mathbb{R}_{\ge 0}$} are called \textbf{ratio-form CSFs} if they have the~form:
    \[\zeta_A(x,y) = \frac{\eta(x)}{\eta(x)+ \kappa(y)} \quad \textrm{and} \quad \zeta_B(x,y) = \frac{\kappa(y)}{\eta(x)+ \kappa(y)}, \]
    where $\eta, \kappa: \mathbb{R}_{\ge 0} \to \mathbb{R}$ are non-negative functions such that $\zeta_A$ and $\zeta_B$ satisfy Conditions $(C1)$ and~$(C2)$ of Definition~\ref{def:CSF_general}.
\end{definition}
Two classical ratio-form CSFs in the literature (see e.g., \citet{corchon2010foundations,hillman1989}) are the power form where \mbox{$\eta(z)=\kappa(z)=z^R, \forall z \ge 0$} and the logit form where \mbox{$\eta(z)= \kappa(z) = e^{Rz},\forall z \ge 0$}, where $R>0$ is a parameter chosen a priori. These functions yield the sharing 50-50 tie-breaking rule, i.e., $\zeta_A(x,y) = \zeta_B(x,y) =1/2$ if $x=y$. We define in Table~\ref{table:CSF} the generalized versions of these ratio-form CSFs---namely, $\mu^R$ and $\nu^R$---using the parameter $\alpha\in (0,1)$ that leads to the general tie-breaking rule as in the GCB game $\CB$.\footnote{When $\alpha=1/2$, the CSFs $\mu^R$ and $\nu^R$ match the classical power form and logit form CSFs. Note that we exclude the cases where $\alpha =0$ or $\alpha =1$ since these are the trivial cases: in the corresponding GLB game, a player, say $\play \in \{A,B\}$, always has the payoff $W_\play$ while player $-\play$'s payoff is always zero regardless how they allocate their resources.} It is trivial to verify that both pairs $(\mu^R_A, \mu^R_B)$ and $(\nu^R_A, \nu^R_B)$ satisfy the Conditions $(C1)$ and $(C2)$. Henceforth, we use the terms power and logit form to indicate the CSFs $\mu^R$ and $\nu^R$ respectively and use the term \emph{Lottery Blotto game with ratio-form CSF} to commonly address the games $\LB(\mu^R)$ and $\LB(\nu^R)$ (i.e., the games formulated by replacing $\zeta$ in Definition~\ref{def:LotteryGame} by either $\mu^R$ or $\nu^R$). An important remark is that both the power and logit form CSFs converge pointwise toward the Blotto functions $\beta_A, \beta_B$ as $R$ tends to infinity (see~Section~\ref{sec:Approx_Lottery_ratio} for more details). This convergence can be observed in Figure \ref{fig1} that illustrates several instances of the ratio-form CSFs in comparison with the Blotto~functions. 
 
\begin{table}[ht]
	\centering
	\caption{Power and logit form CSFs with generalized tie-breaking rule ($\alpha \in (0,1)$).}
    \label{table:CSF}
    \begin{tabular}{|c|c|c|c|}
\hline
        \centering
        {Name} & Notation
        & {If $x^2+ y^2 >0$} & {If $x=y=0$} \\
\hline
        Power form & $\mu^R:=\! (\mu^R_A,\mu^R_B)$ & \makecell{\hspace{-0.30cm}$\mu^R_A(x,y)= \frac{\alpha  {x} ^R}{\alpha  {x} ^R \! +\! (1 \!-\! \alpha) {y} ^R}$ ; \quad \hspace{0.35cm}$\mu^R_B(x,y)= \frac{(1\!-\!\alpha){y} ^R}{\alpha {x}^R \! +\! (1 \! -\! \alpha) {y} ^R}$} & \makecell{\hspace*{-0.60cm}$\mu^R_A(x,y) = \alpha$ \\ $\mu^R_B(x,y)=1-\alpha$} \\
\hline
        Logit form  & $\nu^R := \!(\nu^R_A,\nu^R_B)$ & \makecell{$\nu^R_A(x,y)= \frac{\alpha e^{xR}}{\alpha e^{xR} + (1-\alpha)e^{yR}}$; \quad $\nu^R_B(x,y)= \frac{(1-\alpha)e^{yR}}{\alpha e^{xR} + (1-\alpha)e^{yR}}$} & \makecell{\hspace*{-0.55cm}$\nu^R_A(x,y) = \alpha$ \\ $\nu^R_B(x,y)=1-\alpha$} \\
\hline 
    \end{tabular}
\end{table}

\begin{figure*}[htb]
\centering
    \begin{turn}{90} { \hspace{1.5cm} $\mu^R_A(x, 4)$}  \end{turn}
    \subfloat[Power form CSF, $X_A = 10, \alpha = 0.6$.]{{\includegraphics[height=0.26\textwidth]{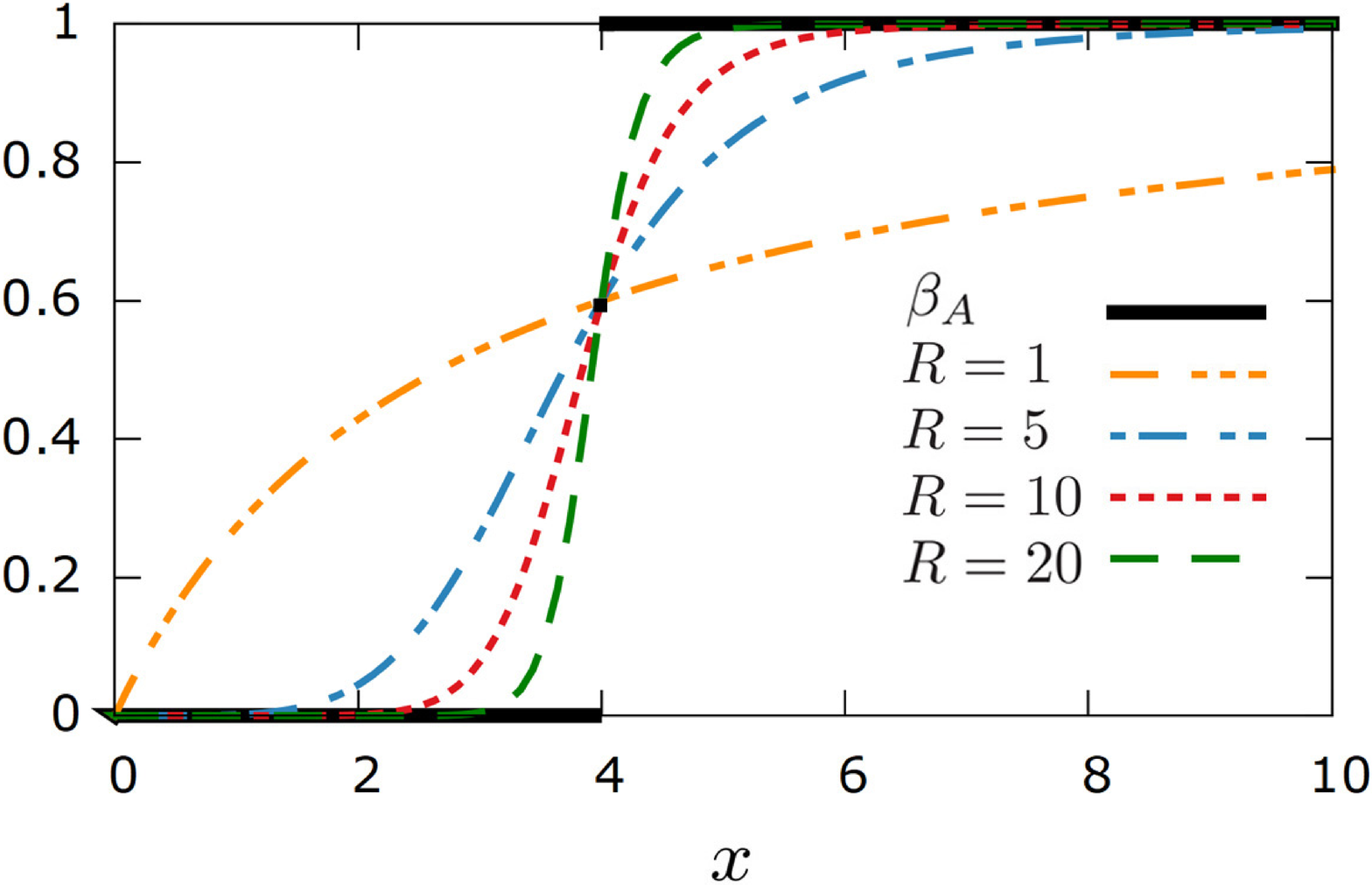} }} 
     \qquad
    \begin{turn}{90} { \hspace{1.5cm} $\nu^R_A(x,4)$} \end{turn}
    \subfloat[Logit CSF, $X_A = 10, \alpha = 0.6$.]{{\includegraphics[height=0.26\textwidth]{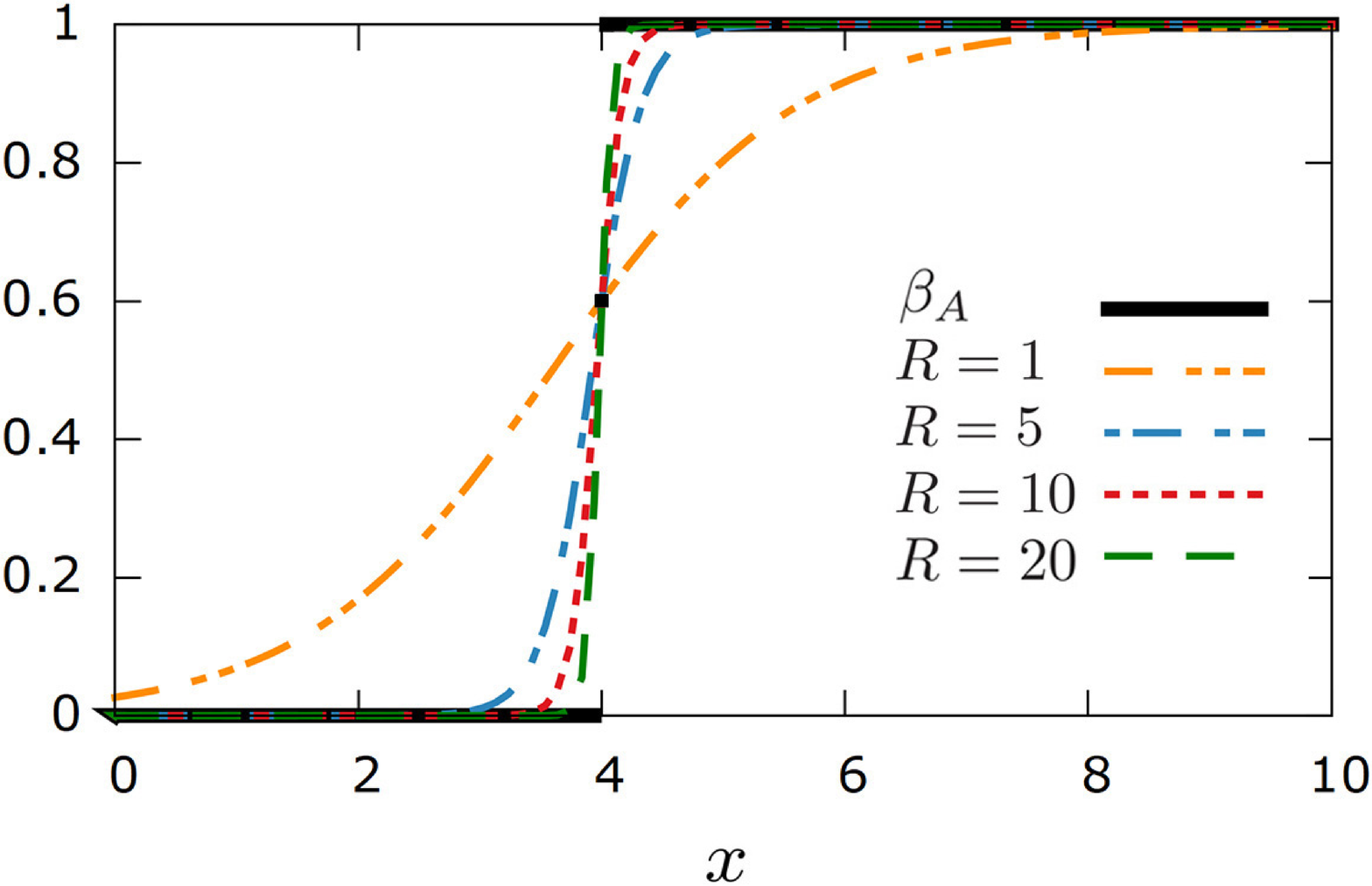} }}%
    \caption{Examples of power form and logit form CSFs in comparison with the Blotto functions.}
    \label{fig1}
\end{figure*}

%
%
\section{Preliminaries}
\label{sec:preliminaries}

In this section, we briefly review some results from the literature that are useful for our analyses of the GCB game and the GLB game; and we show new bounds on the involved parameters, based on \ref{eq:A0}, that are essential for the asymptotic analysis in the next sections. 

The (Nash) equilibrium characterization of the GCB game still remains an open question. However, under certain assumptions, a set of \emph{optimal univariate distributions} (henceforth, OUDs)\footnote{$\left\{F^A_i, F^B_i: i \in [n] \right\}$ is a set of {OUDs} of the game $\CB$ if  they hold the following conditions: $(i)$ $F^A_i, F^B_i$ have supports that are subset of $\mathbb{R}_{\ge 0}$; $(ii)$~\mbox{$\sum_{i \in [n]} {\Ex_{x_i \sim F^\play_i }[ x_i ] } \le X_{\play}$}; $(iii)$ if player $\play$ draws her allocation to battlefield $i$ from $F^\play_i$, $\forall i \in [n]$, Player $-\play$ has no pure strategy inducing a better payoff than when she draws her allocation to battlefield $i$ from $F^{-\play}_i$.} of players is well-known (see \citet{kovenock2020generalizations}). Particularly, observe that we can break down the problem of finding the best-response of a player against a fixed strategy of her opponent into solving $n$ all-pay auctions involving the Lagrange multipliers corresponding to the budget constraints (see e.g., \citet{kovenock2020generalizations,roberson2006,schwartz2014}). The equilibrium of two-player all-pay auctions is well-known and can be expressed as uniform-type distributions (see e.g., \citet{baye1996all,hillman1989}).
These uniform-type distributions derive directly a set of OUDs of the GCB game. Importantly, we have an equilibrium if we can construct a joint distribution with these OUDs such that its realizations always satisfy the budget constraints. However, as mentioned in Section~\ref{sec:Intro}, the existence of such a construction is known only for some special cases and remains unknown in the general setting of $\CB$. Note that if we consider a relaxation of the game that only requires the budget constraints to be hold in expectation (this relaxation is called the General Lotto game by \citet{kovenock2020generalizations}), an equilibrium is to independently draw allocations from the uniform-type~distributions.  

Although in this work we do not attempt to solve the open question of the equilibria characterization of the GCB game, we still use several preliminary results from this approach to construct an approximate equilibrium of the games. We present these results below, using a notation similar to that of \citet{kovenock2020generalizations}.

For each instance of the game $\CB$ (and of the game $\LB$), for any $\gamma \in (0,\infty)$, we~define
\begin{equation*}
    \Omega_A(\gamma):= \left\{ i \in [n] : {v^A_i}/{v^B_i} > \gamma \right\},
\end{equation*}
and consider the following equation with the variable $\gamma$ :
\begin{equation}
    \frac{X_B \gamma}{X_A} = \frac{\gamma^2 \sum\nolimits_{i \in \Omega_{A}(\gamma)}{\frac{(v^B_i)^2}{v^A_i}} + \sum\nolimits_{i \notin \Omega_{A}(\gamma)}{v^A_i}} {\sum\nolimits_{i \in \Omega_{A}(\gamma)}{v^B_i} + \frac{1}{\gamma^2} \sum\nolimits_{i \notin \Omega_{A}(\gamma)}{\frac{(v^A_i)^2}{v^B_i}} }.
\label{eq:Equagamma}
\end{equation}

Let us denote by $\Sn$ \emph{the set containing all positive solutions} of Equation~\eqref{eq:Equagamma} corresponding to the game $\CB$ (or $\LB$).\footnote{Note that \eqref{eq:Equagamma} and $\Sn$ also depend on other parameters of the game $\CB$ but we use the notation with only the subscript $n$ and omit other parameters to lighten the notation.} Based on the intermediate value theorem, the following lemma is proved by \citet{kovenock2020generalizations}.
\begin{lemma}
\label{lem:positivegamma}
    For any game $\CB$ (or $\LB$), Equation \eqref{eq:Equagamma} has at least one positive solution; i.e., $\Sn \neq \emptyset$.
\end{lemma}
%
%
Equation~\eqref{eq:Equagamma} may have more than one solution and solving it (i.e., finding algebraic expressions of its solutions) can be done in $\mO(n \ln(n))$ time.\footnote{To solve this equation, we first sort out all ratios ${v^A_i}/{v^B_i}$ in a non-decreasing order (which can be done in $\mO(n \ln(n))$), then there are three possible cases: $\gam  < \min\{v^A_i/ v^B_i, i \in [n]\}$ or $\gam  \ge \max\{v^A_i/ v^B_i, i \in [n]\}$ or $\exists j: \gam  \in \left[ {v^A_{j}}/{v^B_j},{v^A_{j+1}}/{v^B_{j+1}} \right)$. In all of these cases, Equation \eqref{eq:Equagamma} becomes a cubic equation; therefore, it can be solved algebraically.} Now, corresponding to each positive solution $\gam  \in \Sn$, we define two constants,\footnote{These constants are the Lagrange multipliers corresponding to the budget constraints in finding players' best-response; see \citet{kovenock2020generalizations} for more details.} namely $\LdaA$ and $\LdaB$ as follows:
\begin{align}
    & {\lambda _A^*}: = \frac{({\gam })^2}{2{X_B}}\sum\limits_{i \in {\Omega _A}({\gam })} {\frac{{{{\left( {v_i^B} \right)}^2}}}{{v_i^A}}}  + \frac{1}{2{X_B}}\sum\limits_{i \notin {\Omega _A}({\gam })} {v_i^A} , \label{eq:lambdaA} \\
    & {\lambda _B^*}: = \frac{1}{2{X_A}}\sum\limits_{i \in {\Omega _A}({\gam })} {v_i^B}  + \frac{1}{2(\gam ) ^2 {X_A}}\sum\limits_{i \notin {\Omega _A}({\gam })} {\frac{{{{\left( {v_i^A} \right)}^2}}}{{v_i^B}}}.\label{eq:lambdaB}
\end{align}
Note importantly that we have $\gam  = \LdaA / \LdaB$ (see Lemma~\ref{lem:Preliminary} in~\ref{sec:appen_preliminary} for a proof). We now use these constants $\LdaA$ and $\LdaB$ to define several important distributions.
\begin{definition}
\label{def:UnifromDistributions}
Given a game $\CB$ (or $\LB$), for any $\gam  \in \Sn$ and the corresponding constants $\LdaA, \LdaB$, we define the following random variables and distributions,\footnote{Here, the superscripts $S$ and $W$, standing for strong and weak, are used to emphasize the intuition on players' incentive to play according to these distributions in the Colonel Blotto games: if \mbox{$i \in \Ostar:= \left\{i: v^A_i/\LdaA > v^B_i/\LdaB \right\}$},  Player A has a ``stronger" incentive to win battlefield $i$ and Player B has a ``weaker" incentive; if \mbox{$i \notin \Ostar$}, the roles of players are exchanged.} for each $i \in [n]$:
\begin{enumerate}[(a)]
    \item If $i \in \Ostar$ (i.e., $\frac{v^A_i}{\LdaA} \!  > \! \frac{v^B_i}{ \LdaB}$), we define $\AS$ and $\BW$ as the random variables whose distributions~are
        \begin{align}
        & F_{\AS}\left( x \right) := \frac{x \LdaB}{v^B_i}, \forall x \in \left[0, \frac{v_i^B}{\LdaB}\right], \label{As}
        \\
        & {F_{\BW}}\left( x \right) := \frac{\frac{v^A_i}{\LdaA} - \frac{v^B_i}{\LdaB}}{\frac{v^A_i}{\LdaA}} + \frac{x \LdaA}{v^A_i}, \forall x \in \left[0, \frac{v_i^B}{\LdaB}\right]. \label{Bw}
        \end{align}
    \item If $i \notin \Ostar$ (i.e., $\frac{v^A_i}{\LdaA} \! \le \! \frac{v^B_i}{\LdaB}$), we define $\AW$ and $\BS$ as the random variables whose distributions~are
        \begin{align}
        & {F_{\AW}}\left( x \right) := \frac{\frac{v^B_i}{\LdaB} - \frac{v^A_i}{\LdaA}}{\frac{v^B_i}{\LdaB}} + \frac{x \LdaB}{v^B_i}, \forall x \in \left[0, \frac{v_i^A}{\LdaA}\right],\label{Aw} 
        \\
        & F_{{\BS}}\left( x \right) := \frac{x \LdaA}{v^A_i}, \forall x \in \left[0, \frac{v_i^A}{\LdaA}\right]. \label{Bs}
        \end{align}
\end{enumerate}
To lighten the notation, hereinafter, we often commonly denote these random variables as follows (the corresponding distributions  are denoted by $\FA$ and $\FB$):
\begin{align}
    & A^*_i := \left\{
    	\begin{array}{ll}
    		\AS  & \mbox{if } i \in \Ostar  \\
    		\AW & \mbox{if } i \notin \Ostar
    	\end{array}
    \right.
\textrm{ and } 
    \hspace{-2.5cm}& B^*_i := \left\{
	\begin{array}{ll}
		\BS  & \mbox{if } i \notin \Ostar  \\
		\BW & \mbox{if } i \in \Ostar
	\end{array}
    \right.. \label{A*B*}
\end{align}
\end{definition}
%

We term these distributions the \emph{uniform-type distributions}: $F_{\AS}\left( x \right)$ is the continuous uniform distribution on $\left[ 0, {v^B_i}/{\LdaB} \right]$ and $F_{\BW}\left( x \right)$ is the distribution placing a positive mass $\left( \frac{v^A_i}{\LdaA}\!-\! \frac{v^B_i}{\LdaB}\right)\! \Big/ \! \frac{v^A_i}{\LdaA}$ at $0$ and uniformly distributing the remaining mass on  $\left( 0, {v^B_i}/{\LdaB} \right]$; similarly, $F_{\BS}$ is the uniform distribution on $[0, v^A_i / \LdaA]$ and $F_{\BW}$ is uniform on $(0,v^A_i/\LdaA]$ with a positive mass at $0$.

If Player A can construct and plays a mixed strategy such that her sampled allocation to any battlefield $i \in [n]$ follows the distribution $\FA$, it is optimal for Player B to play such that her allocation to $i$ follows $\FB$ (if it is possible) and vice versa. We will revisit this result (with more details) in Section~\ref{sec:ApproximateBlotto} and in Lemma~\ref{lem:best_response} in~\ref{sec:Appen_Proof_TheoBlotto}. Importantly, under the condition that Player A and Player B can respectively construct joint distributions of \mbox{$\FA, \forall i \in [n]$} and \mbox{$\FB, \forall i \in [n]$} such that their sampled allocations satisfy the budget constraint, these mixed strategies yield an equilibrium of the game $\CB$. However, in general, that condition does not always hold. For instance, although $A^*_i$ and $B^*_i$ have finite upper-bounds,\footnote{Trivially from Proposition~\ref{Prop:BoundLambda}, the random variables $A^*_i, B^*_i, \forall n , \forall i \in [n]$ are upper-bounded by $\wmax / (\wmin n \Lmin)$. In the remainders of the paper, we sometimes need an upper-bound of these random variables that does not depend on $n$: we can prove that they are bounded by $2 X_B$ (see Lemma~\ref{lem:Preliminary} in~\ref{sec:appen_preliminary}).} we note that among these random variables, some may (with strictly positive probability) exceed the budgets $X_A, X_B$ for certain parameters' configuration of the game; therefore, allocating according to $\FA, \FB$ may violate the budget constraints and it is then trivial that there exists no equilibrium yielding $\FA, \FB, \forall i \in [n]$ as marginals. On the other hand, given fixed $X_A, X_B$, if $n$ is large enough, we can guarantee that $A^*_i, B^*_i$ do not exceed the budgets for each $i$; however, even in this case, we still do not have guarantees on the summation of allocations sampled from all $A^*_i, B^*_i, i \in [n]$, i.e., it is still unknown if there exists an equilibrium yielding $\FA,\FB, i \in [n]$  as marginals. Note importantly that the budget-constraints violation of $A^*_i, B^*_i$ does not affect our work and our results hold for any parameters' configuration of the games.

%
%
%

%
Finally, under \ref{eq:A0}, we obtain a novel result, presented below as Proposition~\ref{Prop:BoundLambda}, stating that the parameters $\gam , \LdaA$ and $\LdaB$ are all bounded. The proof of this proposition is given in \ref{sec:appen_preliminary}. From this proof, we observe that the bounds of $\gam , \LdaA$ and $\LdaB$ depend only on the ratios $\wmax/\wmin$ and $X_B/X_A$; when these ratios increase, the ranges in which $\gam$ and $\LdaA,\LdaB$ belong to also become larger (i.e., the ratios $\Gmax/\Gmin $ and $\Lmax/\Lmin$ also~increase). The main results of this work are based on asymptotic analyses in terms of the number of battlefields of the games; therefore, it is noteworthy that the bounds of these parameters \emph{do not} depend on $n$.
\begin{restatable}{proposition}{boundpropo}
    \label{Prop:BoundLambda}
    Given $\wmin, \wmax, X_A, X_B >0$ ($\wmin \le \wmax$, $X_A \le X_B$), there exist constants $\Gmin, \Gmax, \Lmin, \Lmax\! >\!0$ such that for any game $\CB$ (or $\LB$) satisfying \ref{eq:A0}, any $\gam \! \in\! \Sn$ and its corresponding $\LdaA, \LdaB$, we have \mbox{$\Gmin\! \le\! \gam \!  \le\! \Gmax$} and \mbox{$\Lmin\! \le\! \LdaA, \LdaB\! \le\! \Lmax$}.
\end{restatable}
%

%
%
%
%
%
\section{Approximate Equilibria of the Generalized Colonel Blotto Game}
\label{sec:ApproximateBlotto}
In this section, we propose a class of strategies in the GCB game $\CB$, called the independently uniform strategies, and we show that it is an approximate equilibrium (and an approximate min-max strategy in the constant-sum case). Note that the independently uniform strategies are also approximate equilibria of the GLB game $\LB$, we analyze that in~Section~\ref{sec:LotteryApproximation}.

We begin by recalling the concept of approximate equilibria (see e.g., \citet{myerson1991game,Nisan07}) in the context of our games: \emph{for any $\varepsilon \ge 0$, an \emph{\mbox{$\varepsilon$-equilibrium}} of the game $\CB$ is any strategy profile $\left(s^{*},t^{*} \right)$ such that \mbox{$\Pi_A(s,t^*) \!\le\! \Pi_A(s^*,t^*)\!+\!\varepsilon$} and \mbox{$\Pi_B(s^*,t)\! \le\! \Pi_B(s^*,t^*)\! +\! \varepsilon$} for any strategy $s$ and $t$ of players A and B. } Replacing $\Pi_A$ and $\Pi_B$ by $\Pi_A^{\zeta}$ and $\Pi_B^{\zeta}$, we have the definition of \mbox{$\varepsilon$-equilibria} of the game $\LB(\zeta)$. Hereinafter, we use the generic term \emph{approximate~equilibrium} whenever the approximation error $\varepsilon$ need not be emphasized.

%
\subsection{The Independently Uniform Strategies}
\label{sec:IU_Strategy}

Given a game $\CB$ (or $\LB$), let us consider the corresponding Equation \eqref{eq:Equagamma} and set $\Sn$. For any $\gam  \in \Sn$, we define in Definition~\ref{def:IU_strategy} a mixed strategy via an algorithm, called Algorithm~\ref{alg:IU_strategy}. We term this strategy as the \emph{independently uniform} strategy (or $\IU$ strategy), parameterized by $\gam $. Intuitively, this strategy is constructed by a simple procedure: players start by \emph{independently} drawing $n$ numbers from the \emph{uniform-type} distributions defined in Definition~\ref{def:UnifromDistributions}, then they re-scale these numbers to guarantee the budget constraints.

\begin{definition}[The independently uniform strategy]
\label{def:IU_strategy}
In any game $\CB$ (or $\LB$), for any $\gam  \in \Sn$ and any player~$\play \in \{ A, B\}$, $\IU_\play$  is the \textbf{mixed} strategy of player~$\play$ where her allocation $\boldsymbol{x}^\play$ is randomly generated from Algorithm~\ref{alg:IU_strategy}.
\begin{figure}[ht]
    \centering
    \begin{minipage}{0.8\textwidth}
        \begin{algorithm}[H]
            \DontPrintSemicolon
        \SetAlgoVlined
        \SetInd{6pt}{6pt}
        \SetNlSkip{6pt}
            \SetCustomAlgoRuledWidth{0.5\textwidth}
            \KwIn{$n\in\mathbb{N}$, $w^A_i,w^B_i \in [{\wmin},{\wmax}], \forall i \in [n]$, budgets $X_A, X_B$, $\gam  \in \Sn$} 
             \KwOut{$\boldsymbol{x}^A, \boldsymbol{x}^B \in\mathbb{R}^n_{\ge 0}$}
             
            Draw $a_i \sim \FA, b_i \sim \FB, \forall i \in [n]$ independently \;
            \eIf{${\sum\nolimits_{j \in [n]} {{a_j}} = 0}$}{
            $x^A_i:=0, \forall i \in [n]$
            }{${x}^A_i:= \frac{a_i}{\sum\nolimits_{j \in [n]}{a_j} } X_A, \forall i \in [n] $} 
            \eIf{${\sum\nolimits_{j \in [n]} {{b_j}} = 0}$}{
            $x^B_i:=0, \forall i \in [n]$
            }{${x}^B_i:= \frac{b_i}{\sum\nolimits_{j \in [n]}{b_j} } X_B, \forall i \in [n] $} 
            \caption{$\IU$ strategy-generation algorithm.}\label{alg:IU_strategy}
        \end{algorithm}    
    \end{minipage}
\end{figure}
\end{definition}

Henceforth, we use the term $\IU$ strategy to denote the strategy profile $( {\IU_A}, {\IU_B} )$. We also simply use the notation $\IU$ in some places to commonly address either $\IU_A$ or $\IU_B$ strategy in case the name of the player need not be specified. Observe that for any player~$\play \in \{A,B\}$, any output $\boldsymbol{x}^\play$ from Algorithm~\ref{alg:IU_strategy} is an $n$-tuple that satisfies her budget constraint. In other words, $\IU_\play$ is a mixed strategy that is implicitly defined by Algorithm~\ref{alg:IU_strategy} and each run of this algorithm outputs a feasible pure strategy sampled from~$\IU_\play$. Note that the marginals of the $\IU$ strategy are \emph{not} the uniform-type distributions $\FA, \FB, i \in [n]$ defined in~Section~\ref{sec:preliminaries}. In terms of computational complexity, Algorithm~\ref{alg:IU_strategy} terminates in $\mO(n)$ time. Below we discuss the specificity of the outputs of Algorithm~\ref{alg:IU_strategy} in the cases where $\sum_{j \in [n]} a_j =0$ or $\sum_{j \in [n]} b_j =0$.

\begin{remark}
    If $\sum_{j \in [n]} a_j =0$ or $\sum_{j \in [n]} b_j =0$, the $\IU_\play$ strategy allocates zero resource to all battlefields for the corresponding player (line~$3$ and line~$7$ of Algorithm~\ref{alg:IU_strategy}). It may seem more natural that, if \mbox{$\sum_{j \in [n]} a_j =0$}, Player A allocates equally on all battlefields, i.e., set \mbox{$x^A_i:= X_A/n, \forall i \in [n]$} in line $3$ of Algorithm~\ref{alg:IU_strategy} (and similarly for Player B). In reality though, these assignments can be chosen to be any arbitrary $n$-tuple $\boldsymbol{x}^\play$ in $\mathbb{R}^n_{\ge 0}$ as long as \mbox{$\sum \nolimits_{i\in[n]}{x^\play_i} \le X_\play$} without affecting the results in our work. This comes from the fact that in most cases, the conditions in line~$2$ and~$6$ hold with probability zero. They can happen with a positive probability only when one player is the ``weak player" and the other is the ``strong player" on all of the battlefields (i.e., either \mbox{$\Ostar = \emptyset$ or $\Ostar = [n]$)}, e.g., in a constant-sum game $\CCB$. Yet, even in this case, this probability decreases exponentially as the number of battlefields increases (see~\eqref{eq:Prop_all_zero} in~\ref{sec:Appen_Proof_TheoBlotto}). The asymptotic order of the approximation error in all of our results is larger than this probability; therefore, it does not matter which assignment we choose in lines~$3$ and $7$ of Algorithm~\ref{alg:IU_strategy}. Here, we choose to assign $x^A_i = 0, \forall i$ and $x^B_i = 0, \forall i$ to ease the notation in the proofs of the results in the following sections; in particular, it avoids creating a discontinuity outside~$0$ in the CDF of the effective allocation in each battlefield (see also Lemma~\ref{lem:continuity_Ani_and_Bni} in \ref{sec:Appen_Proof_TheoBlotto}). 
    
%
\end{remark}
%
%
%
%
\subsection{Approximate Equilibria of the Generalized Colonel Blotto Game $\CB$}
\label{sec:Approx_NonConstantSum}

We now present our main result stating that the $\IU$~strategy is an approximate equilibrium with an error that is only a negligible fraction of the maximum payoffs that the players can achieve, quickly decreasing as $n$ increases. In the following results, note that since we focus on the setting of games with a large number of battlefields, we now focus on characterizing the approximation error according to $n$ and treat other parameters of the $\CB$ games, including $X_A, X_B, \wmin, \wmax$ and $\alpha$, as constants (but not the values $w^\play_i, v^\play_i, \forall i \in [n], \play \in \{A,B\}$). Note also that we often use the notation $\tilde{\mO}$ that is a variant of the big-$\mO$ notation where logarithmic factors are ignored.\footnote{See~\ref{sec:appen_preliminary} for formal definitions of these notations.}
\begin{restatable}{theorem}{TheoMainBlotto}
\label{TheoMainBlotto}
\leavevmode
Consider GCB games satisfying \ref{eq:A0}, we have the following~results:
\begin{itemize}
    \item [(i)] In any game $\CB$, there exists a positive number \mbox{$\varepsilon = \tilde{\mO} (n^{-1/2})$} such that for any \mbox{$\gam  \in \Sn$}, the following inequalities~hold for any pure strategy $\boldsymbol{x}^A$ and $\boldsymbol{x}^B$ of players A and~B:
        \begin{align}
        & \Pi_A(\boldsymbol{x}^A,{\IU_B}) \le \Pi_A({\IU_A},{\IU_B}) + \varepsilon W_A,\label{eq:MainTheo_A}\\
        & \Pi_B({\IU_A},\boldsymbol{x}^B) \le \Pi_B({\IU_A},{\IU_B}) + \varepsilon W_B. \label{eq:MainTheo_B}
        \end{align}
    \item [(ii)] Given $X_A, X_B>0$ ($X_A \le X_B$), there exists $C^*>0$ such that for any $\varepsilon\! \in (0, 1]$, in any game $\CB$ with $n\!\ge C^* \varepsilon^{-2} \logep$, \eqref{eq:MainTheo_A}~and \eqref{eq:MainTheo_B} hold for any \mbox{$\gam  \in \Sn$}, any pure strategy~\mbox{$\boldsymbol{x}^A$}, $\boldsymbol{x}^B$ of players A~and~B respectively.
\end{itemize}
\end{restatable}
A proof of this theorem is presented in~\ref{sec:Appen_Proof_TheoBlotto}. First, note that although the results in Theorem~\ref{TheoMainBlotto} are stated under \ref{eq:A0}, it does not restrict their scope of application since this, as discussed, is a very mild assumption. Moreover, the two results of Theorem \ref{TheoMainBlotto} are two equivalent statements. In the remainders of this section, we give some interpretations of these two~results.

Result $(i)$ of Theorem~\ref{TheoMainBlotto} states that, given a game $\CB$, there exists no unilateral deviation from an $\IU$ strategy that can profit any player $\play \in \{A,B\}$ more than a small portion of her maximum payoff~$W_\play$. As a direct corollary, the $\IU$ strategy is an approximate equilibrium of the game $\CB$ with a bounded approximation error (depending on $n$); this is formally stated as~follows:
\begin{corollary}\textbf{(Approximate equilibrium of the $\CB$ game)}
\label{Corol:Blotto_Approx_Equi}
In any game $\CB$ satisfying \ref{eq:A0}, there exists a positive number \mbox{$\varepsilon = \tilde{\mO} \left(n^{-1/2}\right)$} such that for any $\gam  \in \Sn$, the $\IU$ strategy is an \mbox{$\varepsilon W$-equilibrium} where $W:= \max\{W_A, W_B\}$.
\end{corollary}

We now discuss the interpretation of this corollary and Result~$(i)$ of Theorem~\ref{TheoMainBlotto} in a sequence of games $\CB$ (satisfying \ref{eq:A0}) that have larger and larger numbers of battlefields (i.e., $n$ increases) to see the relation between the games' parameters and the approximation error of the $\IU$ strategies. First, observe that as $n$ increases, the error $\varepsilon W$ of this approximate equilibrium \emph{might not} decrease to $0$. This is due to the fact that although $\varepsilon$ decreases as $n$ increases, $W$ does not. Indeed, recall that $W_A$ and $W_B$ are the total values that players A and B assign on all battlefields; as the number of battlefield increases, by \ref{eq:A0}, $W_A$, $W_B$ and $W = \max\{W_A, W_B\}$ inevitably increase (at a rate $\mO (n)$). This, however, does \emph{not} question the significance of the $\IU$ strategies as a good approximate equilibrium (with $n$ is large enough). In fact, it is not meaningful to compare the magnitude of approximation errors of $\IU$ strategies in two $\CB$ games having different sizes. Instead, we should compare the ratio between the approximation error and $W$---an upper-bound on players' payoffs---which is relative to the scale of the considered games. This ratio is exactly $\varepsilon$ (we call this \emph{the level of error}). Importantly, as we consider $\CB$ games with larger and larger $n$, $\varepsilon$ indeed quickly tends to 0; the bound $\tilde{\mO} (n^{-1/2})$ indicates the order of this relation showing how fast $\varepsilon$ decreases as~$n$~increases. 

Said differently, the result of Corollary~\ref{Corol:Blotto_Approx_Equi} states that, by playing the $\IU$ strategies, a player guarantees to not lose more than a small fraction $\varepsilon$ of her maximum total payoff $W$, which is indeed the meaningful statement in our setting. An alternative formulation would be to consider a sequence of $\CB$ games (where $n$ increases) whose the battlefields values are re-scaled such that $W =1$ in all games. In this case, Corollary~\ref{Corol:Blotto_Approx_Equi} would indicate that the $\IU$ strategy is an $\varepsilon$-equilibrium of $\CB$ with $\varepsilon \rightarrow 0$ when $n \rightarrow \infty$. We prefer our formulation without re-scaling as it better emphasizes the relationship to the battlefield values. 

We also note that the upper-bound on $\varepsilon$ also depends on other parameters of the game $\CB$, including $X_A, X_B, \wmin$ and $\wmax$ (given by \ref{eq:A0}).\footnote{This dependency is implicitly presented in the asymptotic notation $\tilde{\mO}$ in Result~$(i)$ and the constant $C^*$ in Result~$(ii)$ of Theorem~\ref{TheoMainBlotto}.} We can extract from the proof of Theorem~\ref{TheoMainBlotto} (see \ref{sec:Appen_Proof_TheoBlotto}) that as $\wmax / \wmin$ and/or $X_B/X_A$ increases, $\varepsilon$ also increases, i.e., for games with higher heterogeneity of the battlefields values and/or higher asymmetry in players' budgets, the $\IU$ strategy yields higher errors. Additionally, to keep the generality, Result~$(i)$ of Theorem~\ref{TheoMainBlotto} is presented such that the approximation error $\varepsilon$ is commonly addressed for any $\IU$ strategy with any $\gam  \in \Sn$. For each specific solution $\gam $ of Equation \eqref{eq:Equagamma} (implying $\LdaA$ and $\LdaB$), the corresponding $\IU$ strategy is an approximate equilibrium of $\CB$ with an approximation error that is at most (and it might be strictly smaller than)~$\varepsilon$.

On the other hand, Result~$(ii)$ of Theorem~\ref{TheoMainBlotto} indicates the number of battlefields that a GCB game needs to contain---relative to players' budgets and bounds on battlefields' values---in order to guarantee a desired level of the approximation error by using the $\IU$ strategy as an approximate equilibrium. Hence, in practical situations involving large instances of the Colonel Blotto game, the $\IU$ strategy (simply and efficiently constructed by Algorithm~\ref{alg:IU_strategy}) can be used as a safe replacement for any Nash equilibrium whose construction may be unknown or too complicated. Now, let us introduce an important notation:
\begin{definition}
\label{def:AnBnDistribution}
    Corresponding to the players' allocations toward each battlefield $i \in [n]$, let $\FAn$ and $\FBn$ denote the univariate marginal distributions  of the $\IU_A$ and $\IU_B$ strategies (see~\eqref{eq:A^n_Def} and \eqref{eq:B^n_Def} in~\ref{sec:Appen_Proof_TheoBlotto} for a more explicit~formulation of $\FAn$ and $\FBn$).
\end{definition}
Intuitively, Result~$(ii)$ can be proved by showing the following two results: \emph{(a)} when Player B's allocation to the battlefield $i \in [n]$ follows $\FB$, the best response of Player A is to play such that her allocation to $i$ follows the distribution $\FA$ (and vice versa); \emph{(b)} as $n$---the number of battlefields---increases, $\FAn$ and $\FBn$ uniformly converge toward the distributions $\FA$ and $\FB$, i.e., the marginal distributions of the $\IU$ strategy approximate the distributions $\FA$ and $\FB$. This convergence can be proved by applying concentration inequalities on the random variables $\sum \nolimits_{j \in [n]} A^*_j$ and $\sum \nolimits_{j \in [n]} B^*_j$ (see~Lemma~\ref{lem:convergence} in~\ref{sec:Appen_Proof_TheoBlotto}); moreover, the relation between $\varepsilon$ and $n$ in the results of Theorem~\ref{TheoMainBlotto} depends directly on the rate of this convergence. In this work, we use the Hoeffding's inequality (\citet{hoeffding1963probability}) that yields a better convergence rate than working with other types of concentration inequalities (e.g., Chebyshev's inequality). To complete the proof of Result~$(ii)$, we finally show that as $n$ increases, when player $-\play \in \{A,B\}$ plays the $\IU_{-\play}$ strategy, the $\IU_{\play}$'s payoff of player~$\play$ converges toward her best-response payoff. Note that these payoffs can be written as expectations with respect to different measures (see \eqref{eq:Pi_A_pure}, \eqref{eq: Pi_A_IU} and Lemma~\ref{lem:SufCon} in~\ref{sec:Appen_Proof_TheoBlotto}). To prove the convergence of payoffs, we use a variant of the portmanteau theorem (see~Lemma~\ref{lem:portmanteau} in \ref{sec:Appen_Proof_TheoBlotto}) regarding the equivalent definitions of the weak convergence of a sequence of~measures. Note importantly that a direct application of the portmanteau theorem leads to a slow convergence rate (notably, \eqref{eq:MainTheo_A} and~\eqref{eq:MainTheo_B} only hold when $n = \Omega(\varepsilon^{-4}))$. This is due to the fact that the players' payoffs involve the bounded Lipschitz functions $\FA$ and $\FB$ and that these functions depend on $n$, particularly, their Lipschitz constants (that are either $\LdaA/v^A_i$ or $\LdaB/v^B_i$) increase as $n$ increases. In order to obtain the convergence rate as indicated in Theorem~\ref{TheoMainBlotto}, we exploit the special relation between $\FAn$ and $\FA$, and between $\FBn$ and $\FB$; then we apply a telescoping-sum trick allowing us to avoid the need of using the Lipschitz properties (for more details, see the proof of Lemma~\ref{lem:portmanteau} in~\ref{sec:appen_proof_lem_portmanteau}).

%
%
%
%
%
\subsection{Approximate Equilibria of the Constant-sum Colonel Blotto Game}
\label{sec:Approx_ConstantSum}

In this section, we discuss the constant-sum game variant: the game $\CCB$ defined in Definition~\ref{def:constantsumGame}. As an instance of the GCB game, the game $\CCB$ satisfies all results presented in Sections~\ref{sec:IU_Strategy} and~\ref{sec:Approx_NonConstantSum}. Additionally, we show that any $\IU$ strategy is an approximate max-min strategy of the game~$\CCB$. 

In any game $\CCB$, Equation \eqref{eq:Equagamma} has a unique solution \mbox{${\gamma}^*\!=\! {X_B}/{X_A}\!\ge\!1$} which, in turn, uniquely induces \mbox{$\LdaA\!=\!{1}/(2X_B)$} and \mbox{$\LdaB= {X_A}/{(2{X_B}^2)}$}. Moreover, in the game $\CCB$, observe that \mbox{${v^A_i}/{v^B_i}=1 \le X_B/X_A = \gam _A/ \gam _B, \forall i \in [n]$}; therefore, we have \mbox{$\Ostar = \emptyset$}. Intuitively, Player A is the ``weak player'' (and B the ``strong player'') in \emph{all} battlefields. Now, note that in the game $\CCB$, $W:= \max \{W_A,W_B \} = W_A = W_B$, we apply Theorem~\ref{TheoMainBlotto} and obtain the following result.
\begin{corollary}
\label{corol:constant_sum_Equi}
In any game $\CCB$ satisfying \ref{eq:A0}, there exists a positive number \mbox{$\varepsilon \le \tilde{\mO} (n^{-1/2})$} such that the ${\rm IU} ^{{\gamma}^*}$ strategy is an \mbox{$\varepsilon {W}$-equilibrium} with ${\gamma}^* \in \Sn =\left\{ {X_B}/{X_A} \right\}$.
\end{corollary}
Note that if an equilibrium exists in $\CCB$ (note that its existence still remains an open question), then the set of OUDs is unique (see e.g., Corollary~1 of \citet{kovenock2020generalizations}) and they correspond to the distributions $F_{\AW}$ and $F_{\BS}$ for any \mbox{$i \in [n]$} as defined in~\eqref{Aw} and~\eqref{Bs} (where $\LdaA$ and $\LdaB$ are respectively replaced by the values ${1}/(2X_B)$ and ${X_A}/{(2{X_B}^2)}$). The marginals of the $\IU$ strategy with $\gam  = X_B/X_A$ converge toward these unique~OUDs.

Finally, we deduce that the $\IU$ strategy is an approximate $\max$-$\min$ strategy of the game $\CCB$; this is formally stated as follows.
\begin{corollary}
\label{corol:max_minConstantSum}
 In any game $\CCB$ satisfying \ref{eq:A0}, there exists a positive number $\varepsilon \le \tilde{O}(n^{-1/2})$ such that the following inequalities hold for ${\gamma}^* \in \Sn =\left\{ {X_B}/{X_A} \right\}$ and any strategy $\tilde{s}$ and $\tilde{t}$ of players A and B:
    \begin{align}
        & \min \limits_{t}{\Pi_A(\tilde{s}, t)} \le \min \limits_{t} {\Pi_A(\IU_A, t)} + \varepsilon {W},\\
        & \min \limits_{s}{\Pi_B(s, \tilde{t})} \le \min \limits_{s} {\Pi_B(s,\IU_B)} + \varepsilon {W}. 
    \end{align}
\end{corollary}
Intuitively, if player $\play \in \{A,B \}$ plays the $\IU_{\play}$ strategy, she guarantees a near-optimal payoff even in the worst-case scenario when her opponent $-\play$ plays strategies that minimize $\play$'s payoff (no matters how it affects $-\play$'s payoff). The proofs of Corollary~\ref{corol:constant_sum_Equi} and Corollary~\ref{corol:max_minConstantSum} can be trivially deduced by applying specifically Theorem~\ref{TheoMainBlotto} to the constant-sum Colonel Blotto games and thus are omitted in this~work.

%
%
%
%
%
%

\section{Approximate Equilibria of the Generalized Lottery Blotto Game}
\label{sec:LotteryApproximation}

In this section, we present the results regarding the $\IU$ strategy in the GLB game. In Section~\ref{sec:Approx_Lottery_LB_n}, we analyze the game $\LB(\zeta)$ with an arbitrary pair of CSFs $\zeta=(\zeta_A,\zeta_B)$ and show that the $\IU$ strategy is an approximate equilibrium of $\LB(\zeta)$ with an error depending on the number of battlefields as well as the dissimilarity between $\zeta_A$ and $\beta_A$ (and between $\zeta_B$ and $\beta_B$). In Section~\ref{sec:Approx_Lottery_ratio}, we illustrate this result in two particular instances, the games $\LB(\mu^R)$ and $\LB(\nu^R)$, belonging to the class of Lottery Blotto games with ratio-form CSFs. We characterize the approximation error of the $\IU$ strategy in these games according to $n$ and the parameter $R$ of these~CSFs.

%
%
%
%
\subsection{Approximate Equilibria of Generalized Lottery Blotto games $\LB(\zeta)$}
\label{sec:Approx_Lottery_LB_n}

We start by defining a parameter that expresses the dissimilarity between a given pair of CSFs \mbox{$\zeta=(\zeta_A,\zeta_B)$} and the Blotto functions $\beta_A, \beta_B$ (defined in \eqref{eq:betafunction}). First, note that for any $n$ and $i \in [n]$, the random variables $A^*_i, B^*_i$ (corresponding to the distributions defined in Definition~\ref{def:UnifromDistributions}) and $A^n_i, B^n_i$ (corresponding to the distributions defined in Definition~\ref{def:AnBnDistribution}) are all upper-bounded by $2 X_B$.\footnote{By definition, $A^n_i, B^n_i$ are bounded by $X_A$, $X_B$ and thus by $2X_B$; for a proof that $A^*_i, B^*_i$ admit the same upper-bound, see Lemma~\ref{lem:Preliminary} in~\ref{sec:appen_preliminary}.} Now, given any $\varepsilon>0$, for any pair of CSFs \mbox{$\zeta=(\zeta_A,\zeta_B)$} and any $x^*\in [0, 2 X_B]$ and $y^* \in [0, 2X_B]$ (i.e., any number that can be sampled from $\FA,\FB, \FAn$ or $\FBn$), we introduce the following~sets:
\begin{align}
    & \X(y^*,\varepsilon):= \left\{x \in [0,2 X_B]: |\zeta_A(x,y^*) - \beta_A(x,y^*)| \ge \varepsilon \right\}, \label{eq:Xset} \\
    & \Y(x^*,\varepsilon):= \left\{y \in [0,2 X_B]: |\zeta_B(x^*,y) - \beta_B(x^*,y)| \ge \varepsilon \right\}.\label{eq:Yset}
\end{align}


%
%
%
\begin{definition}
\label{def:delta}
For any pair of CSFs $\zeta=(\zeta_A, \zeta_B)$, $\varepsilon>0$ and $\gam  \in \Sn$, we define the following~set\footnote{Note that $\FA, \FB$ are continuous, bounded functions on $[0,2X_B]$; therefore, they attain a maximum on this interval.}
%
%
\begin{equation*}
\Del := \left\{ \delta\in [0,1] : \max \limits_{i \in [n]}   \max \limits_{y^* \in [0,2X_B]}{\int_{\X(y^*,\varepsilon)} { \!\!\!\!\!\de \FA(x)}} \le \delta, \;\; \textrm{ and } \;\;
     \max \limits_{i \in [n]}  \max \limits_{x^* \in [0,2X_B]}{\int_{\Y(x^*,\varepsilon)} { \!\!\!\!\!\de \FB(y)}}  \le \delta
\right\}.    
\end{equation*}
\end{definition}
Intuitively, the set $\Del$ contains all numbers $\delta \in [0,1]$ such that for any allocation $y^*$ of Player B toward an arbitrary battlefield $i$, if Player A draws an allocation $x$ from the distribution $\FA$, it only happens with probability at most $\delta$ that the value of the CSF $\zeta_A$ at $(x, y^*)$ is significantly different (i.e., $\varepsilon$-away) from that of the Blotto function $\beta_A$; and we have a similar statement for the distribution $\FB$ of Player B and any allocation $x^*$ of Player~A. Note that the set $\Del$ depends on $\FA$ and $\FB$, thus it depends on~$\gam $. We can trivially see that $\Del$ is an interval with the form $[\delta_0, 1]$ since if $\delta_0 \in \Del$ then $\delta \in \Del$ for any $\delta \ge \delta_0$.

Based on the convergence of $\FAn$ and $\FBn$ toward $\FA$ and $\FB$ (see the details in Lemma~\ref{lem:convergence} in~\ref{sec:Appen_Proof_TheoBlotto}), we can prove the following lemma (a formal proof is given~in~\ref{sec:appen_proof_lem_delta}):
\begin{restatable}{lemma}{deltalemma}
\label{lem:deltalemma}
   Given $\wmin, \wmax, X_A, X_B >0$ ($\wmin \le \wmax$, $X_A \le X_B$), there exists a constant $L_0>0$ such that for any $\varepsilon \in (0,1]$ and $n \ge L_0 \varepsilon^{-2} \logep$, for any game $\LB(\zeta)$ satisfying \ref{eq:A0}, any $\gamma^* \in \Sn$ , $\delta \in \Del$ and $i \in [n]$, we have:
    \begin{equation}
       \max \left\{ \sup_{y^* \in [0,2X_B]}{\int \nolimits_{\X(y^*,\varepsilon)} { \de \FAn(x)}}, \sup_{x^* \in [0,2X_B]}{\int \nolimits_{\Y(x^*,\varepsilon)} { \de \FBn(y)}} \right\}  \le \delta  + \varepsilon. \label{eq:lemma_Fn_LB}
    \end{equation}
\end{restatable}
\noindent Intuitively, this lemma provides an upper-bound for the probability of the value of the CSFs $\zeta$ being $\varepsilon$-away from the Blotto functions when Player A (resp. Player B) plays such that her allocation to battlefields $i$ follows $\FAn$ (resp. $\FBn$), i.e., when she plays the $\IU$ strategy.

Based on the definition of~$\Del$ and Lemma~\ref{lem:deltalemma}, we can now show the following result regarding the $\IU$ strategy in the GLB game.
\begin{restatable}{theorem}{LotteTheo}
    \label{theo:Lottery_generic_approx}
    
   Consider GLB games satisfying \ref{eq:A0}, we have the following~results:
    \begin{itemize}
        \item[(i)] In any game $\LB(\zeta)$, there exists a positive number \mbox{$\varepsilon \le \tilde{\mO} (n^{-1/2})$} such that for any $\gam  \in \Sn$ and $\delta \in \Del$, the following inequalities hold for any pure strategy $\boldsymbol{x}^A$ and $\boldsymbol{x}^B$ of players A and B:
        \begin{align}
            & \Pi^{\zeta}_A(\boldsymbol{x}^A,{\IU_B}) \le \Pi^{\zeta}_A({\IU_A},{\IU_B}) + \left(8\delta + 13{\varepsilon} \right) W_A,\label{eq:lottery_theo_A}\\
            & \Pi^{\zeta}_B({\IU_A},\boldsymbol{x}^B) \le \Pi^{\zeta}_B({\IU_A},{\IU_B})  + \left(8\delta + 13{\varepsilon} \right) W_B. \label{eq:lottery_theo_B}
        \end{align}
        \item[(ii)] Given $ X_A, X_B >0$ ($X_A \le X_B$), there exists $L^*>0$ such that for any $\varepsilon\! \in (0, 1]$, in any game $\LB(\zeta)$ where $ n \ge L^* \varepsilon^{-2} \logep$, \eqref{eq:lottery_theo_A} and \eqref{eq:lottery_theo_B} hold for any $\gam  \in \Sn$, $\delta \in \Del$ and any pure strategy $\boldsymbol{x}^A, \boldsymbol{x}^B$ of players A and~B.
    \end{itemize} 
\end{restatable}
The proof of this theorem is given in \ref{sec:appen_proof_theoLottery}. The main idea to prove these results is that we can approximate the players' payoffs in the game $\LB(\zeta)$ when they play the $\IU$ strategies by that in the corresponding game $\CB$ (the difference between these payoffs is controlled by the parameter $\delta \in \Del$); and then use the results from Section~\ref{sec:ApproximateBlotto} for the game $\CB$ (involving the error $\varepsilon$) to prove~\eqref{eq:lottery_theo_A} and~\eqref{eq:lottery_theo_B}. The coefficients (8 and 13) in front of the parameters $\delta$ and $\varepsilon$ come from the application of several triangle inequalities to connect these approximate results. Note that if the CSFs $\zeta_A$ and $\zeta_B$ are Lipschitz continuous on $[0,2X_B] \times [0,2X_B]$, we can avoid the need to approximate several terms involved in the analysis of using the $\IU$ strategy in the game $\LB(\zeta)$ via the corresponding terms in the game $\CB$; thus, we can improve the results in Theorem~\ref{theo:Lottery_generic_approx} to obtain an approximation error of $2\delta + 5 \varepsilon$ instead of $8\delta + 13 \varepsilon $ (see Remark~\ref{remark_conti_CSF} in~\ref{sec:appen_remark_conti_CSF} for more details). Here, to keep the generality, we do not include the continuity assumption of the CSFs in Theorem~\ref{theo:Lottery_generic_approx} (recall that our definition of a CSF allows for~discontinuity).

Intuitively, Result $(i)$ of Theorem~\ref{theo:Lottery_generic_approx} determines the order of the approximation error while using $\IU$ in any given game $\LB(\zeta)$. Straightforwardly, we can deduce that the $\IU$ strategy is an approximate equilibrium of the game $\LB(\zeta)$, formally stated as~follows. 
\begin{corollary}\textbf{(Approximate equilibria of the $\LB$ game)}
\label{corol:Lottery_LB_n}
   In any game $\LB(\zeta)$ satisfying \ref{eq:A0}, there exists a positive number \mbox{$\varepsilon \le \tilde{\mO} (n^{-1/2})$} such that for any $\gam  \in \Sn$ and $\delta \in \Del$, the $\IU$ strategy is an $\left(8\delta + 13{\varepsilon} \right)W$-equilibrium where \mbox{$W:= \max\{W_A,W_B \}$}.
\end{corollary}
We observe that the error bound in Theorem~\ref{theo:Lottery_generic_approx} (and in Corollary~\ref{corol:Lottery_LB_n}) is valid for any $\delta$ of the set $\Del$. Naturally, it is the tightest for  $\delta_0 = \min \{\delta: \delta \in \Del\}$; but this quantity is not always easy to compute; for example, in the Lottery Blotto games with the power and logit form CSFs (i.e., $\LB(\mu^R)$ and $\LB(\nu^R)$). Still, in Section~\ref{sec:Approx_Lottery_ratio}, we show that there exists an element of $\Delmu$ and $\Delnu$ that is negligibly small, given appropriate parameter's configurations of the games; in other words, we can still obtain a good error's upper-bound for the $\IU$ strategy in these games. 
Note that, on the other hand, the GCB game $\CB$ can be considered as an instance of the game $\LB(\zeta)$ where the CSFs are $\zeta_A= \beta_A$ and $\zeta_B = \beta_B$; therefore, it also satisfies Theorem~\ref{theo:Lottery_generic_approx}. In $\CB$, we trivially have \mbox{$\X(y^*, \varepsilon) = \Y(x^*, \varepsilon) = \emptyset$} for any $x^*, y^*$; thus $\Del = [0,1]$ for any \mbox{$\varepsilon>0$ and $\min \{\delta: \delta \in \Del\} =0$}.\footnote{Note also that for the case of the game $\CB$, the left-hand side in~\eqref{eq:lemma_Fn_LB} equals zero for any $n$ and $i \in [n]$.} This is consistent with results obtained in~Theorem~\ref{TheoMainBlotto} in Section~\ref{sec:ApproximateBlotto}. 

In Theorem~\ref{theo:Lottery_generic_approx}, Result~$(ii)$ is an equivalent statement of Result~$(i)$. It indicates the number of battlefields needed to guarantee a certain level of approximation error when using the $\IU$ strategy in the game $\LB(\zeta)$. For instance, to obtain an approximate equilibrium of the game $\LB(\zeta)$ where the level of error is less than a certain number $\bar \varepsilon$, one needs $\varepsilon \le \bar{\varepsilon}$ such that we can find a $\delta \in \Del$ satisfying $8 \delta + 13 \varepsilon \le \bar{\varepsilon}$; from these parameters, by Result~$(ii)$, one can deduce the sufficient number of battlefields needed for an $\LB$ game to yield that desired level of~error.

Finally, in the constant-sum variant of the GLB game, denoted by $\CLB(\zeta)$ (i.e., when \mbox{$w^A_i = w^B_i$}, \mbox{$\forall i \in [n]$}), we can easily deduce from Theorem~\ref{theo:Lottery_generic_approx} that the $\IU$ strategy is also an approximate max-min~strategy: 
\begin{corollary}
\label{corol:max_minLB}
    In any game $\CLB(\zeta)$ satisfying \ref{eq:A0}, there exists \mbox{$\varepsilon \le \tilde{\mO} (n^{-1/2})$} such that for any \mbox{$\gam  \in \Sn = \{X_B/ X_A\}$} and $\delta \in \Del$, the following inequalities hold for any strategy $\tilde{s}$ and $\tilde{t}$ of players A and~B:\footnote{Recall that in the constant-sum variant, $W:= \max\{W_A, W_B\} = W_A =W_B$.}
    \begin{align*}
        & \min \limits_{t}{\Pi_A^\zeta(\tilde{s}, t)} \le \min \limits_{t} {\Pi_A^\zeta(\IU_A, t)} + (8 \delta + 13\varepsilon) {W},  \\
        & \min \limits_{s}{\Pi_B^\zeta(s, \tilde{t})} \le \min \limits_{s} {\Pi_B^\zeta(s,\IU_B)} + (8 \delta + 13\varepsilon){W}.  
    \end{align*}
\end{corollary}
%

\subsection{Approximate Equilibria of the Lottery Blotto Games with Ratio-form CSFs}
\label{sec:Approx_Lottery_ratio}

We now consider the games $\LB(\mu^R)$ and $\LB(\nu^R)$ where $\mu^R$ and $\nu^R$ are pairs of CSFs that are defined in Table~\ref{table:CSF}. Note again that for those CSFs, we do not consider the degenerate cases where $\alpha=0$ or $\alpha =1$ in which trivial equilibria~exist in the corresponding GLB games. The games $\LB(\mu^R)$ and $\LB(\nu^R)$ are instances of the game $\LB(\zeta)$ studied in Section~\ref{sec:Approx_Lottery_LB_n}; therefore, by Theorem~\ref{theo:Lottery_generic_approx} (and Corollary~\ref{corol:Lottery_LB_n}), the $\IU$ strategy is also an approximate equilibrium of them. In this section, we focus on characterizing the approximation error of the $\IU$ strategy in these games according to $n$ (the number of battlefields) and $R$ (the corresponding parameter of the CSFs). We will show that this error quickly tends to zero as $n$ and $R$ increase under appropriate conditions. To do this, we first notice that although it is non-trivial to analyze the closed form of the sets $\Delmu$ and $\Delnu$ and find their minimum, we can find small elements of theses~sets. 
\begin{restatable}{lemma}{lemmamunu}
\label{lem:delta_mu_nu}
Fix $n \ge 2 $, $R>0$ and $\alpha \in (0,1)$, for any $\varepsilon <
\min \{\alpha, 1 - \alpha\}$, we have:\footnote{The asymptotic notations are taken w.r.t. when $\varepsilon \rightarrow 0$.} 
\begin{itemize}
    \item[$(i)$] In any game $\LB(\mu^R)$ satisfying \ref{eq:A0} and having $\alpha$ as the tie-breaking parameter, there exists \mbox{$\delta_{\mu}\! =\! \min\{1, \mO\! \left(n(\varepsilon^{-1/R} \!- \!1)  \right) \}$} such that \mbox{$ \delta_{\mu} \in \Delmu$} for any $\gam  \in \Sn$.
    \item[($ii$)] In any game $\LB(\nu^R)$ satisfying \ref{eq:A0} and having $\alpha$ as the tie-breaking parameter, there exists \mbox{$\delta_{\nu} \!=\! \min \{1,\mO\!\left(n R^{-1} \ln(\varepsilon^{- 1}) \right)\}$} such that \mbox{$ \delta_{\nu} \in \Delnu$} for any $\gam  \in \Sn$.
\end{itemize}
\end{restatable}

The proof of Lemma~\ref{lem:delta_mu_nu} is given in~\ref{sec:appen_proof_ratio-form}. Note that for the sake of generality, the parameters $\delta_{\mu}$ and $\delta_{\nu}$ are indicated in this lemma in such a way that they do not depend on~$\gam $, but for each $\gam  \in \Sn$, we can find smaller elements of the corresponding sets $\Delmu$ and $\Delnu$. More importantly, for a fixed $n$, the numbers $\delta_\mu$ and $\delta_\nu$ decrease as $R$ increases; but $\delta_\mu$ and $\delta_\nu$ increase as $\varepsilon$ decreases. While the lemma is valid for any parameter values, since $1$ is a trivial element of $\Delmu$ and $\Delnu$, it is useful only if $\delta_\mu,\delta_\nu < 1$;  this is guaranteed whenever $R \ge \mO \left( n \ln(\varepsilon^{-1}) \right)$. Note finally that the condition $\varepsilon < \min\{\alpha, 1- \alpha\}$ in the statement of Lemma~\ref{lem:delta_mu_nu} does not limit its use since our goal is to obtain asymptotic results on the $\IU$ strategy when $\varepsilon$ tends to $0$. Moreover, in the games $\LB(\mu^R)$ and $\LB(\nu^R)$ where $\alpha$ is either very close to $0$ or $1$, one player has a very high advantage and always obtains large gains from all battlefields (where her allocation is strictly positive) while her opponent gains very little regardless of her allocations; therefore, there exists (many) trivial approximate equilibria with small errors. 

Combining the results of Corollary~\ref{corol:Lottery_LB_n} and Lemma~\ref{lem:delta_mu_nu}, we can deduce directly that in any game $\LB(\mu^R)$ (resp. $\LB(\nu^R)$), there exists $\varepsilon \le \tilde{\mO}(n^{-1/2})$ such that for any $\gam  \in \Sn$, the $\IU$ strategy is an $(8\varepsilon + 13\delta_{\mu}) W$-equilibrium (resp. $(8\varepsilon + 13\delta_{\nu}) W$-equilibrium). Next, we look for the asymptotic relation between these error terms and the parameters $n, R$ of the games. First, as $n$ increases, the error level $\varepsilon$ decreases; on the other hand, from Lemma~\ref{lem:delta_mu_nu}, the number $\delta_\mu$ (and $\delta_\nu$) decreases if $R$ increases with a faster rate than $\tilde{\mO}(n)$. However, there is a trade-off between $\varepsilon$ and $\delta_\mu$ (or $\delta_\nu$): as $\varepsilon$ decreases, $\delta_\mu$ (and $\delta_\nu$) increases and vice versa. 
To handle this trade-off between $\delta_\mu$ and $\varepsilon$ (resp. $\delta_\nu$ and $\varepsilon$), we can first find a condition on $n$ that generates a small error $\varepsilon$, and then find a condition on $R$ (with respect to $n$) such that the error $\delta_\mu$ (resp. $\delta_\nu$) is of the same order as $\varepsilon$. Formally, we state the result that the $\IU$ strategy yields an approximate equilibrium of the games $\LB(\mu^R)$ and $\LB(\nu^R)$ with any arbitrary small~error in the next theorem. 
\begin{restatable}{theorem}{theoratioform}\textbf{(Approximate equilibria of Lottery Blotto games with ratio-form CSFs)}
    \label{theoratioform}
    \noindent Given $\wmin, \wmax, X_A, X_B >0$ ($\wmin \!\le \!\wmax$, $X_A \!\le\! X_B$) and $\alpha \in (0,1)$, there exists $\tilde{L}\!>\!0$ such that for any \mbox{$\barep \in \left(0, \min\{\alpha, 1\!-\! \alpha\} \right)$}, in any game $\LB(\mu^R)$ and $\LB(\nu^R)$---satisfying \ref{eq:A0} and having $\alpha$ as the tie-breaking-rule~parameter---where \mbox{$ n \ge \tilde{L} \barep^{-2} \logbar$}, \mbox{$R \ge  \mO \left(\frac{n}{\barep} \ln\left( \frac{1}{\barep}\right) \right)$}, the $\IU$ strategy is an $\barep W$-equilibrium for any \mbox{$\gam  \in \Sn$}.
\end{restatable}
The proof of this theorem is based on Theorem~\ref{theo:Lottery_generic_approx} and Lemma~\ref{lem:delta_mu_nu} (see~\ref{sec:appen_proof_theoratioLB} for more details). Theorem~\ref{theoratioform} involves a double limit in $R$ and $n$. Intuitively, if $n$ and $R$ increase but $R$ increases with a slower rate, then $\varepsilon$ (in Theorem~\ref{theo:Lottery_generic_approx}) decreases but the corresponding $\delta_\mu$ and $\delta_\nu$ (in Lemma~\ref{lem:delta_mu_nu}) do not decrease; thus, the total error is not guaranteed to decrease. Therefore, in order to obtain an approximate equilibrium of the games $\LB(\mu^R)$ or $\LB(\nu^R)$ with a small level of error (i.e., $\tep$ in Theorem~\ref{theoratioform}), one needs that these games have a large number of battlefields and that the parameter $R$ is large enough such that the CSFs $\mu^R_A, \mu^R_B$ (and $\nu^R_A, \nu^R_B$) approximate well the Blotto functions $\beta_A, \beta_B$. This is consistent with the observation that as $R$ increases, $\mu^R_A, \mu^R_B$ (and $\nu^R_A, \nu^R_B$) converge pointwise towards $\beta_A, \beta_B$ respectively.  

\section{Conclusion}
\label{Conclu}

In this work, we consider the Generalized Colonel Blotto game---the most general version of the Colonel Blotto game with heterogeneous battlefields and asymmetric players. While most of (if not all) works in the literature attempt (but do not completely succeed) to construct an exact equilibrium of the Colonel Blotto game by looking for a joint distribution with the uniform-type marginals that satisfies the budget constraints, we take a different angle. We propose a class of strategies called the $\IU$ strategies that is simply constructed by an efficient algorithm; the $\IU$ strategies guarantee the budget constraints but their marginals are not the optimal univariate distributions. Yet, we prove the $\IU$ strategy to be an approximate equilibrium of the GCB games. We also study an extended game called the Generalized Lottery Blotto game and obtain similar results. We characterize the approximate error in our results in terms of the number of battlefields of the games. Our work extends the scope of application of the GCB games and its variants. 

Throughout the paper, we emphasized the dependence of the approximation error on the number of battlefields $n$. Yet, although the dependence on other parameters of the games is not explicitly emphasized, it can be extracted from our analysis and the proofs of the stated results. Note that most results in this paper are obtained under an assumption on the games' parameters, but this assumption holds for all relevant practical applications. It is also interesting to note that although the notion of approximate equilibrium is defined in terms of payoffs (the payoffs when players play the $\IU$ strategy are close to optimal), the $\IU$ strategy also approximates the equilibrium marginals (if an equilibrium exists)---that is, it is also an approximate equilibrium in terms of~strategies.

Our approximation results are valid even in the case where no equilibrium exists (and we do not include the assumption that requires the existence of the equilibrium). Particularly in the cases of the Colonel Blotto game where it is known that there exists no equilibrium yielding the uniform-type marginals, the $\IU$ strategy is still an approximate equilibrium, yet we suspect that in those cases the approximation error might be large. On the other hand, our work does not solve the question of the existence of an exact Nash equilibrium. In particular, we leave as future work the investigation of possible conditions under which a Nash equilibrium exists, for instance for a large-enough number of battlefields. We also finally note that, in the GCB game, the existence of multiple solutions $\gam $ of Equation~\eqref{eq:Equagamma} leads to problems of equilibrium selection (in practical contexts involving a social welfare measurement) among the $\IU$ strategies with different $\gam  \in \Sn$, which we also leave as future~work.

\section*{Acknowledgment}

This work was supported by the French National Research Agency through the ``Investissements d'avenir'' program (ANR-15-IDEX-02) and through grant ANR-16-TERC0012; and by the Alexander von Humboldt~Foundation.

\newpage

\bibliography{mybibfile}

\begin{thebibliography}{62}
\expandafter\ifx\csname natexlab\endcsname\relax\def\natexlab#1{#1}\fi
\providecommand{\url}[1]{\texttt{#1}}
\providecommand{\href}[2]{#2}
\providecommand{\path}[1]{#1}
\providecommand{\DOIprefix}{doi:}
\providecommand{\ArXivprefix}{arXiv:}
\providecommand{\URLprefix}{URL: }
\providecommand{\Pubmedprefix}{pmid:}
\providecommand{\doi}[1]{\href{http://dx.doi.org/#1}{\path{#1}}}
\providecommand{\Pubmed}[1]{\href{pmid:#1}{\path{#1}}}
\providecommand{\bibinfo}[2]{#2}
\ifx\xfnm\relax \def\xfnm[#1]{\unskip,\space#1}\fi
\bibitem[{Adamo \& Matros(2009)}]{adamo2009blotto}
\bibinfo{author}{Adamo, T.}, \& \bibinfo{author}{Matros, A.}
  (\bibinfo{year}{2009}).
\newblock \bibinfo{title}{A blotto game with incomplete information}.
\newblock {\it \bibinfo{journal}{Economics Letters}\/},  {\it
  \bibinfo{volume}{105}\/}, \bibinfo{pages}{100--102}.
\bibitem[{Ahmadinejad et~al.(2016)Ahmadinejad, Dehghani, Hajiaghayi, Lucier,
  Mahini \& Seddighin}]{Ahmadinejad16a}
\bibinfo{author}{Ahmadinejad, A.~M.}, \bibinfo{author}{Dehghani, S.},
  \bibinfo{author}{Hajiaghayi, M.~T.}, \bibinfo{author}{Lucier, B.},
  \bibinfo{author}{Mahini, H.}, \& \bibinfo{author}{Seddighin, S.}
  (\bibinfo{year}{2016}).
\newblock \bibinfo{title}{From duels to battlefields: Computing equilibria of
  blotto and other games}.
\newblock In {\it \bibinfo{booktitle}{Proceedings of the 13th AAAI Conference
  on Artificial Intelligence (AAAI)}\/} (pp. \bibinfo{pages}{369--375}).
\bibitem[{Alcalde \& Dahm(2007)}]{alcalde2007tullock}
\bibinfo{author}{Alcalde, J.}, \& \bibinfo{author}{Dahm, M.}
  (\bibinfo{year}{2007}).
\newblock \bibinfo{title}{Tullock and hirshleifer: a meeting of the minds}.
\newblock {\it \bibinfo{journal}{Review of Economic Design}\/},  {\it
  \bibinfo{volume}{11}\/}, \bibinfo{pages}{101--124}.
\bibitem[{Baye et~al.(1996)Baye, Kovenock \& de~Vries}]{baye1996all}
\bibinfo{author}{Baye, M.~R.}, \bibinfo{author}{Kovenock, D.}, \&
  \bibinfo{author}{de~Vries, C.~G.} (\bibinfo{year}{1996}).
\newblock \bibinfo{title}{The all-pay auction with complete information}.
\newblock {\it \bibinfo{journal}{Economic Theory}\/},  {\it
  \bibinfo{volume}{8}\/}, \bibinfo{pages}{291--305}.
\bibitem[{Behnezhad et~al.(2018)Behnezhad, Blum, Derakhshan, HajiAghayi,
  Mahdian, Papadimitriou, Rivest, Seddighin \&
  Stark}]{behnezhad2018battlefields}
\bibinfo{author}{Behnezhad, S.}, \bibinfo{author}{Blum, A.},
  \bibinfo{author}{Derakhshan, M.}, \bibinfo{author}{HajiAghayi, M.},
  \bibinfo{author}{Mahdian, M.}, \bibinfo{author}{Papadimitriou, C.~H.},
  \bibinfo{author}{Rivest, R.~L.}, \bibinfo{author}{Seddighin, S.}, \&
  \bibinfo{author}{Stark, P.~B.} (\bibinfo{year}{2018}).
\newblock \bibinfo{title}{From battlefields to elections: Winning strategies of
  blotto and auditing games}.
\newblock In {\it \bibinfo{booktitle}{Proceedings of the 29th Annual ACM-SIAM
  Symposium on Discrete Algorithms}\/} (pp. \bibinfo{pages}{2291--2310}).
\newblock \bibinfo{organization}{SIAM}.
\bibitem[{Behnezhad et~al.(2019)Behnezhad, Blum, Derakhshan, Hajiaghayi,
  Papadimitriou \& Seddighin}]{Behnezhad19a}
\bibinfo{author}{Behnezhad, S.}, \bibinfo{author}{Blum, A.},
  \bibinfo{author}{Derakhshan, M.}, \bibinfo{author}{Hajiaghayi, M.},
  \bibinfo{author}{Papadimitriou, C.~H.}, \& \bibinfo{author}{Seddighin, S.}
  (\bibinfo{year}{2019}).
\newblock \bibinfo{title}{Optimal strategies of blotto games: Beyond
  convexity}.
\newblock In {\it \bibinfo{booktitle}{Proceedings of the 2019 ACM Conference on
  Economics and Computation (EC)}\/} (p. \bibinfo{pages}{597–616}).
\bibitem[{Behnezhad et~al.(2017)Behnezhad, Dehghani, Derakhshan, Aghayi \&
  Seddighin}]{Behnezhad17a}
\bibinfo{author}{Behnezhad, S.}, \bibinfo{author}{Dehghani, S.},
  \bibinfo{author}{Derakhshan, M.}, \bibinfo{author}{Aghayi, M. T.~H.}, \&
  \bibinfo{author}{Seddighin, S.} (\bibinfo{year}{2017}).
\newblock \bibinfo{title}{Faster and simpler algorithm for optimal strategies
  of blotto game}.
\newblock In {\it \bibinfo{booktitle}{Proceedings of the 31st AAAI Conference
  on Artificial Intelligence (AAAI)}\/} (pp. \bibinfo{pages}{369--375}).
\bibitem[{Boix-Adser\`{a} et~al.(2020)Boix-Adser\`{a}, Edelman \&
  Jayanti}]{boix-adsera2020}
\bibinfo{author}{Boix-Adser\`{a}, E.}, \bibinfo{author}{Edelman, B.~L.}, \&
  \bibinfo{author}{Jayanti, S.} (\bibinfo{year}{2020}).
\newblock \bibinfo{title}{The multiplayer colonel blotto game}.
\newblock In {\it \bibinfo{booktitle}{Proceedings of the 21st ACM Conference on
  Economics and Computation (EC)}\/} (p. \bibinfo{pages}{47–48}).
\newblock \bibinfo{address}{New York, NY, USA}: \bibinfo{publisher}{Association
  for Computing Machinery}.
\bibitem[{Borel(1921)}]{borel1921}
\bibinfo{author}{Borel, E.} (\bibinfo{year}{1921}).
\newblock \bibinfo{title}{La th{\'e}orie du jeu et les {\'e}quations
  int{\'e}grales {\`a} noyau sym{\'e}trique}.
\newblock {\it \bibinfo{journal}{Comptes rendus de l’Acad{\'e}mie des
  Sciences}\/},  {\it \bibinfo{volume}{173}\/}, \bibinfo{pages}{58}.
\bibitem[{Borel \& Ville(1938)}]{borel1938}
\bibinfo{author}{Borel, E.}, \& \bibinfo{author}{Ville, J.}
  (\bibinfo{year}{1938}).
\newblock {\it \bibinfo{title}{Application de la th{\'e}orie des
  probabilit{\'e}s aux jeux de hasard}\/}.
\newblock \bibinfo{publisher}{Gauthier-Villars}.
\newblock \bibinfo{note}{Original edition by Gauthier-Villars, Paris, 1938;
  reprinted at the end of Th{\'e}orie math{\'e}matique du bridge {\`a} la
  port{\'e}e de tous, by E. Borel \& A. Ch{\'e}ron, Editions Jacques Gabay,
  Paris}.
\bibitem[{Brassard \& Bratley(1996)}]{brassard1996}
\bibinfo{author}{Brassard, G.}, \& \bibinfo{author}{Bratley, P.}
  (\bibinfo{year}{1996}).
\newblock {\it \bibinfo{title}{Fundamentals of Algorithmics}\/}.
\newblock \bibinfo{publisher}{Prentice-Hall, Inc.}
\bibitem[{Che \& Gale(2000)}]{che2000difference}
\bibinfo{author}{Che, Y.-K.}, \& \bibinfo{author}{Gale, I.}
  (\bibinfo{year}{2000}).
\newblock \bibinfo{title}{Difference-form contests and the robustness of
  all-pay auctions}.
\newblock {\it \bibinfo{journal}{Games and Economic Behavior}\/},  {\it
  \bibinfo{volume}{30}\/}, \bibinfo{pages}{22 -- 43}.
\bibitem[{Chia(2012)}]{chia2012}
\bibinfo{author}{Chia, P.~H.} (\bibinfo{year}{2012}).
\newblock \bibinfo{title}{{C}olonel {B}lotto in web security}.
\newblock In {\it \bibinfo{booktitle}{The 11th Workshop on Economics and
  Information Security, WEIS Rump Session}\/} (pp. \bibinfo{pages}{141--150}).
\bibitem[{Clark \& Riis(1998)}]{clark1998contest}
\bibinfo{author}{Clark, D.~J.}, \& \bibinfo{author}{Riis, C.}
  (\bibinfo{year}{1998}).
\newblock \bibinfo{title}{Contest success functions: an extension}.
\newblock {\it \bibinfo{journal}{Economic Theory}\/},  {\it
  \bibinfo{volume}{11}\/}, \bibinfo{pages}{201--204}.
\bibitem[{Corch{\'o}n(2007)}]{corchon2007theory}
\bibinfo{author}{Corch{\'o}n, L.~C.} (\bibinfo{year}{2007}).
\newblock \bibinfo{title}{The theory of contests: a survey}.
\newblock {\it \bibinfo{journal}{Review of Economic Design}\/},  {\it
  \bibinfo{volume}{11}\/}, \bibinfo{pages}{69--100}.
\bibitem[{Corch{\'o}n \& Dahm(2010)}]{corchon2010foundations}
\bibinfo{author}{Corch{\'o}n, L.~C.}, \& \bibinfo{author}{Dahm, M.}
  (\bibinfo{year}{2010}).
\newblock \bibinfo{title}{Foundations for contest success functions}.
\newblock {\it \bibinfo{journal}{Economic Theory}\/},  {\it
  \bibinfo{volume}{43}\/}, \bibinfo{pages}{81--98}.
\bibitem[{Cormen et~al.(2009)Cormen, Leiserson, Rivest \& Stein}]{cormen2009}
\bibinfo{author}{Cormen, T.~H.}, \bibinfo{author}{Leiserson, C.~E.},
  \bibinfo{author}{Rivest, R.~L.}, \& \bibinfo{author}{Stein, C.}
  (\bibinfo{year}{2009}).
\newblock {\it \bibinfo{title}{Introduction to Algorithms, Third Edition}\/}.
\newblock \bibinfo{publisher}{The MIT Press}.
\bibitem[{Duffy \& Matros(2015{\natexlab{a}})}]{duffy2015stochastic}
\bibinfo{author}{Duffy, J.}, \& \bibinfo{author}{Matros, A.}
  (\bibinfo{year}{2015}{\natexlab{a}}).
\newblock \bibinfo{title}{Stochastic asymmetric blotto games: Some new
  results}.
\newblock {\it \bibinfo{journal}{Economics Letters}\/},  {\it
  \bibinfo{volume}{134}\/}, \bibinfo{pages}{4--8}.
\bibitem[{Duffy \& Matros(2015{\natexlab{b}})}]{duffy2015}
\bibinfo{author}{Duffy, J.}, \& \bibinfo{author}{Matros, A.}
  (\bibinfo{year}{2015}{\natexlab{b}}).
\newblock \bibinfo{title}{Stochastic asymmetric {B}lotto games: Some new
  results}.
\newblock {\it \bibinfo{journal}{Economics Letters}\/},  {\it
  \bibinfo{volume}{134}\/}, \bibinfo{pages}{4--8}.
\bibitem[{Friedman(1958)}]{friedman1958}
\bibinfo{author}{Friedman, L.} (\bibinfo{year}{1958}).
\newblock \bibinfo{title}{Game-theory models in the allocation of advertising
  expenditures}.
\newblock {\it \bibinfo{journal}{Operations Research}\/},  {\it
  \bibinfo{volume}{6}\/}, \bibinfo{pages}{699--709}.
\bibitem[{Fu \& Iyer(2019)}]{fu2019multimarket}
\bibinfo{author}{Fu, Q.}, \& \bibinfo{author}{Iyer, G.} (\bibinfo{year}{2019}).
\newblock \bibinfo{title}{Multimarket value creation and competition}.
\newblock {\it \bibinfo{journal}{Marketing Science}\/},  {\it
  \bibinfo{volume}{38}\/}, \bibinfo{pages}{129--149}.
\bibitem[{Fu \& Wu(2019)}]{fu2019contests}
\bibinfo{author}{Fu, Q.}, \& \bibinfo{author}{Wu, Z.} (\bibinfo{year}{2019}).
\newblock \bibinfo{title}{Contests: Theory and topics}.
\newblock In {\it \bibinfo{booktitle}{Oxford Research Encyclopedia of Economics
  and Finance}\/}.
\newblock \bibinfo{publisher}{Oxford University Press}.
\bibitem[{Gross(1950)}]{gross1950}
\bibinfo{author}{Gross, O.} (\bibinfo{year}{1950}).
\newblock {\it \bibinfo{title}{The symmetric Blotto game}\/}.
\newblock \bibinfo{type}{Technical Report} US Air Force Project RAND Research
  Memorandum.
\bibitem[{Gross \& Wagner(1950)}]{grosswagner}
\bibinfo{author}{Gross, O.}, \& \bibinfo{author}{Wagner, R.}
  (\bibinfo{year}{1950}).
\newblock {\it \bibinfo{title}{A continuous Colonel Blotto game}\/}.
\newblock \bibinfo{type}{Technical Report} RAND project air force Santa Monica
  CA.
\bibitem[{{Hajimirsaadeghi} \& {Mandayam}(2017)}]{hajimirsaadeghi2017dynamic}
\bibinfo{author}{{Hajimirsaadeghi}, M.}, \& \bibinfo{author}{{Mandayam}, N.~B.}
  (\bibinfo{year}{2017}).
\newblock \bibinfo{title}{A dynamic colonel {B}lotto game model for spectrum
  sharing in wireless networks}.
\newblock In {\it \bibinfo{booktitle}{Proceedings of the 55th Annual Allerton
  Conference on Communication, Control, and Computing (Allerton)}\/} (pp.
  \bibinfo{pages}{287--294}).
\bibitem[{Hart(2008)}]{hart2008}
\bibinfo{author}{Hart, S.} (\bibinfo{year}{2008}).
\newblock \bibinfo{title}{Discrete {C}olonel {B}lotto and {G}eneral {L}otto
  games}.
\newblock {\it \bibinfo{journal}{International Journal of Game Theory}\/},
  {\it \bibinfo{volume}{36}\/}, \bibinfo{pages}{441--460}.
\bibitem[{Hillman \& Riley(1989)}]{hillman1989}
\bibinfo{author}{Hillman, A.~L.}, \& \bibinfo{author}{Riley, J.~G.}
  (\bibinfo{year}{1989}).
\newblock \bibinfo{title}{Politically contestable rents and transfers}.
\newblock {\it \bibinfo{journal}{Economics \& Politics}\/},  {\it
  \bibinfo{volume}{1}\/}, \bibinfo{pages}{17--39}.
\bibitem[{Hirshleifer(1989)}]{hirshleifer1989conflict}
\bibinfo{author}{Hirshleifer, J.} (\bibinfo{year}{1989}).
\newblock \bibinfo{title}{Conflict and rent-seeking success functions: Ratio
  vs. difference models of relative success}.
\newblock {\it \bibinfo{journal}{Public choice}\/},  {\it
  \bibinfo{volume}{63}\/}, \bibinfo{pages}{101--112}.
\bibitem[{Hoeffding(1963)}]{hoeffding1963probability}
\bibinfo{author}{Hoeffding, W.} (\bibinfo{year}{1963}).
\newblock \bibinfo{title}{Probability inequalities for sums of bounded random
  variables}.
\newblock {\it \bibinfo{journal}{Journal of the American Statistical
  Association,}\/},  {\it \bibinfo{volume}{58}\/}, \bibinfo{pages}{13--30}.
\bibitem[{Hortala-Vallve \& Llorente-Saguer(2012)}]{hortala2012}
\bibinfo{author}{Hortala-Vallve, R.}, \& \bibinfo{author}{Llorente-Saguer, A.}
  (\bibinfo{year}{2012}).
\newblock \bibinfo{title}{{Pure strategy Nash equilibria in non-zero sum
  colonel Blotto games}}.
\newblock {\it \bibinfo{journal}{International Journal of Game Theory}\/},
  {\it \bibinfo{volume}{41}\/}, \bibinfo{pages}{331--343}.
\bibitem[{Kim \& Kim(2019)}]{kim2019existence}
\bibinfo{author}{Kim, B.}, \& \bibinfo{author}{Kim, J.} (\bibinfo{year}{2019}).
\newblock \bibinfo{title}{Existence of a unique nash equilibrium for an
  asymmetric lottery blotto game with weighted majority}.
\newblock {\it \bibinfo{journal}{Journal of Mathematical Analysis and
  Applications}\/},  {\it \bibinfo{volume}{479}\/},
  \bibinfo{pages}{1403--1415}.
\bibitem[{Kim et~al.(2018)Kim, Kim \& Kim}]{kim2018lottery}
\bibinfo{author}{Kim, G.~J.}, \bibinfo{author}{Kim, J.}, \&
  \bibinfo{author}{Kim, B.} (\bibinfo{year}{2018}).
\newblock \bibinfo{title}{A lottery blotto game with heterogeneous items of
  asymmetric valuations}.
\newblock {\it \bibinfo{journal}{Economics Letters}\/},  {\it
  \bibinfo{volume}{173}\/}, \bibinfo{pages}{1--5}.
\bibitem[{Klumpp et~al.(2019)Klumpp, Konrad \& Solomon}]{klumpp2019dynamics}
\bibinfo{author}{Klumpp, T.}, \bibinfo{author}{Konrad, K.~A.}, \&
  \bibinfo{author}{Solomon, A.} (\bibinfo{year}{2019}).
\newblock \bibinfo{title}{The dynamics of majoritarian blotto games}.
\newblock {\it \bibinfo{journal}{Games and Economic Behavior}\/},  {\it
  \bibinfo{volume}{117}\/}, \bibinfo{pages}{402--419}.
\bibitem[{Kovenock \& Arjona(2019)}]{kovenock2019full}
\bibinfo{author}{Kovenock, D.}, \& \bibinfo{author}{Arjona, D.~R.}
  (\bibinfo{year}{2019}).
\newblock \bibinfo{title}{A full characterization of best-response functions in
  the lottery colonel blotto game}.
\newblock {\it \bibinfo{journal}{Economics Letters}\/},  {\it
  \bibinfo{volume}{182}\/}, \bibinfo{pages}{33--36}.
\bibitem[{Kovenock \& Roberson(2011)}]{kovenock2011blotto}
\bibinfo{author}{Kovenock, D.}, \& \bibinfo{author}{Roberson, B.}
  (\bibinfo{year}{2011}).
\newblock \bibinfo{title}{A blotto game with multi-dimensional incomplete
  information}.
\newblock {\it \bibinfo{journal}{Economics Letters}\/},  {\it
  \bibinfo{volume}{113}\/}, \bibinfo{pages}{273--275}.
\bibitem[{Kovenock \& Roberson(2012{\natexlab{a}})}]{kovenock2012}
\bibinfo{author}{Kovenock, D.}, \& \bibinfo{author}{Roberson, B.}
  (\bibinfo{year}{2012}{\natexlab{a}}).
\newblock \bibinfo{title}{Coalitional {C}olonel {B}lotto games with application
  to the economics of alliances}.
\newblock {\it \bibinfo{journal}{Journal of Public Economic Theory}\/},  {\it
  \bibinfo{volume}{14}\/}, \bibinfo{pages}{653--676}.
\bibitem[{Kovenock \& Roberson(2012{\natexlab{b}})}]{kovenock2012conflicts}
\bibinfo{author}{Kovenock, D.}, \& \bibinfo{author}{Roberson, B.}
  (\bibinfo{year}{2012}{\natexlab{b}}).
\newblock \bibinfo{title}{Conflicts with multiple battlefields}.
\newblock In {\it \bibinfo{booktitle}{The Oxford Handbook of the Economics of
  Peace and Conflict}\/}.
\newblock \bibinfo{publisher}{Oxford University Press}.
\bibitem[{Kovenock \& Roberson(2020)}]{kovenock2020generalizations}
\bibinfo{author}{Kovenock, D.}, \& \bibinfo{author}{Roberson, B.}
  (\bibinfo{year}{2020}).
\newblock \bibinfo{title}{Generalizations of the {G}eneral {L}otto and
  {C}olonel {B}lotto games}.
\newblock {\it \bibinfo{journal}{Economic Theory}\/},  (pp.
  \bibinfo{pages}{1--36}).
\bibitem[{Kvasov(2007)}]{kvasov2007contests}
\bibinfo{author}{Kvasov, D.} (\bibinfo{year}{2007}).
\newblock \bibinfo{title}{Contests with limited resources}.
\newblock {\it \bibinfo{journal}{Journal of Economic Theory}\/},  {\it
  \bibinfo{volume}{136}\/}, \bibinfo{pages}{738--748}.
\bibitem[{Laslier(2002)}]{laslier2002}
\bibinfo{author}{Laslier, J.~F.} (\bibinfo{year}{2002}).
\newblock \bibinfo{title}{How two-party competition treats minorities}.
\newblock {\it \bibinfo{journal}{Review of Economic Design}\/},  {\it
  \bibinfo{volume}{7}\/}, \bibinfo{pages}{297--307}.
\bibitem[{Laslier(2005)}]{laslier2005}
\bibinfo{author}{Laslier, J.~F.} (\bibinfo{year}{2005}).
\newblock \bibinfo{title}{Party objectives in the “divide a dollar”
  electoral competition}.
\newblock {\it \bibinfo{journal}{Social Choice and Strategic Decisions}\/},
  (pp. \bibinfo{pages}{113--130}).
\bibitem[{Laslier \& Picard(2002)}]{laslier2002distributive}
\bibinfo{author}{Laslier, J.~F.}, \& \bibinfo{author}{Picard, N.}
  (\bibinfo{year}{2002}).
\newblock \bibinfo{title}{Distributive politics and electoral competition}.
\newblock {\it \bibinfo{journal}{Journal of Economic Theory}\/},  {\it
  \bibinfo{volume}{103}\/}, \bibinfo{pages}{106--130}.
\bibitem[{Macdonell \& Mastronardi(2015)}]{macdonell2015}
\bibinfo{author}{Macdonell, S.~T.}, \& \bibinfo{author}{Mastronardi, N.}
  (\bibinfo{year}{2015}).
\newblock \bibinfo{title}{Waging simple wars: a complete characterization of
  two-battlefield blotto equilibria}.
\newblock {\it \bibinfo{journal}{Economic Theory}\/},  {\it
  \bibinfo{volume}{58}\/}, \bibinfo{pages}{183--216}.
\bibitem[{Masucci \& Silva(2014)}]{masucci2014}
\bibinfo{author}{Masucci, A.~M.}, \& \bibinfo{author}{Silva, A.}
  (\bibinfo{year}{2014}).
\newblock \bibinfo{title}{Strategic resource allocation for competitive
  influence in social networks}.
\newblock In {\it \bibinfo{booktitle}{Proceedings of the 52nd Annual Allerton
  Conference on Communication, Control, and Computing (Allerton)}\/} (pp.
  \bibinfo{pages}{951--958}).
\bibitem[{Masucci \& Silva(2015)}]{masucci2015}
\bibinfo{author}{Masucci, A.~M.}, \& \bibinfo{author}{Silva, A.}
  (\bibinfo{year}{2015}).
\newblock \bibinfo{title}{Defensive resource allocation in social networks}.
\newblock In {\it \bibinfo{booktitle}{Proceedings of the 54th IEEE Conference
  on Decision and Control (CDC)}\/} (pp. \bibinfo{pages}{2927--2932}).
\bibitem[{Myerson(1991)}]{myerson1991game}
\bibinfo{author}{Myerson, R.~B.} (\bibinfo{year}{1991}).
\newblock {\it \bibinfo{title}{Game Theory: Analysis of Conflict}\/}.
\newblock \bibinfo{publisher}{Harvard University Press}.
\bibitem[{Myerson(1993)}]{myerson1993incentives}
\bibinfo{author}{Myerson, R.~B.} (\bibinfo{year}{1993}).
\newblock \bibinfo{title}{Incentives to cultivate favored minorities under
  alternative electoral systems}.
\newblock {\it \bibinfo{journal}{American Political Science Review}\/},  {\it
  \bibinfo{volume}{87}\/}, \bibinfo{pages}{856--869}.
\bibitem[{Nisan et~al.(2007)Nisan, Roughgarden, Tardos \& Vazirani}]{Nisan07}
\bibinfo{author}{Nisan, N.}, \bibinfo{author}{Roughgarden, T.},
  \bibinfo{author}{Tardos, E.}, \& \bibinfo{author}{Vazirani, V.~V.}
  (\bibinfo{year}{2007}).
\newblock {\it \bibinfo{title}{Algorithmic Game Theory}\/}.
\newblock \bibinfo{publisher}{Cambridge University Press}.
\bibitem[{Osório(2013)}]{osorio2013}
\bibinfo{author}{Osório, A.} (\bibinfo{year}{2013}).
\newblock \bibinfo{title}{The lottery {B}lotto game}.
\newblock {\it \bibinfo{journal}{Economics Letters}\/},  {\it
  \bibinfo{volume}{120}\/}, \bibinfo{pages}{164--166}.
\bibitem[{Paarporn et~al.(2019)Paarporn, Chandan, Alizadeh \&
  Marden}]{paarporn2019characterizing}
\bibinfo{author}{Paarporn, K.}, \bibinfo{author}{Chandan, R.},
  \bibinfo{author}{Alizadeh, M.}, \& \bibinfo{author}{Marden, J.~R.}
  (\bibinfo{year}{2019}).
\newblock \bibinfo{title}{Characterizing the interplay between information and
  strength in blotto games}.
\newblock {\it \bibinfo{journal}{arXiv preprint arXiv:1909.03382}\/}, .
\bibitem[{Powell(2009)}]{powell2009}
\bibinfo{author}{Powell, R.} (\bibinfo{year}{2009}).
\newblock \bibinfo{title}{Sequential, nonzero-sum “{B}lotto”: Allocating
  defensive resources prior to attack}.
\newblock {\it \bibinfo{journal}{Games and Economic Behavior}\/},  {\it
  \bibinfo{volume}{67}\/}, \bibinfo{pages}{611--615}.
\bibitem[{Rinott et~al.(2012)Rinott, Scarsini \& Yu}]{rinott2012}
\bibinfo{author}{Rinott, Y.}, \bibinfo{author}{Scarsini, M.}, \&
  \bibinfo{author}{Yu, Y.} (\bibinfo{year}{2012}).
\newblock \bibinfo{title}{A {C}olonel {B}lotto gladiator game}.
\newblock {\it \bibinfo{journal}{Mathematics of Operations Research}\/},  {\it
  \bibinfo{volume}{37}\/}, \bibinfo{pages}{574--590}.
\bibitem[{Roberson(2006)}]{roberson2006}
\bibinfo{author}{Roberson, B.} (\bibinfo{year}{2006}).
\newblock \bibinfo{title}{The {C}olonel {B}lotto game}.
\newblock {\it \bibinfo{journal}{Economic Theory}\/},  {\it
  \bibinfo{volume}{29}\/}, \bibinfo{pages}{1--24}.
\bibitem[{Robson(2005)}]{robson2005}
\bibinfo{author}{Robson, A.} (\bibinfo{year}{2005}).
\newblock {\it \bibinfo{title}{Multi-item contests. Australian National
  University}\/}.
\newblock \bibinfo{type}{Technical Report} Working Paper.
\bibitem[{Schwartz et~al.(2014)Schwartz, Loiseau \& Sastry}]{schwartz2014}
\bibinfo{author}{Schwartz, G.}, \bibinfo{author}{Loiseau, P.}, \&
  \bibinfo{author}{Sastry, S.~S.} (\bibinfo{year}{2014}).
\newblock \bibinfo{title}{The heterogeneous {C}olonel {B}lotto game}.
\newblock In {\it \bibinfo{booktitle}{Proceedings of the 7th International
  Conference on Network Games, Control and Optimization (NetGCoop)}\/} (pp.
  \bibinfo{pages}{232--238}).
\bibitem[{Shubik \& Weber(1981)}]{shubik1981systems}
\bibinfo{author}{Shubik, M.}, \& \bibinfo{author}{Weber, R.~J.}
  (\bibinfo{year}{1981}).
\newblock \bibinfo{title}{Systems defense games: Colonel blotto, command and
  control}.
\newblock {\it \bibinfo{journal}{Naval Research Logistics Quarterly}\/},  {\it
  \bibinfo{volume}{28}\/}, \bibinfo{pages}{281--287}.
\bibitem[{Skaperdas(1996)}]{skaperdas1996}
\bibinfo{author}{Skaperdas, S.} (\bibinfo{year}{1996}).
\newblock \bibinfo{title}{Contest success functions}.
\newblock {\it \bibinfo{journal}{Economic theory}\/},  {\it
  \bibinfo{volume}{7}\/}, \bibinfo{pages}{283--290}.
\bibitem[{Thomas(2017)}]{thomas2017}
\bibinfo{author}{Thomas, C.} (\bibinfo{year}{2017}).
\newblock \bibinfo{title}{N-dimensional {B}lotto game with heterogeneous
  battlefield values}.
\newblock {\it \bibinfo{journal}{Economic Theory}\/},  (pp.
  \bibinfo{pages}{1--36}).
\bibitem[{Tullock(1980)}]{tullock1980}
\bibinfo{author}{Tullock, G.} (\bibinfo{year}{1980}).
\newblock \bibinfo{title}{Efficient rent seeking}.
\newblock In {\it \bibinfo{booktitle}{{L}ockard {A.A., T}ullock {G.} (eds)
  (2001). {E}fficient Rent Seeking Chronicle of an Intellectual Quagmire}\/}.
\newblock \bibinfo{publisher}{Springer}.
\newblock \bibinfo{note}{(Reprint of Tullock, Gordon (1980) in Toward a Theory
  of the Rent-Seeking Society. Efficient Rent Seeking)}.
\bibitem[{Van~der Vaart(2000)}]{van2000asymptotic}
\bibinfo{author}{Van~der Vaart, A.~W.} (\bibinfo{year}{2000}).
\newblock {\it \bibinfo{title}{Asymptotic statistics}\/}
  volume~\bibinfo{volume}{3}.
\newblock \bibinfo{publisher}{Cambridge university press}.
\bibitem[{Vu et~al.(2018)Vu, Loiseau \& Silva}]{vu18a}
\bibinfo{author}{Vu, D.~Q.}, \bibinfo{author}{Loiseau, P.}, \&
  \bibinfo{author}{Silva, A.} (\bibinfo{year}{2018}).
\newblock \bibinfo{title}{Efficient computation of approximate equilibria in
  discrete colonel blotto games}.
\newblock In {\it \bibinfo{booktitle}{Proceedings of the 27th International
  Joint Conference on Artificial Intelligence and the 23rd European Conference
  on Artificial Intelligence (IJCAI-ECAI)}\/} (pp. \bibinfo{pages}{519--526}).
\bibitem[{Weinstein(2005)}]{weinstein2005}
\bibinfo{author}{Weinstein, J.} (\bibinfo{year}{2005}).
\newblock \bibinfo{title}{Two notes on the {B}lotto game. {N}orthwestern
  {U}niversity.}

\end{thebibliography}

\newpage
\appendix
\setcounter{lemma}{0}
\renewcommand{\thelemma}{\Alph{section}\arabic{lemma}}
\renewcommand{\theremark}{\Alph{section}\arabic{remark}}
\renewcommand{\theproposition}{\Alph{section}\arabic{proposition}}

\section{Nomenclatures and Preliminaries}
\label{sec:appen_preliminary}

\begin{centering}
\begin{table}[htb!]
	\caption{Table of Notation}
	\captionsetup{justification=centering}
\begin{tabular}{r c p{12cm} }
\multicolumn{3}{c}{\underline{Abbreviation}}\\
    GCB (or $\CB$) & $\triangleq$ & Generalized Colonel Blotto game (with $n$ battlefields)\\
    GLB (or $\LB$) & $\triangleq$ & Generalized Lottery Blotto game (with $n$ battlefields)\\
    $\CCB$ & $\triangleq$ & the constant-sum versions of $\CB$\\
	$\CLB$ & $\triangleq$ & the constant-sum versions of $\LB$\\
    CSF & $\triangleq$ & contest success function\\
    OUD & $\triangleq$ &  optimal univariate distribution\\
    $\IU$ ($=(\IU_A, \IU_B)$) & $\triangleq$ & independent uniform strategy (corresponding to $\gamma^*$)\\  
    \multicolumn{3}{c}{\underline{Games' Parameters}}\\
    $X_A, X_B$ & $\triangleq$ & budgets of Player A and B respectively\\  
    $n$ & $\triangleq$ & number of battlefields\\  
    $w^A_i, w^B_i$ & $\triangleq$ & values of battlefield $i$ assessed by Player A and B respectively\\
    $\wmin, \wmax$ & $\triangleq$ & lower and upper bounds of battlefields' values\\  
    $W_A, W_B$ & $\triangleq$ & sums of battlefields' values, $W_A:= \sum_{i=1}^n w^A_i$, $W_B:= \sum_{i=1}^n w^B_i$\\
    $W$ & $\triangleq$ &  $\max\{W_A, W_B\}$\\  
    $v^A_i, v^B_i$ & $\triangleq$ & normalized values of battlefield $i$ assessed by Player A and B\\  
    $x^A_i, x^B_i$ & $\triangleq$ & the allocation to battlefield $i$ of Player A and B respectively\\  
    $ \Pi_A(s,t), \Pi_B(s,t)$ & $\triangleq$ & players' payoffs inG CB games when playing the strategies $s$ and~$t$\\ 
    $\alpha$ & $\triangleq$ & the tie-breaking parameter \\
    $\beta_A, \beta_B$ & $\triangleq$ & Blotto functions (see \eqref{eq:betafunction})\\ 
    $\zeta$ & $\triangleq$ & $(\zeta_A, \zeta_B)$---the generic CSFs\\
    $\LB(\zeta)$ & $\triangleq$ & the GLB game with CSFs $\zeta_A, \zeta_B$\\
    $\mu^R$ & $\triangleq$ & $(\mu_A^R, \mu^R_B)$---the power form CSFs with parameter $R$ (see Table~\ref{table:CSF})\\
    $\nu^R$ & $\triangleq$ & $(\nu_A^R, \nu^R_B)$---the logit form CSFs with parameter $R$ (see Table~\ref{table:CSF})\\
    $ \Pi^{\zeta}_A(s,t), \Pi^{\zeta}_B(s,t)$ & $\triangleq$ & players' payoffs in $\LB(\zeta)$ games when playing the strategies $s$, $t$\\ 
    $\X(y^*, \varepsilon), \Y(x^*, \varepsilon) $  & $\triangleq$ & the sets characterizing the dissimilarity between $(\beta_A,\beta_B)$ and $(\zeta_A,\zeta_B)$ (see~\eqref{eq:Xset},~\eqref{eq:Yset}) \\
    $\Del$  & $\triangleq$ &  see Definition~\ref{def:delta} \\
    \multicolumn{3}{c}{\underline{$\IU$ Strategies}}\\
    $\gamma^*$ & $\triangleq$ & a positive solution of Equation~\eqref{eq:Equagamma}\\
    $\Sn$ & $\triangleq$ & the set of positive solutions of Equation~\eqref{eq:Equagamma} w.r.t. $\CB$\\
    $\Gmin, \Gmax$ & $\triangleq$ & lower and upper bounds of any $\gamma^* \in \Sn$ (see Proposition~\ref{Prop:BoundLambda})\\
    $\Omega_A(\gamma^*)$ & $\triangleq$ & $ \left\{i \in [n]: v^A_i/v^B_i > \gamma^* \right\}$\\
    $\LdaA, \LdaB$ & $\triangleq$ & Lagrange multipliers corresponding to $\gamma^*$ (see \eqref{eq:lambdaA}, \eqref{eq:lambdaB}) \\
    $\Lmin, \Lmax$ & $\triangleq$ & lower and upper bounds of $\LdaA, \LdaB$ (see Proposition~\ref{Prop:BoundLambda})\\
    $\FA,\FB$& $\triangleq$ & uniform-type distributions (see \eqref{As}-\eqref{A*B*}\\ 
    $\AS,\AW,A^n_i$ & $\triangleq$ & random variables defined in~\eqref{As}-\eqref{A*B*}\\
    $ \BS,\BW,B^n_i$ &   &   \\
    $\FAn, \FBn$ & $\triangleq$ & the marginals corresponding to battlefield $i$ of the $\IU$ strategy\\
    $A_{=0}, A_{> 0}$ & $\triangleq$ & the events $\{\sum_{i \in [n]} A^*_i = 0\}$ and $\{\sum_{i \in [n]} A^*_i > 0\}$, respectively\\
    $B_{=0}, B_{> 0}$ & $\triangleq$ & the events $\{\sum_{i \in [n]} B^*_i = 0\}$ and $\{\sum_{i \in [n]} B^*_i > 0\}$, respectively\\
    \bottomrule
\end{tabular}
\label{table:notation}
\end{table}
\end{centering}

Throughout the paper, we use the asymptotic notation (Bachmann–Landau notations) ${\mO}$ by its standard definition,\footnote{This definition is used by textbooks e.g., \citet{brassard1996,cormen2009}.} i.e., for any (real-valued) functions $f$, $g$ defined on an unbounded subset $S \subset \mathbb{R}^n_{>0}$, and $g(\boldsymbol{z}) >0$ for any $\boldsymbol{z} \in S$, we write $f(\boldsymbol{z}) = \mO\left(g(\boldsymbol{z}) \right)$ if \mbox{$ \exists M,C >0: |f(\boldsymbol{z})| \le C g(\boldsymbol{z}), \forall \boldsymbol{z} \in S: \boldsymbol{z}_i \ge M, \forall i \in [n]$}. Moreover, we also use another variant of $\mO$ that is $\tilde{\mO}$ where the logarithmic terms (in $\boldsymbol{z}_i$) are ignored. We also write $h(\boldsymbol{z}) \le  \mO(g(\boldsymbol{z}))$ if there exists a term $f(\boldsymbol{z}) = \mO(g(\boldsymbol{z}))$ such that $h(\boldsymbol{z}) \le f(\boldsymbol{z}), \forall \boldsymbol{z} \in S$ (similar notation for $\tilde{\mO}$ can be deduced).

In the remainders of this section, we introduce and prove several preliminary lemmas that are useful for our analysis.

\begin{lemma}
\label{lem:Preliminary}
    Given a game $\CB$ (or $\LB$), for any $\gam  \in \Sn$, we have:
    \begin{enumerate}
        \item[$(i)$] $\LdaA, \LdaB>0$ and $\gam  = {\LdaA}/{\LdaB}$.
        \item[$(ii)$] For any $i \in [n]$, $\Ex[\AS]= \frac{1}{2}\frac{v^B_i}{\LdaB}$, $\Ex[\AW]= \left(\frac{v^A_i}{\LdaA}\right)^2 \frac{\LdaB}{2 v^B_i}$, $\Ex[\BS]= \frac{1}{2}\frac{v^A_i}{\LdaA}$ and  $\Ex[\BW]= \left(\frac{v^B_i}{\LdaB}\right)^2 \frac{\LdaA}{2 v^A_i}$. 
        \item[$(iii)$] $X_A = \sum\nolimits_{i \in [n]}{\Ex[A^*_i]}$ and $X_B = \sum\nolimits_{i \in [n]}{\Ex[B^*_i]}$.%
        \item[$(iv)$] For any $i \in [n]$, $A^*_i$ and $B^*_i$ have a constant upper-bound; particularly, 
        \begin{equation*}
            \prob\left(A^*_i  \le 2X_B \right) =  \prob\left(B^*_i \le  2X_B \right)  = 1.
        \end{equation*}
\end{enumerate}
\end{lemma}

\begin{proof}
 $ $\newline
\begin{itemize}
    \item[$(i)$] The positivity of $\LdaA$ and $\LdaB$ follows from the positivity of $\gam $ and the definitions of $\LdaA$ and $\LdaB$ in \eqref{eq:lambdaA} and~\eqref{eq:lambdaB}. By dividing \eqref{eq:lambdaA} by \eqref{eq:lambdaB} and combining with \eqref{eq:Equagamma}, we trivially have that $\gam  = {\LdaA}/{\LdaB}$. 
    \item[$(ii)$] These results come directly from the definitions of the distributions $F_{\AS}$, $F_{\AW}$ $F_{\BS}$ and~$F_{\BW}$.
    \item[$(iii)$] We multiply both sides of \eqref{eq:lambdaB} by $X_A/\LdaB$ and both sides of \eqref{eq:lambdaA} by~$X_B/\LdaA$ then using the fact that $\gam  = \LdaA/\LdaB$ to obtain the following:
    \begin{align}
        X_A & = \sum_{j \in \Ostar}{\frac{1}{2} \frac{v^B_j}{\LdaB}} +  \sum_{j \notin \Ostar}{\left(\frac{v^A_j}{\LdaA}\right)^2 \frac{\LdaB}{2 v^B_j}}, \label{eq:prelimi_XA}\\
        X_B & = \sum_{j \in \Ostar}{\left(\frac{v^B_j}{\LdaB}\right)^2 \frac{\LdaA}{2 v^A_j}} +  \sum_{j \notin \Ostar}{\frac{1}{2}\frac{v^A_j}{\LdaA}}. \label{eq:prelimi_XB}
    \end{align}
        Combining with $(ii)$, we deduce that $X_A = \sum\nolimits_{i \in [n]}{\Ex[A^*_i]}$ and $X_B = \sum\nolimits_{i \in [n]}{\Ex[B^*_i]}$.
    \item[$(iv)$] If $i \in \Ostar$, we have $A^*_i = \AS$ and $B^*_i = \BW$. Recalling Definition~\ref{def:UnifromDistributions}, we have that \mbox{$\prob \left (A^S_i \le {v^B_i}/{\LdaB} \right) = 1$} and \mbox{$\prob \left( B^W_i \le {v^B_i}/{\LdaB} \right) =1$}. On the other hand, from \eqref{eq:prelimi_XA}, we deduce
    \begin{equation*}
            X_B \ge X_A \ge \sum_{j \in \Ostar}{\frac{v^B_j}{2\LdaB}} \ge \frac{v^B_i}{2 \LdaB}. 
    \end{equation*}
    Therefore, $ \prob ( A^S_i \! \le\! 2X_B) \!\ge \!\prob\left(A^S_i \! \le \! {v^B_i}/{\LdaB} \right) \! = \! 1 $ and $\prob(B^W_i  \le 2X_B) \ge \prob(B^W_i  \le {v^B_i}/{\LdaB}) = 1$. We conclude that for any $i \in \Ostar$, $A^*_i, B^*_i$ are bounded by $2X_B$. 
    
    If $i \notin \Ostar$, we have $A^*_i = \AW$ and $B^*_i = \BS$. Recalling Definition~\ref{def:UnifromDistributions}, we have that \mbox{$\prob \left (A^W_i \le {v^A_i}/{\LdaA} \right) = 1$} and \mbox{$\prob \left( B^S_i \le {v^A_i}/{\LdaA} \right) =1$}. On the other hand, from \eqref{eq:prelimi_XB}, we deduce
    \begin{equation*}
            X_B  \ge \sum_{j \notin \Ostar}{\frac{v^A_j}{2\LdaA}} \ge \frac{v^A_i}{2 \LdaA}. 
    \end{equation*}
    Therefore, $ \prob ( A^W_i  \le 2X_B) \ge \prob\left(A^W_i  \le {v^A_i}/{\LdaA} \right) =1$ and $\prob(B^S_i  \le 2X_B) \ge \prob(B^S_i  \le {v^A_i}/{\LdaA})=1$. We conclude that for $i \notin \Ostar$, $A^*_i, B^*_i$ are also bounded by $2X_B$. \qed
\end{itemize}
\end{proof}
%

\boundpropo*
\begin{proof}
 
Let $\gam  \in \Sn$, we consider the following cases:
\paragraph{Case 1: If $0< \gam  < \min \limits_{i \in [n]} \left\{ \frac{v^A_i}{v^B_i} \right\}$.} In this case, $\Ostar = [n]$, and since $\gam $ is a solution of \eqref{eq:Equagamma}, we deduce:
    \begin{equation*}
        \gam  = \frac{X_B}{X_A} \frac{\sum\nolimits_{i=1}^n{v^B_i}} {\sum\nolimits_{i=1}^n{\frac{(v^B_i)^2}{v^A_i}}} \ge \frac{X_B}{X_A} \frac{n \frac{\wmin}{n \wmax}}{n\frac{\left(\frac{\wmax}{n \wmin} \right)^2 }{\frac{\wmin}{n \wmax}}} = \frac{X_B}{X_A} \left(\frac{\wmin}{\wmax} \right)^4.
    \end{equation*}

Here, the inequality comes directly from~\eqref{eq:bound_v^p_i}.
\paragraph{Case 2: If $\gam   \ge \max \limits_{i \in [n]} \left\{ \frac{v^A_i}{v^B_i} \right\}$.} In this case, $\Ostar = \emptyset$, and since $\gam $ is a solution of \eqref{eq:Equagamma}, we~deduce:
     \begin{equation*}
         \gam  = \frac{X_B}{X_A} \frac{\sum\nolimits_{i=1}^n{\frac{(v^A_i)^2}{v^B_i}}} {\sum\nolimits_{i=1}^n{v^A_i}}  \le \frac{X_B}{X_A} \left(\frac{\wmax}{\wmin} \right)^4.
     \end{equation*}

\paragraph{Case 3: If $\exists i, j : \frac{v^A_i}{v^B_i} \le \gam  < \frac{v^A_j}{v^B_j}$.} In this case, trivially from \eqref{eq:bound_v^p_i}, we have $\gam  \in \left[ {\left( \frac{\wmin}{\wmax} \right)^2} , {\left(\frac{\wmax}{\wmin} \right)^2} \right]$.

In conclusion, by denoting \mbox{$\Gmin := \min \left\{\frac{X_B}{X_A} \left(\frac{\wmin}{\wmax} \right)^4 , {\left( \frac{\wmin}{\wmax} \right)^2} \right\}$} and \mbox{$\Gmax := \max \left\{ \frac{X_B}{X_A} \left(\frac{\wmax}{\wmin} \right)^4,{\left(\frac{\wmax}{\wmin} \right)^2}\right\} =  \frac{X_B}{X_A} \left(\frac{\wmax}{\wmin} \right)^4$}, we have the conclusion on the bounds of $\gam $. 

On the other hand, from the definition of $\LdaA$ in \eqref{eq:lambdaA}, we deduce
\begin{align*}
    \LdaA &\ge \frac{({\gam })^2}{2{X_B}}\sum\nolimits_{i \in {\Omega _A}({\gam })} { \left(\frac{\wmin}{n \wmax}\right)^2 \frac{1}{\frac{\wmax}{n \wmin}} }  + \frac{1}{{2{X_B}}}\sum\nolimits_{i \notin {\Omega _A}({\gam })} {\frac{\wmin}{n \wmax}} \\
    & \ge \min \left\{\frac{({\gam })^2}{2{X_B}},\frac{1}{2{X_B}} \right\} \cdot \sum \nolimits_{i \in [n]} {\frac{1}{n} \left(\frac{\wmin}{ \wmax}\right)^3}\\
    & \ge \min \left\{\frac{({\gam })^2}{2{X_B}},\frac{1}{2{X_B}} \right\} \cdot \left(\frac{\wmin}{ \wmax}\right)^3.
    \end{align*}
    Similarly, we have the upper-bound
    \begin{equation*}
    \LdaA \le \max \left\{\frac{({\gam })^2}{2{X_B}},\frac{1}{2{X_B}} \right\} \cdot \left[ \sum\limits_{i \in {\Omega _A}({\gam })} { \frac{1}{n} \left(\frac{\wmax}{ \wmin}\right)^3 }  + \sum\limits_{i \notin {\Omega _A}({\gam })} { \frac{1}{n} \left(\frac{\wmax}{ \wmin}\right)^3 }\right] = \max \left\{\frac{({\gam })^2}{2{X_B}},\frac{1}{2{X_B}} \right\} \cdot\left(\frac{\wmax}{ \wmin}\right)^3.
\end{equation*}
Similarly, we can prove that $\min\left\{\frac{1}{2{X_A}}, \frac{1}{2{\gam } ^2 {X_A}}\right\} \left(\frac{\wmin}{ \wmax}\right)^3 \!\le\! \LdaB \le \max \left\{\frac{1}{2{X_A}}, \frac{1}{2{\gam } ^2 {X_A}}  \right\} \left(\frac{\wmax}{ \wmin}\right)^3 $; therefore, 
\begin{equation*}
    \min \left\{ \frac{{\gam }^2}{2{X_B}},\frac{1}{2{X_B}} , \frac{1}{2{X_A}}, \frac{1}{2(\gam ) ^2 {X_A}} \right\}\left(\frac{\wmin}{ \wmax}\right)^3 \le \LdaA, \LdaB \le \max \left\{\frac{{\gam }^2}{2{X_B}},\frac{1}{2{X_B}}, \frac{1}{2{X_A}}, \frac{1}{2(\gam ) ^2 {X_A}} \right\} \left(\frac{\wmax}{ \wmin}\right)^3.
\end{equation*}
Since $\gam  \in [\Gmin, \Gmax ]$, $\LdaA$ and $\LdaB$ are bounded in $[\Lmin, \Lmax]$, where
\begin{align*}
    & \Lmin:= \min \left\{ \frac{{\Gmin}^2}{2{X_B}},\frac{1}{2{X_B}} , \frac{1}{2{X_A}}, \frac{1}{2 \Gmax ^2 {X_A}} \right\}\left(\frac{\wmin}{ \wmax}\right)^3, \\
    & \Lmax:= \max \left\{\frac{{\Gmax}^2}{2{X_B}},\frac{1}{2{X_B}}, \frac{1}{2{X_A}}, \frac{1}{2 \Gmin ^2 {X_A}} \right\} \left(\frac{\wmax}{ \wmin}\right)^3.
\end{align*}
\qed
\end{proof}

Finally, we prove a trivial result that will be used quite often in the remainder of this work. 

\begin{lemma}
\label{lem:log_pre}
    For any $\hep>0$ and $\hC \ge 1$, we have that \mbox{$(\ln(\hC) + 1) \ln \left( \frac{1}{\min\{\hep, 1/\e \}}  \right) \ge \ln \left( \frac{\hC}{\hep} \right) $}.
\end{lemma}

\begin{proof}

$ $\newline

\underline{Case 1:} If $\hep < 1/\e$. In this case, we have $\ln(1/\hep) > 1$; therefore,
\begin{align*}
    (\ln(\hC) + 1) \ln \left( \frac{1}{\min\{\hep, 1/\e \}}  \right)  = (\ln(\hC) + 1) \ln \left( \frac{1}{\hep}  \right) & = \ln(\hC)\ln \left( \frac{1}{\hep}  \right) + \ln \left( \frac{1}{\hep}  \right) \\
    & > \ln(\hC) + \ln \left( \frac{1}{\hep}  \right)\\
    & = \ln \left( \frac{\hC}{\hep} \right).
\end{align*}

\underline{Case 2:} If $\hep \ge 1/\e$. We have $\ln(1/\hep) \le 1$; therefore,
\begin{align*}
    (\ln(\hC) + 1) \ln \left( \frac{1}{\min\{\hep, 1/\e \}}  \right)  = (\ln(\hC) + 1) \ln \left( \frac{1}{1/\e}  \right)  = \ln(\hC) + 1 \ge \ln(\hC) + \ln \left( \frac{1}{\hep}  \right)  = \ln \left( \frac{\hC}{\hep} \right). 
\end{align*}
\qed
\end{proof}

\section{Proof of Theorem~\ref{TheoMainBlotto}}
\label{sec:Appen_Proof_TheoBlotto}

First note that in the remainders of the paper, for any bounded, non-negative random variable $Z$ (i.e., \mbox{$\exists C>0:$} $\prob(Z\in[0,C])=1$), any measurable function $g$ on $\mathbb{R}$, we write $\int \nolimits_0^{\infty} \! {g(x) \de {F_{Z}(x)}}$ instead of $\int \nolimits_0^{C} \!  {g(x) \de {F_{Z}(x)}}$ if there is no need to emphasize the bounds of $Z$. For the sake of notation, we also denote by $A_{=0}$ the event $\left\{ \sum \nolimits_{j \in [n]}{A^*_j} = 0 \right\}$ and by $A_{>0}$ its complement event, that is $\left\{ \sum \nolimits_{j \in [n]}{A^*_j} > 0 \right\}$. Similarly, we denote by $B_{=0}$ the event $\left\{ \sum \nolimits_{j \in [n]}{B^*_j} = 0 \right\}$ and by $B_{>0}$ the event $\left\{ \sum \nolimits_{j \in [n]}{B^*_j} > 0 \right\}$.

Recall the notation $\FAn$ and $\FBn$ as the univariate marginal distributions corresponding to battlefield $i \in [n]$ of the $\IU_A$ and $\IU_B$ strategies (the corresponding random variables are denoted $A^n_i$ and $B^n_i$). Due to the definition of the $\IU$ strategy (via Algorithm~\ref{alg:IU_strategy}), for any $x\ge 0$ and $i \in [n]$, we have:
\begin{align}
    \FAn(x) &= \prob\left(\left\{A^n_i \le x\right\} \bigcap A_{=0} \right)  + \prob\left( \left\{A^n_i \le x \right\} \bigcap A_{>0} \right) \nonumber \\
                %
                & = \prob\left( A_{=0}\right) + \! \prob \left( \left\{ \frac{A^*_i \cdot X_A}{\sum \nolimits_{j\in[n]} A^*_j} \!\le\! x \right\}   \bigcap A_{>0} \right). \label{eq:A^n_Def}
\end{align}
Here, we have used the fact that if $\sum \nolimits_{j\in[n]} A^*_j = 0$ (i.e., when $A_{=0}$ happens), then $A^n_i = 0$ by definition and thus, $\prob(A^n_i \le x) =1$ and $\prob\left(\left\{A^n_i \le x\right\} \bigcap A_{=0} \right) = \prob(A_{=0})$. Similarly to~\eqref{eq:A^n_Def}, for any $x\ge 0$ and $i \in [n]$, 
\begin{equation}
    \FBn(x) = \prob\left( B_{=0}\right) + \! \prob \left( \left\{ \frac{B^*_i \cdot X_B}{\sum \nolimits_{j\in[n]} B^*_j} \!\le\! x \right\}   \bigcap B_{>0} \right). \label{eq:B^n_Def}
\end{equation}
Regarding the random variables $A^n_i$ and $B^n_i$ ($i \in [n]$), we prepare a lemma stating several useful results as follows (its proof is given in~\ref{sec:appen_proof_continuity}).
\begin{lemma}
    \label{lem:continuity_Ani_and_Bni}
    For any $n$ and $i \in [n]$, we have
    \begin{itemize}
        \item[$(i)$] $\prob(A^n_i =0) = \prob(A^*_i = 0)$ and $\prob(B^n_i =0) = \prob(B^*_i =0 )$.
        \item[$(ii)$] $\prob(A^n_i = x) = \prob(B^n_i =y) = 0$ for any $x \in (0,\infty) \backslash \{X_A\}$ and $y \in (0,\infty) \backslash \{X_B\}$.
         \item[$(iii)$] $\prob(A^n_i = X_A) \le \left( 1 -\frac{\Lmin}{\Lmax} \frac{\wmin ^2}{\wmax ^2} \right)^{n-1}$ and $\prob(B^n_i = X_B) \le \left( 1 -\frac{\Lmin}{\Lmax} \frac{\wmin ^2}{\wmax ^2} \right)^{n-1}$.
        %
    \end{itemize}
\end{lemma}
Intuitively, Result $(ii)$ states that the function $\FAn$ (resp. $\FBn$) is continuous on $(0, X_A)$ (resp. $(0,X_B)$). The discontinuity of $\FAn$ (resp. $\FBn$) at $X_A$ (resp. at $X_B$) is due to the normalization step involved in the definition of the $\IU$ strategy; note that the probability that $A^n_i = X_A$ (resp. $B^n_i = X_B$) quickly tends to zero when $n$ increases as has been shown in Result $(iii)$. Finally, Result~$(i)$ shows that in some cases, $\FAn$ and $\FBn$ may be discontinuous at $0$. This is due to the fact that the functions $\FA$ and $\FB$ may be discontinuous at $0$. Moreover, recall that we chose the assignments of the outputs in line 3 and 7 of Algorithm~\ref{alg:IU_strategy} to be allocating zero to every battlefield, i.e., the mass at $0$ of $\FAn$ and $\FBn$ is added by a (negligibly small) positive probability. While other assignments do not affect our results, they make $\FAn$ (resp. $\FBn$) be discontinuous at some points differing from $0$ and $X_A$ (resp. $X_B$), e.g., if in line 3 of Algorithm~\ref{alg:IU_strategy}, we assign $x^A_i = X_A/n$, the distribution $\FAn$ would also be discontinuous at the point $X_A/n$. Our choice of assignments provides more convenience in our analysis since we have to consider their discontinuity at $0$ in any~case.

Finally, with all the preparation steps mentioned above, we are ready to prove Theorem~\ref{TheoMainBlotto}.
\TheoMainBlotto*

\begin{proof}
 
In this section, we first give a proof of Result~$(ii)$ of Theorem~\ref{TheoMainBlotto}. Based on this, we then deduce Result $(i)$ of Theorem~\ref{TheoMainBlotto}. We first look for the condition on $n$ such that \eqref{eq:MainTheo_A} holds for any pure strategy $\boldsymbol{x}^A$ of Player A. The proof that~\eqref{eq:MainTheo_B} holds for any pure strategy of Player B under the same condition can be done similarly and thus is omitted.

First, we write explicitly the payoffs of Player A when Player B plays the $\IU_B$ strategy and Player A plays either the pure strategy $\boldsymbol{x}^A$ or the $\IU_A$ strategy:
\begin{align}
     && \Pi_A(\boldsymbol{x}^A, \IU_B)   & =    \alpha \sum \limits_{i=1}^{n}w^A_i \prob (B^n_i = x^A_i)   + \sum \limits_{i=1}^{n}w^A_i \prob (B^n_i < x^A_i), \label{eq:Pi_A_pure}\\
    && \Pi_A(\IU_A, \IU_B)   & = \alpha \sum \limits_{i=1}^{n}w^A_i \prob (B^n_i = A^n_i) + \sum \limits_{i=1}^{n}w^A_i \prob (B^n_i< A^n_i) \nonumber \\
    &&  & = \alpha \sum\limits_{i = 1}^n {\int_0^\infty  {w_i^A{\prob(B^n_i = x)} \de {\FAn(x)}} } + \sum\limits_{i = 1}^n {\int_0^\infty  {w_i^A{\prob(B^n_i < x)} \de {\FAn(x)}} }. \label{eq: Pi_A_IU}
\end{align}
We then prepare a useful lemma, its proof is given in \ref{sec:appen_proof_lem:SufCon}. Intuitively, this lemma shows that as $n$ is large enough, we can prove~\eqref{eq:MainTheo_A} without the need of analyzing separately the case where players get tie allocations (that is our results hold regardless of the tie-breaking-rule parameter $\alpha$). 
\begin{restatable}{lemma}{lem:SufCon}
\label{lem:SufCon}
   Given $\wmin, \wmax, X_A, X_B >0$ ($\wmin \le \wmax$, $X_A \le X_B$), there exists a constant $C^*_0>0$ such that for any $\varepsilon \in (0,1]$ and $n\ge C^*_0 \logep $, for any game $\CB$ and $\gam  \in \Sn$ the following inequality is a sufficient condition of~\eqref{eq:MainTheo_A}:
    \begin{equation}
         \sum \limits_{i=1}^{n}v^A_i \FBn\left( x^A_i\right) \le \sum\limits_{i = 1}^n {\int_0^\infty  {v_i^A{\FBn(x)} \de {\FAn(x)}} } + \frac{\varepsilon}{2}.  \label{eq:lem_ignore_tie}
    \end{equation}
\end{restatable}
In the remainders of the proof, we focus on~\eqref{eq:lem_ignore_tie} and look for the condition of $n$ such that it holds; this will be done in the following five~steps. After that, from Lemma~\ref{lem:SufCon}, we can conclude that~\eqref{eq:MainTheo_A} also holds with the corresponding condition on $n$.

\paragraph{\underline{Step 1:} Prove that $\{\FA\}_i$ is optimal against $\{\FB\}_i$.} 
\begin{restatable}{lemma}{lemdevia}
\label{lem:best_response}
    In any game $\CB$, for any pure strategy $\boldsymbol{x}^A$ of Player A and $\gam  \in \Sn$, we~have 
    \begin{equation}
        \sum\limits_{i=1}^n {v_i^A{\FB\left( {x_i^A} \right)} } \le \sum\limits_{i = 1}^n {\int_0^\infty  {v_i^A{\FB(x)} \de {\FA(x)}} } . \label{eq:lem_best_res}
    \end{equation}
\end{restatable}
\noindent The  proof of Lemma \ref{lem:best_response} is given in \ref{sec:appen_proof_best_respond}. This lemma can be interpreted as follows: if the allocation of Player B to battlefield $i$ follows the distribution $\FB$, then it is optimal for Player A to play such that her allocation at this battlefield follows $\FA$ (we do not know if it is possible to construct a mixed strategy such that Player A's allocation at battlefield $i$ follows $\FA$ for all $i \in [n]$; however, this does not affect our results in this work). Using this lemma, we will analyze the validity of \eqref{eq:lem_ignore_tie} by proving that, as $n \rightarrow \infty$, the terms in \eqref{eq:lem_ignore_tie} respectively converge toward the terms in \eqref{eq:lem_best_res}. To do this, we consider the next step.


\paragraph{\underline{Step 2:} Prove that $\FAn$ and $\FBn$ uniformly converge toward $\FA$ and $\FB$ as $n$~increases.}
\begin{restatable}{lemma}{lemmaconverge}
\label{lem:convergence}
    Given $\wmin, \wmax, X_A, X_B >0$ ($\wmin \le \wmax$, $X_A \le X_B$), there exists $C_1>0$ such that for any $ \varepsilon_1 \in (0, 1]$, $n \ge C_1  {\varepsilon_1^{-2}} \logone$ and $ i \in [n]$,
    \begin{equation}
        \mathop {\sup }\limits_{ x \in [0, \infty) } \left|\FAn(x) - \FA(x)\right| \le \varepsilon_1 \hspace{0.3cm} \text{ and } \hspace{0.3cm} \mathop {\sup }\limits_{ x \in [0, \infty) } \left|\FBn(x) - \FB(x)\right| \le \varepsilon_1. \label{ineqconver}
    \end{equation}
\end{restatable}
A proof of this lemma is given in~\ref{sec:appen_proof_lem_convergence}. The main intuition of this result comes from the fact that $A^n_i$ (resp. $B^n_i$) is the normalization of $A^*_i, i \in [n]$ (except for the special cases of the events $A_{=0}$ and $B_{=0}$) and the use of concentration inequalities on the random variables $\sum_{j \in [n]} A^*_j$ (and $\sum_{j \in [n]} B^*_j$). In this work, we apply the Hoeffding's inequality (Theorem~2, \citet{hoeffding1963probability}) to obtain the rate of convergence indicated here in Lemma~\ref{lem:convergence}. 
\paragraph{\underline{Step 3:} Prove that the left-hand-side of \eqref{eq:lem_ignore_tie} converges toward the left-hand-side of \eqref{eq:lem_best_res}.} 

Take $C_1$ as indicated in Lemma~\ref{lem:convergence}, we define \mbox{$C^*_1 \!:=\! 16 C_1(\ln(4) \!+ \!1)$} and deduce that \mbox{$ \frac{C^*_1}{\varepsilon^{2}} \logep \! \ge\! C_1 \left(\frac{4}{\varepsilon} \right)^2  \ln \left(\frac{1}{\min\{\frac{\varepsilon}{4}, \frac{1}{\e}  \}} \right)$}.\footnote{This is due to $C^*_1 \cdot \varepsilon^{-2} \logep = C_1 \left(\frac{4}{\varepsilon} \right)^2  \cdot (\ln(4)+1)\logep \ge  C_1 \cdot  \left(\frac{4}{\varepsilon} \right)^2 \ln \left(\frac{4}{\varepsilon} \right)$; here, we have applied Lemma~\ref{lem:log_pre} with $\hep:= \varepsilon$ and $\hC:= 4$; moreover, $\frac{\varepsilon}{4} =  \min\{\frac{\varepsilon}{4}, \frac{1}{e}\}$ since $\varepsilon \le 1$; thus, we can rewrite \mbox{$\ln\left(\frac{4}{\varepsilon} \right) = \ln \left(\frac{1}{\min\{\varepsilon/4, 1/\e\}} \right)$}.}

Therefore, take $\varepsilon_1:= \varepsilon/4$, for any \mbox{$n \ge C^*_1 \varepsilon^{-2} \logep$, we have $n \ge C_1 \varepsilon_1^{-2} \logone$}; apply Lemma~\ref{lem:convergence}, for any pure strategy $\boldsymbol{x}^A$ of Player A, we~have
\begin{align}
    \left| \sum \limits_{i=1}^{n}v^A_i \FBn\left( x^A_i\right) -  \sum \limits_{i=1}^{n}v^A_i \FB\left( x^A_i\right)  \right|
        \le  & \sum \limits_{i=1}^n{v^A_i \mathop {\sup }\limits_{ x \in [0, \infty) }  \left| {{\FBn}\left( {x} \right)}-{{\FB}\left( {x} \right)} \right|} \nonumber \\
        \le & \sum \limits_{i=1}^n{v^A_i \frac{\varepsilon}{4}} = \frac{\varepsilon}{4} \label{lhs}.
\end{align}
\paragraph{\underline{Step 4:} Prove that the right-hand-side of \eqref{eq:lem_ignore_tie} converges toward the right-hand-side of \eqref{eq:lem_best_res}.}
We consider the difference of the involved terms as follows.
\begin{align}
    & \left|\sum\limits_{i = 1}^n {\int_0^\infty  {v_i^A{\FBn(x)} \de {\FAn(x)}} }  - \sum\limits_{i = 1}^n {\int_0^\infty  {v_i^A{\FB(x)} \de {\FA(x)}} }  \right| \nonumber \\
    \le &  \sum\limits_{i = 1}^n v^A_i  \int_0^\infty  {\left|{\FBn(x) - \FB(x)} \right| \de {\FAn(x)}} +  \sum\limits_{i = 1}^n v^A_i \left| {\int_0^\infty  {{\FB}\left( x \right)}  \de {\FAn}\left( x \right) - \int_0^\infty  {{\FB}\left( x \right)}  \de {\FA}\left( x \right)} \right| \label{rhs_1}.
\end{align}
Let us define $C^*_2:= C_1 \cdot 64 (\ln(8)+1)$ (again, $C_1$ is the constant indicated in Lemma~\ref{lem:convergence}), we have that $C^*_2 \varepsilon^{-2} \logep \ge C_1 \left( \frac{8}{\varepsilon} \right)^{2} \ln \left(\frac{1}{\min\{\frac{\varepsilon}{8}, \frac{1}{\e}  \}} \right)$.\footnote{This is due to $C^*_2 \cdot \varepsilon^{-2} \logep = C_1 \left(\frac{8}{\varepsilon} \right)^2  \cdot (\ln(8)+1)\logep \ge  C_1 \cdot  \left(\frac{8}{\varepsilon} \right)^2 \ln \left(\frac{8}{\varepsilon} \right)$; here, we have applied Lemma~\ref{lem:log_pre} with $\hep:= \varepsilon$ and $\hC:= 8$; moreover, $\frac{\varepsilon}{8} =  \min\{\frac{\varepsilon}{8}, \frac{1}{e}\}$ since $\varepsilon \le 1$; thus, we can rewrite \mbox{$\ln\left(\frac{8}{\varepsilon} \right) = \ln \left(\frac{1}{\min\{\varepsilon/8, 1/\e\}} \right)$}.}  Therefore, take $\varepsilon_1:= \varepsilon/8$, for any \mbox{$n \ge C^*_2 \varepsilon^{-2} \logep$}, we have $n \ge \varepsilon_1^{-2} \logone$ and by Lemma \ref{lem:convergence}, we have 
\begin{equation}
    \sum\limits_{i = 1}^n v^A_i  \int_0^\infty  {\left|{\FBn(x) - \FB(x)} \right| \de {\FAn(x)}} \le \sum\limits_{i = 1}^n v^A_i  \int_0^\infty  { \frac{\varepsilon}{8} \de {\FAn(x)}} = \sum\limits_{i = 1}^n v^A_i \frac{\varepsilon}{8} . \label{bla1}
\end{equation}
Now, we need to find an upper-bound of the second term in the right-hand-side of\eqref{rhs_1}. To do this, we present a lemma, called Lemma~\ref{lem:portmanteau} (stated below), that is based on the portmanteau lemma (see, e.g., \citet{van2000asymptotic}) regarding the weak convergence of a sequence of~measures. Note importantly that by a direct application of the portmanteau lemma (since $\FB$ is Lipschitz continuous and from Lemma~\ref{lem:convergence}, $\FAn$ uniformly converges to $\FA$), we can prove that $\int_0^\infty  { \FB \left( x \right)}  \de {\FAn}\left( x \right)$ converges toward $\int_0^\infty  {\FB \left( x \right)}  \de {\FA}\left( x \right)$ as $n \rightarrow \infty$; however, note that the convergence rate obtained by doing this is large due to the fact that the Lipschitz constant of $\FB$ (that is $\LdaA / v^A_i$) increases as $n$ increases. To obtain a better convergence rate as indicated in~Lemma~\ref{lem:portmanteau}, we exploit the properties of the involved functions that allow us to use the telescoping sum trick (see~\ref{sec:appen_proof_lem_portmanteau} for more details).
\begin{lemma}
\label{lem:portmanteau}
    Given $\wmin, \wmax, X_A, X_B >0$ ($\wmin \le \wmax$, $X_A \le X_B$), there exists a constant $C_2>0$ such that for any  $\varepsilon_2\in (0, 1]$, $n \ge C_2 \cdot  {\varepsilon_2^{-2} \logtwo}$ and $ i \in [n]$, we have
    \begin{equation}
    \label{eq:portmanteau_main}
        \left| {\int_0^\infty  { \FB \left( x \right)}  \de {\FAn}\left( x \right) - \int_0^\infty  {\FB \left( x \right)}  \de {\FA}\left( x \right)} \right| \le \varepsilon_2.
    \end{equation}
\end{lemma}
%
%
%
%
The proof of Lemma~\ref{lem:portmanteau} is given~in~\ref{sec:appen_proof_lem_portmanteau}. Based on this constant $C_2$, we define \mbox{$C^*_3:= 8^2 C_2 (\ln8 \!+\!1)$}. Now, take $\varepsilon_2:= \varepsilon/8$, $C^*_3 \varepsilon^{-2} \logep \ge C_2 \varepsilon_2^{-2} \logtwo
$;\footnote{Once again, apply Lemma~\ref{lem:log_pre}, $C^*_3 \varepsilon^{-2} \logep\! = \!C_2 \left( \frac{8}{\varepsilon}\right)^2 (\ln(8)\! +\! 1) \logep \ge C_2 \left( \frac{8}{\varepsilon}\right)^2 \ln\left(\frac{8}{\varepsilon} \right)$; moreover, we have \mbox{$\varepsilon_2\!:= \!\frac{\varepsilon}{8} \!= \!\min \left\{\frac{\varepsilon}{8}, \frac{1}{\e} \right\}$}.}
thus, for $n \ge C^*_3 \varepsilon^{-2} \logep$, we have \mbox{$n \ge C_2 \varepsilon_2^{-2} \logtwo$} and by Lemma \ref{lem:portmanteau}, we deduce
\begin{equation*}
    \left| {\int_0^\infty  { \FB \left( x \right)}  \de {\FAn}\left( x \right) - \int_0^\infty  {\FB \left( x \right)}  \de {\FA}\left( x \right)} \right| \le \varepsilon/8.
    \end{equation*}

%

Combine this with \eqref{rhs_1} and \eqref{bla1}, for any $n = \max\{C^*_2, C^*_3 \} \varepsilon^{-2} \logep$, we~have
\begin{equation}
    \left|\sum\limits_{i = 1}^n {\int_0^\infty  {v_i^A{\FBn(x)} \de {\FAn(x)}} }  - \sum\limits_{i = 1}^n {\int_0^\infty  {v_i^A{\FB(x)} \de {\FA(x)}} }  \right| \le  \sum\limits_{i = 1}^n v^A_i \varepsilon/8 + \sum\limits_{i = 1}^n v^A_i \varepsilon/8  =  \frac{\varepsilon}{4}. \label{rhs}
\end{equation}

\paragraph{\underline{Step 5:} Conclusion.} For any $n \ge \max\{C^*_1, C^*_2, C^*_3\} \varepsilon^{-2} \logep$ and any pure strategy $\boldsymbol{x}^A$ of Player~A, we conclude that
\begin{align*}
     \sum \limits_{i=1}^{n}v^A_i \FBn\left( x^A_i\right) & \le   \sum \limits_{i=1}^{n}v^A_i \FB\left( x^A_i\right)  + \frac{\varepsilon}{4} &  \textrm{(from \eqref{lhs})} \\
    & \le  \sum\limits_{i = 1}^n {\int_0^\infty  {v_i^A{\FB(x)} \de {\FA(x)}} } +\frac{\varepsilon}{4}	& \textrm{(from \eqref{eq:lem_best_res})} \\
    & \le \sum\limits_{i = 1}^n {\int_0^\infty  {v_i^A{\FBn(x)} \de {\FAn(x)}} }  + \frac{\varepsilon}{4} + \frac{\varepsilon}{4} &     \textrm{(from \eqref{rhs})} \\
    & = \sum\limits_{i = 1}^n {\int_0^\infty  {v_i^A{\FBn(x)} \de {\FAn(x)}} }  + \frac{\varepsilon}{2}.
\end{align*}
This is exactly \eqref{eq:lem_ignore_tie}; therefore, apply Lemma~\ref{lem:SufCon} and denote \mbox{$C^*_{\eqref{eq:MainTheo_A}}\!:=\! \max\{C^*_0, C^*_1, C^*_2, C^*_3\}$}, we have proved that \eqref{eq:MainTheo_A} holds for any $n\ge C^*_{\eqref{eq:MainTheo_A}} \varepsilon^{-2} \logep$. Similarly, we can prove that there exists a constant $C^*_{\eqref{eq:MainTheo_B}}$ such that~\eqref{eq:MainTheo_B} holds for any $n \ge C^*_{\eqref{eq:MainTheo_B}} \varepsilon^{-2} \logep$. Finally, define \mbox{$C^*:= \max\{C^*_{\eqref{eq:MainTheo_A}}, C^*_{\eqref{eq:MainTheo_B}}\}$}, we conclude the proof for Result~$(ii)$ of Theorem~\ref{TheoMainBlotto}.

Now, to obtain Result~$(i)$ of Theorem~\ref{TheoMainBlotto}, we prove that it is implied by Result~$(ii)$ of Theorem~\ref{TheoMainBlotto}. Note that the constant $C^*$ found in the Result~$(ii)$ of Theorem~\ref{TheoMainBlotto} does not depend on neither $n$ nor $\varepsilon$. Moreover, the function 
    \begin{align*}
          \xi \colon (0,\infty) &\to (0, \infty)\\
                     \tep  &\mapsto  C^*  \tep ^{-2} \ln\left( \frac{1}{\min\{ \tep , 1/\e\}} \right).
    \end{align*}
is continuous and increases to infinity when $ \varepsilon $ tends to zero. Therefore, for any $n \ge 1$, there exists an \mbox{$ \varepsilon  >0$} such that $n = C^* \varepsilon ^{-2} \logep$. Now, apply Result~$(ii)$, \eqref{eq:MainTheo_A} and~\eqref{eq:MainTheo_B} hold in the game $\CB$ for any $\gam \in \Sn$ and pure strategies $\boldsymbol{x}^A, \boldsymbol{x}^B$. We conclude the proof by notice that if $ \varepsilon  \ge 1/\e $, we have $n=C^*  \varepsilon  ^{-2}$ and thus $ \varepsilon  = \sqrt{n/C^*} = \mO(n^{-1/2})$; on the other hand, if $ \varepsilon  < \frac{1}{\e} $, we have $\ln\left(\frac{1}{\varepsilon} \right) >1$ that induces \mbox{$n=C^*  \varepsilon ^{-2} \ln \left(\frac{1}{ \varepsilon } \right) \ge C^*  \varepsilon  ^{-2} \ge  \frac{C^*}{ \varepsilon }$}, thus, $\frac{1}{ \varepsilon } \le \frac{n}{C^*}$. We deduce that \mbox{$ \varepsilon  = \sqrt{\frac{C^*}{n} \ln\left(\frac{1}{ \varepsilon } \right) } \le \sqrt{\frac{C^*}{n}  \ln\left(\frac{n}{C^*} \right) }  = \tilde{\mO}(n^{-1/2})$}. \qed

\end{proof}

     \subsection{Proof of Lemma~\ref{lem:continuity_Ani_and_Bni}}
    \label{sec:appen_proof_continuity}
\begin{itemize}
    \item[$(i)$] Assuming $A^*_i =0$, if $\sum_{j \neq i}A^*_j =0$ then $A^n_i = 0$ (due to line 3 of Algorithm~\ref{alg:IU_strategy}) and if $\sum_{j \neq i} A^*_j > 0$ then $A^n_i = A^*_i \big/ \sum_{j \in [n]} A^*_j  = 0$. Reversely, assuming $A^n_i = 0$, then regardless whether $\sum_{j \in [n]} A^*_j=0$ or $\sum_{j \in [n]} A^*_j >0$, we have $A^*_i =0$. Therefore, $A^n_i = 0 \Leftrightarrow A^*_i =0$ for any $n$ and $i \in [n]$. Similarly, we can prove that $B^n_i = 0 \Leftrightarrow B^*_i =0$.
    \item[$(ii)$] The results are trivial in cases where $x > X_A$ and $y>X_B$ due to the definition of $A^n_i$ and $B^n_i$ (that guarantees that with probability $1$, $A^n_i \le X_A$ and $B^n_i \le X_B$). In the following, we consider the case where $x\in (0,X_A)$. For any $n$, $i \in [n]$, we denote $Z_i:= \sum_{j \neq i}A^*_j$ and obtain:
    \begin{align*}
        & \prob(A^n_i = x) \\
        =& \prob\left(  \{A^n_i =x \} \bigcap A_{>0}  \right)  & \textrm{(since } x>0) \\
            =& \prob\left(  \left\{A^*_i =\frac{x}{X_A} \sum \nolimits_{j \in [n]}A^*_j \right\} \bigcap A_{>0}  \right) \\
            = & \prob\left(  \left\{A^*_i \left(1- \frac{x}{X_A} \right) =\frac{x}{X_A} \sum \nolimits_{j \neq i}A^*_j \right\} \bigcap A_{>0}  \right)\\
            = & \prob\left(  \left\{A^*_i =\frac{Z_i \cdot x}{X_A -x} \right\} \bigcap A_{>0}  \right) & \textrm{(note that } X_A-x >0) \\
            \le & \prob(\{A^*_i = Z_i =0 \} \cap A_{>0}) + \int_{z>0}\prob \left(A^*_i = \frac{z \cdot x}{X_A - x} \right)  \de F_{Z_i}(z) \\
            \le & \prob(A_{=0} \cap A_{>0}) + \int_{z>0}{0~ \de F_{Z_i}(z) } 
                      \\
            = & 0.
    \end{align*}
Here, the second-to-last inequality comes from the fact that $\frac{z x}{X_A-x} >0, \forall z >0, \forall x \in (0,X_A)$ and $\prob(A^*_i = a) =0 $ for any $a >0$. Similarly, we can prove that $\prob(B^n_i = y) = 0$ for any $y \in (0,X_B)$.

%
\item[$(iii)$] We have 
\begin{align*}
    \prob(A^n_i = X_A) & = \prob\left( \left\{A^*_i = \sum_{j \in [n]}A^*_j  \right\} \bigcap A_{>0} \right) \\
    & \le \prob\left( \sum_{j \neq i} A^*_j =0 \right)\\
    & = \prod_{j \neq i} \prob \left( A^*_j = 0 \right)  & \textrm{(since } A^*_j, j \in [n] \textrm{ are non-negative and independent)}.
\end{align*}
Now, if there exists $j \neq i $ such that $ j \in \Ostar$, then $\prob(A^*_j = 0) =0 $ due to the fact that $A^*_j = A^S_{\gam ,j}$ and the definition of $ A^S_{\gam ,j}$ (see~\eqref{As}). In this case, $ \prod_{j \neq i} \prob \left( A^*_j = 0 \right)=0$. On the other hand, if $j \notin \Ostar $ for any $j \neq i$, then $A^*_j= A^W_{\gam , j}$ for $j \neq i$; therefore,
\begin{equation*}
      \prod_{j \neq i} \prob \left( A^*_j = 0 \right) = \prod_{j \neq i}\left[{\left(\frac{v^B_j}{\LdaB} - \frac{v^A_j}{\LdaA} \right)} \middle/ {\frac{v^B_j}{\LdaB}}\right] = \prod_{j\neq i} \left(1 - \frac{v^A_j}{v^B_j}\frac{\LdaB}{\LdaA} \right) \le \left(1-\frac{\Lmin}{\Lmax} \frac{\frac{\wmin}{n\wmax}}{\frac{\wmax}{n\wmin}}  \right)^{n-1} =  \left(1-\frac{\Lmin}{\Lmax} \frac{\wmin^2}{\wmax^2} \right)^{n-1}.   
\end{equation*}
Here, to obtain the last equality, we use ~\eqref{eq:bound_v^p_i} for the bounds of $v^A_j, v^B_j$ and Proposition~\ref{Prop:BoundLambda} for the bounds of $\LdaA, \LdaB$.

Similarly, we can obtain $\prob(B^n_i = X_B) \le \left(1-\frac{\Lmin}{\Lmax} \frac{\wmin^2}{\wmax^2} \right)^{n-1}$. \qed

\end{itemize}


    \subsection{Proof of Lemma~\ref{lem:SufCon}}
    \label{sec:appen_proof_lem:SufCon}
Fix $\varepsilon \in (0, 1]$ and assume that \eqref{eq:lem_ignore_tie} is satisfied, we prove that~\eqref{eq:MainTheo_A} also holds by comparing the terms in~\eqref{eq:lem_ignore_tie} with the terms in~\eqref{eq:MainTheo_A}. First, due to the fact that $\alpha \le 1$, we can find a lower bound of the left-hand side of \eqref{eq:lem_ignore_tie} as follows:
\begin{align}
    \sum \limits_{i=1}^{n}v^A_i \FBn\left( x^A_i\right) & = \sum \limits_{i=1}^{n}v^A_i \prob(B^n_i = x^A_i) + \sum \limits_{i=1}^{n}v^A_i \prob(B^n_i < x^A_i) \nonumber \\
    & \ge  \alpha \sum \limits_{i=1}^{n}v^A_i \prob(B^n_i = x^A_i) + \sum \limits_{i=1}^{n}v^A_i \prob(B^n_i < x^A_i)  \nonumber \\
    & = \Pi_A(\boldsymbol{x}^A, \IU_B)/W_A. \label{eq:B2_lhs}
\end{align}

Now, we turn our focus to the right-hand-side of~\eqref{eq:lem_ignore_tie}, we can rewrite the involved term as follows.
\begin{align*}
       \sum\limits_{i = 1}^n {\int_0^\infty  {v_i^A{\FBn(x)} \de {\FAn(x)}} } & = \sum\limits_{i = 1}^n {\int_0^\infty  {v_i^A{\prob(B^n_i = x)} \de {\FAn(x)}} } + \sum\limits_{i = 1}^n {\int_0^\infty  {v_i^A{\prob(B^n_i < x)} \de {\FAn(x)}} }.
\end{align*}
We observe that $\sum\limits_{i = 1}^n {\int_0^\infty  {v_i^A{\FBn(x)} \de {\FAn(x)}} } $ is very similar to the expression of $\Pi_A(\boldsymbol{x}^A, \IU_B)$ stated in~\eqref{eq: Pi_A_IU}. The main difference lies at the coefficient of the term related to the tie cases that is the tie-breaking parameter~$\alpha$. Therefore, we consider the following two cases of $\alpha$:

\textit{Case 1:} $\alpha = 1$. For any $n$, divide two sides of~\eqref{eq: Pi_A_IU} (with $\alpha =1$) by $W_A$ and recall that \mbox{$v^A_i:= w^A_i/W_A, \forall i $}, we trivially have \mbox{$\sum\limits_{i = 1}^n {\int_0^\infty  {v_i^A{\FBn(x)} \de {\FAn(x)}} } = \Pi_A(\IU_A, \IU_B) / W_A$} . 

\textit{Case 2:} $\alpha < 1$.  Due to Results~$(ii)$ and $(iii)$ of Lemma~\ref{lem:continuity_Ani_and_Bni}, for any $x >0$, we have \mbox{$\prob(B^n_i = x) \le D^{n-1}$} where we define $D:=  \left(1-\frac{\Lmin}{\Lmax} \frac{\wmin^2}{\wmax^2} \right) < 1$. We consider two cases of $\alpha$ as follows.

\begin{itemize}
    \item If $2(1-\alpha) \le 1$, define $\hat{C}_{1} :=  \frac{1}{\ln(1/D)} + 1>0$, we have that\footnote{If $\varepsilon < 1/ \e$, then $\ln(1/\varepsilon) > 1$ and $\hat{C}_1 \logep = \frac{\ln(1/\varepsilon)}{\ln(1/D)} + \ln(1/\varepsilon) >  \log_{D}{\varepsilon} +1 $; otherwise, if $\varepsilon \ge 1/e$, we have $\ln(1/\varepsilon) \le 1 $ and $\hat{C}_1 \logep =  \frac{1}{\ln(1/D)} + 1 \ge \frac{\ln(1/\varepsilon)}{\ln(1/D)} + 1 =   \log_{D}{\varepsilon} +1$ (note that $\ln(1/D)>0$ since $D<1$).}
    \mbox{$\hat{C}_1 \logep \ge \log_{D} {\varepsilon} + 1$}; therefore, for any $n \ge \hat{C}_1 \logep$, we obtain \mbox{$n-1 \ge \log_D {\varepsilon}$} and
    \begin{equation*}
        D^{n-1} \le D^{\log_D {\varepsilon}} = \varepsilon \le \frac{\varepsilon}{2(1-\alpha)} \qquad (\textrm{note that }  D < 1 \textrm{ and in this case } 2(1-\alpha) \le 1).
    \end{equation*}
     \item If $2(1\!-\!\alpha)\! > \!1$, define $\hat{C}_{2}\! := \! \frac{1}{\ln(1/D)} + \frac{\ln(2\!-\!2\alpha)}{\ln(1/D)} \!+ \! 1 \!> \!0$; we have
     \mbox{$\hat{C}_2 \logep \ge  \log_{D} {\frac{\varepsilon}{2(1-\alpha)}} + 1$}.\footnote{If $\varepsilon< 1/e$, we have \mbox{$\hat{C}_2 \logep =  \log_D{\varepsilon} + \left(   \log_{1/D} {(2-2\alpha)} + 1\right)   \ln\left(\frac{1}{\varepsilon} \right) >   \log_D {\varepsilon} \!- \!\log_D {(2-2\alpha)}  \!+ \!1$}; otherwise, if $\varepsilon \!\ge\! 1/\e$, we have \mbox{$\hat{C}_2 \logep = \hC_2 \ge   \frac{\ln(1/\varepsilon)}{\ln(1/D)} \! + \!   \frac{\ln(2-2\alpha)}{\ln(1/D)}  \! + \! 1 \ge  \log_D {\varepsilon} \!- \!\log_D {(2\!-\!2\alpha)}  \!+ \!1$} (since \mbox{$1 \ge \ln\left(\frac{1}{\varepsilon}\right)$}).}
     We conclude that for any $n \ge \hat{C}_2 \logep$, we obtain \mbox{$n-1 \ge \log_D {\left( \frac{ \varepsilon}{2(1-\alpha)} \right)}$} and
    \begin{equation*}
        D^{n-1} \le D^{\log_D {\frac{\varepsilon}{2(1- \alpha)}}} = \frac{\varepsilon}{2(1-\alpha)}.
    \end{equation*}
\end{itemize}
Let us define $C^*_0 = \max\{\hat{C}_1, \hat{C}_2\} > 0$, we conclude that for any $\alpha < 1$, $n\ge C^*_0 \logep$, $i \in [n]$ and $x >0$, we have
\begin{equation}
    \label{eq:probAn=Bn}
  \prob(B^n_i =x) \le   D^{n-1} \le \frac{\varepsilon}{2(1- \alpha)}.
\end{equation}
%
%

Note also that \mbox{$\prob(A^n_i = B^n_i = 0 ) = \prob(A^n_i = 0) \prob(B^n_i=0) = \prob(A^*_i =0) \prob(B^*_i=0) = 0, \forall i$},\footnote{Note that if \mbox{$i \in \Ostar$} then \mbox{$\prob(A^*_i = 0 ) = 0$}, if $i \notin \Ostar$ then $ \prob(B^*_i = 0 ) = 0$ (see~\eqref{As}-\eqref{A*B*}); therefore, \mbox{$\prob(A^*_i \!=\! 0 ) \prob(A^*_i\!=\!0) \!=\!0, \forall i$}.}
we conclude that when $\alpha < 1$, for any $n \ge C^*_0 \logep$, we have 
\begin{align*}
       & \sum\limits_{i = 1}^n {\int_0^\infty  {v_i^A{\FBn(x)} \de {\FAn(x)}} } \\
        = &  \left[ \sum\limits_{i = 1}^n {\int_0^\infty  {v_i^A{\prob(B^n_i < x)} \de {\FAn(x)}} } + \alpha \sum\limits_{i = 1}^n  {\int_0^\infty  {v_i^A{\prob(B^n_i = x)} \de {\FAn(x)}} }  \right] \\
                     & \hspace{5.5cm}       +  (1-\alpha)\sum\limits_{i = 1}^n  {\int_0^\infty  {v_i^A{\prob(B^n_i = x)} \de {\FAn(x)}} } \\
       = &  \Pi_A(\IU_A, \IU_B)/W_A +  (1-\alpha) \sum\limits_{i = 1}^n  v_i^A {\int_{(0, \infty)} {  \frac{\varepsilon}{2(1-\alpha)}  \de {\FAn(x)}} } + (1-\alpha) \sum_{i=1}^n v^A_i \prob(A^n_i = B^n_i =0) \\
       = & \Pi_A(\IU_A, \IU_B)/W_A  + (1-\alpha) \sum\limits_{i = 1}^n  v_i^A {\int_{(0,\infty)} {  \frac{\varepsilon}{2(1-\alpha)}  \de {\FAn(x)}} } + 0         \\
       \le & \Pi_A(\IU_A, \IU_B)/W_A  + (1- \alpha)\frac{\varepsilon}{2(1-\alpha)} \\
       = & \Pi_A(\IU_A, \IU_B)/W_A  + {\varepsilon}/{2}. 
\end{align*}
 
In conclusion, regardless of the value of $\alpha$, for any $n \ge C^*_0 \logep$, we have 
\begin{equation}
    \label{eq:B2_rhs}
    \sum\limits_{i = 1}^n {\int_0^\infty  {v_i^A{\FBn(x)} \de {\FAn(x)}} } \le \Pi_A(\IU_A, \IU_B)/W_A  + {\varepsilon}/{2}.
\end{equation}

Combine~\eqref{eq:B2_lhs}, \eqref{eq:B2_rhs} and the assumption that \eqref{eq:lem_ignore_tie} holds, for any $n \ge C^*_0 \logep$, we have
\begin{equation*}
    \frac{\Pi_A(\boldsymbol{x}^A, \IU_B)}{W_A} \stackrel{\eqref{eq:B2_lhs}}{\le}  \sum \limits_{i=1}^{n}v^A_i \FBn\left( x^A_i\right) \stackrel{\eqref{eq:lem_ignore_tie}}{\le} \sum\limits_{i = 1}^n {\int_0^\infty  {v_i^A{\FBn(x)} \de {\FAn(x)}} }\! +\! \varepsilon/2 \stackrel{\eqref{eq:B2_rhs}}{\le}  \frac{\Pi_A(\IU_A, \IU_B)}{W_A} \!+\! \varepsilon.
\end{equation*}
By multiplying both sides of this inequality by $W_A$, we obtain \eqref{eq:MainTheo_A}.    \qed

\subsection{Proof of Lemma~\ref{lem:best_response}}
\label{sec:appen_proof_best_respond}

We compute the right-hand-side of \eqref{eq:lem_best_res} based on the definition of $\FA$ and $\FB$ (see Definition~\ref{def:UnifromDistributions}).
\begin{align}
    \sum\limits_{i = 1}^n {\int_0^\infty  {v_i^A{\FB(x)} \de {\FA(x)}} }	& =	\sum\limits_{i \in {\Omega _A}\left( {{\gam }} \right)} {\int_0^\infty  {v_i^A{F_{B_{{\gam },i}^W}(x)} \de {F_{A_{{\gam },i}^S}(x)}} }  + \sum\limits_{i \notin {\Omega _A}\left( {{\gam }} \right)} {\int_0^\infty  {v_i^A{F_{B_{{\gam },i}^S}(x)} \de {F_{A_{{\gam },i}^W}(x)}} } \nonumber \\
    &= \sum\limits_{i \in {\Omega _A}\left( {{\gam }} \right)} {\int_0^{\frac{{v_i^B}}{{\lambda _B^*}}} {v_i^A\left( {\frac{{\frac{{v_i^A}}{{\lambda _A^*}} - \frac{{v_i^B}}{{\lambda _B^*}}}}{{\frac{{v_i^A}}{{\lambda _A^*}}}} + \frac{x}{{\frac{{v_i^A}}{{\lambda _A^*}}}}} \right)\frac{1}{{\frac{{v_i^B}}{{\lambda _B^*}}}} \de x} }  + \sum\limits_{i \notin {\Omega _A}\left( {{\gam }} \right)} {\int_0^{\frac{{v_i^A}}{{\lambda _A^*}}} {v_i^A\frac{x}{{\frac{{v_i^A}}{{\lambda _A^*}}}}\frac{1}{{\frac{{v_i^B}}{{\lambda _B^*}}}} \de x} }  \nonumber \\
    & = \sum\limits_{i \in {\Omega _A}\left( {{\gam }} \right)} v_i^A\left( 1 - \frac{{v_i^B \gam }}{2v^A_i} \right) +  \sum\limits_{i \notin {\Omega _A}\left( {{\gam }} \right)} {(v_i^A)^2\frac{1}{{2\gam  v^B_i}}}. \label{eq:proof_best_res_1}
\end{align}
On the other hand, for any pure strategy $\boldsymbol{x}^A$ of Player A, we have:
\begin{align}
    \sum\limits_{i=1}^n {v_i^A{\FB\left( {x_i^A} \right)} } =    & \sum\limits_{i \in {\Omega _A}\left( {{\gam }} \right)} v_i^A{F_{B_{\gam ,i}^W}}\left( {x_i^A} \right)  + \sum\limits_{i \notin {\Omega _A}\left( {{\gam }} \right)} v_i^A{F_{B_{\gam ,i}^S}}\left( {x_i^A} \right)  \nonumber  \\
    \le 	& \sum\limits_{i \in {\Omega _A}\left( {{\gam }} \right)} v^A_i \left(\frac{\frac{v^A_i}{\LdaA} - \frac{v^B_i}{\LdaB}}{\frac{v^A_i}{\LdaA}} + \frac{x^A_i \LdaA}{v^A_i}\right) + \sum\limits_{i \notin {\Omega _A}\left( {{\gam }} \right)} v_i^A \left( \frac{x^A_i \LdaA}{v^A_i} \right)  \nonumber  \\
    \le 	& \sum\limits_{i \in {\Omega _A}\left( {{\gam }} \right)} {\left( {\frac{{v_i^A}}{{\lambda _A^*}} - \frac{{v_i^B}}{{\lambda _B^*}}} \right)\lambda _A^*}  + \lambda _A^* X_A 	
                && \textrm{ (since } \sum \nolimits_{i=1}^n{ x^A_i} \le X_A \textrm{)} \nonumber  \\
    =& \sum\limits_{i \in {\Omega _A}\left( {{\gam }} \right)} v_i^A\left( 1 - \frac{{v_i^B \gam }}{2v^A_i} \right) +  \sum\limits_{i \notin {\Omega _A}\left( {{\gam }} \right)} {(v_i^A)^2\frac{1}{{2\gam  v^B_i}}} . \label{eq:proof_best_res_2}
\end{align}
Here, to obtain the last equality, we use~\eqref{eq:prelimi_XA} to rewrite $X_A$. Finally, from~\eqref{eq:proof_best_res_1} and \eqref{eq:proof_best_res_2}, we conclude that~\eqref{eq:lem_best_res} holds for any $\boldsymbol{x}^A$ and $\gam $. \qed

    \subsection{Proof of Lemma~\ref{lem:convergence}}
    \label{sec:appen_proof_lem_convergence}

Since the definition of $\FAn$ involves $\prob(A_{=0})$ (see \eqref{eq:A^n_Def}), we first look for an upper-bound of $\prob(A_{=0})$. For any $n$ and $\gam \in \Sn$, if $\Ostar  \neq \emptyset$, i.e., there exists $i$ such that $A^*_i = \AS$, then $\prob(A^*_i = 0) =0 $ due to the definition of $ A^S_{\gam ,i}$ (see~\eqref{As}); in this case, $\prob(A_{=0}) = \prod_{j \in [n]} \prob \left( A^*_j = 0 \right)=0$. On the other hand, if $\Ostar = \emptyset $, then $A^*_j= A^W_{\gam , j}$ for any $j \in [n]$; therefore,
\begin{equation}
    \prob(A_{=0}) =  \prod_{j \in [n]} \prob \left( A^*_j = 0 \right) = \prod_{j \in [n]}\left[{\left(\frac{v^B_j}{\LdaB} - \frac{v^A_j}{\LdaA} \right)} \middle/ {\frac{v^B_j}{\LdaB}}\right] = \prod_{j \in [n]} \left(1 - \frac{v^A_j}{v^B_j}\frac{\LdaB}{\LdaA} \right) \le  \left(1-\frac{\Lmin}{\Lmax} \frac{\wmin^2}{\wmax^2} \right)^n . \label{eq:Prop_all_zero}  
\end{equation}
Here, the last inequality comes directly from~\eqref{eq:bound_v^p_i} and Proposition~\ref{Prop:BoundLambda}. Recall the notation \mbox{$D:= \left(1-\frac{\Lmin}{\Lmax} \frac{\wmin^2}{\wmax^2} \right)$} and define $\tilde{C}_{0}\! :=\!  \frac{\ln(4)+1}{\ln(1/D)}\!>\!0$, we have \mbox{$\tilde{C}_0  \logone \! \ge  \! \log_D {\left(\frac{\varepsilon_1}{ 4} \right)}$}.\footnote{This is due to the fact that \mbox{$\tilde{C}_0 \cdot \logone =  \frac{1}{\ln(1/D)}\left(\ln(4)+1\right)\logone \ge  \frac{\ln(4/\varepsilon_1)}{\ln(1/D)} = \log_D {\left( \frac{\varepsilon_1}{4} \right)}$}; here, we have applied Lemma~\ref{lem:log_pre} (see~\ref{sec:appen_preliminary}) for $\hep := \varepsilon_1$ and $\hC:= 4$.} Therefore, for any \mbox{$n \!\ge\! \tilde{C}_0 \logone$} we have $n \ge \log_D{(\varepsilon_1/4)}$ and since $D<1$ we have:
\begin{equation}
\label{eq:proofB4_A=0}
    \prob(A_{=0}) \le D^n \le D^{\log_D(\varepsilon_1 /4)} = \varepsilon_1/ 4.
\end{equation}

Now, we look for an upper-bound of $\prob(A_{>0})$. For any $n$, define the constants $\epsilon_n:= \frac{\varepsilon_1 }{4} \frac{ \wmin}{n \wmax \Lmax}$ and \mbox{$\tau :=  \frac{1}{X_A} \left( \frac{\wmax}{n \wmin \Lmin}\frac{1}{\epsilon_n} + 1 \right)   = \frac{1}{X_A} \left[ \frac{4 \Lmax}{\varepsilon_1 \Lmin} \left( \frac{\wmax}{\wmin}\right)^2  + 1 \right]$}, we consider the following term for any $i \in [n]$:
\begin{align}
    & \prob\left( \left\{  A^*_i - \frac{A^*_i }{\sum \nolimits_{j \in [n]}{A^*_j}}X_A  > \epsilon_n \right\} \bigcap A_{>0} \right)  \nonumber \\
     & \le  \prob\left( \left\{  \left| A^*_i - \frac{A^*_i }{\sum \nolimits_{j \in [n]}{A^*_j}}X_A \right| > \epsilon_n \right\} \bigcap A_{>0} \right)  \nonumber \\
    %
    & \le  \prob\left(  A^*_i \left|\sum \nolimits _{j \in [n]}{A^*_j} - X_A \right| \!>\! \epsilon_n \sum \nolimits _{j \in [n]}{A^*_j} \right) \nonumber \\
    & =  \prob\left(  A^*_i \left|\sum \nolimits _{j \in [n]}{A^*_j} - X_A \right| \!>\! \epsilon_n X_A \! -\! \epsilon_n \left(X_A \! -\! \sum \nolimits _{j \in [n]}{A^*_j} \right) \right) \nonumber \\
    & \le  \prob\left(  A^*_i \left|\sum \nolimits _{j \in [n]}{A^*_j} - X_A \right| \!>\! \epsilon_n X_A \! -\! \epsilon_n \left|\sum \nolimits _{j \in [n]}{A^*_j}- X_A \right| \right) \nonumber \\
    & = \prob\left(  \left|\sum \nolimits_{j \in [n]}{A^*_j} -  X_A \right| \!>\!  \frac{\epsilon_n X_A}{ A^*_i \! + \!\epsilon_n}\right)  \nonumber \\
     & \le   \prob\left( \left| \sum \nolimits_{j \in [n]}{A^*_j} -  X_A  \right| \!>\!   \frac{\epsilon_n X_A}{\frac{\wmax}{n \wmin  \Lmin}  + \epsilon_n}\right) \nonumber \\
     & =  \prob\left( \left| \sum \nolimits_{j \in [n]}{A^*_j} -  X_A \right| \!>\! \frac{1}{\tau}  \right) \label{1.1}.
\end{align}
Here, the second-to-last inequality comes from the fact that for any $i \in [n]$, $A^*_i$ is upper-bounded by either ${v^A_i}/{\LdaA}$ or~${v^B_i}/{\LdaB}$ (see~\eqref{As} and~\eqref{Aw}), thus, it is bounded by ${\wmax}/({n \wmin \Lmin})$ (due to~\eqref{eq:bound_v^p_i} and Proposition~\ref{Prop:BoundLambda}).

Recall that $X_A= \mathbb{E}\left[\sum \limits _{j=1}^n{A^*_j}\right]$ (see Lemma~\ref{lem:Preliminary}-$(iii)$), we use the Hoeffding's inequality (see e.g., Theorem 2, \citet{hoeffding1963probability}) on the random variables $A^*_i, i \in [n]$ (bounded in $\left[ 0,  {\wmax}/(n \wmin \Lmin)\right]$) to~obtain
\begin{align}
    P\left(  \left| \sum \limits _{j \in [n]}{A^*_j} - X_A \right|  > \frac{1}{\tau} \right) 
    %
    %
	 \le &  2 \exp \left(\frac{-2\frac{1}{\tau^2}}{\sum\limits_{j \in [n]} {\left(  \frac{\wmax}{n \wmin \Lmin} \right)^2}} \right)  \nonumber \\
     = & 2 \exp \left[ \frac{-2n}{\tau ^2}  \left(\frac{ \Lmin \wmin }{\wmax }  \right)^2\right].\label{eq:from_Hoeffiding}
\end{align}
Now, we define \mbox{$\tilde{C}_1:= \frac{1}{2}\left(\frac{4}{X_A} \frac{\Lmax}{\Lmin}  \frac{\wmax^2}{\wmin^2} + \frac{1}{X_A} \right)^2(\ln8 + 1) \left(\frac{\wmax}{\wmin \Lmin} \right)^2$}; due to the definition of $\tau$, we have that\footnote{This is due to \mbox{$\tilde{C}_1 \! \cdot \!\frac{1}{\varepsilon_1^2} \logone \!=\! \frac{1}{2}\left[\frac{1}{X_A} \! \left( \frac{4}{\varepsilon_1} \frac{\Lmax}{\Lmin}  \frac{\wmax^2}{\wmin^2} \! +\!\frac{1}{\varepsilon_1} \right) \right]^2    \left( \ln(8) \! + \! 1 \right) \logone  \left(\frac{\wmax}{\wmin \Lmin} \right)^2 \ge \frac{\tau^2}{2}   \ln \left(\frac{8}{\varepsilon_1} \right)   \left(\frac{\wmax}{\wmin \Lmin} \right)^2$}; here, we have used Lemma~\ref{lem:log_pre} with $\hep:= \varepsilon_1$ and $\hC:=8$ and the fact that $1/\varepsilon \ge 1$.}
\mbox{$\tilde{C}_1 \cdot \frac{1}{\varepsilon_1^2} \logone \ge \frac{\tau^2}{2} \ln \left( \frac{8}{\varepsilon_1 } \right) \left(\frac{\wmax}{\wmin \Lmin} \right)^2$}; therefore, for any $n \ge \tilde{C}_1 {\varepsilon_1^{-2}} \logone$, we can deduce that \mbox{$\frac{2n}{\tau^2}\left(\frac{ \wmin \Lmin }{\wmax }  \right)^2 \!\ge\! \ln\left(\frac{8}{\varepsilon_1} \right) $} and thus,
\begin{equation}
 2 \exp \left[ \frac{-2n}{\tau^2}  \left(\frac{ \Lmin \wmin }{\wmax }  \right)^2\right] \le 2 \exp \left[-\ln\left(\frac{8}{\varepsilon_1} \right) \right] = \frac{\varepsilon_1 }{4}. \label{inequN1}
\end{equation}
Combining \eqref{1.1}, \eqref{eq:from_Hoeffiding} and \eqref{inequN1}, we deduce that
\begin{equation}
    \prob\left( \left\{ A^*_i -  \frac{A^*_i }{\sum \nolimits_{j \in [n]}{A^*_j}}X_A >  \epsilon_n \right\} \bigcap A_{>0} \right)  \le \frac{\varepsilon_1}{4}, \forall n \ge \tilde{C}_1 {\varepsilon_1^{-2}} \logone. \label{eq:bound_E_non_0}
\end{equation}

Finally, note that for any $n$, $i \in [n]$ and $x\ge 0$, we also have
\begin{align}
     & \prob \left( \left\{ \frac{A^*_i \cdot X_A }{\sum \nolimits_{j\in[n]} A^*_j} \le x \right\} \bigcap  A_{>0} \right) \nonumber \\
     = & \prob \left( \left\{ \frac{A^*_i X_A}{\sum \limits_{j\in[n]} A^*_j} \! \le \! x \right\} \bigcap \left\{   A_i^*\!- \! \frac{A^*_i  X_A}{\sum \limits_{j\in[n]} A^*_j}  \! \le \! \epsilon_n \right\} \bigcap A_{>0}  \right) \\       & \hspace{6cm}      \!+\!       \prob \left( \left\{ \frac{A^*_i \!\cdot\! X_A}{\sum \limits_{j\in[n]} A^*_j} \le x \right\} \bigcap \left\{  A_i^* \! - \! \frac{A^*_i  X_A}{\sum \limits_{j\in[n]} A^*_j} \! > \! \epsilon_n \right\} \bigcap A_{>0} \right)  \nonumber \\
     \le & \prob \left(\{A^*_i \le x\! +\! \epsilon_n \}   \right) +  \prob \left( \left\{  A_i^* \! - \! \frac{A^*_i  X_A}{\sum \limits_{j\in[n]} A^*_j} \! > \! \epsilon_n \right\} \bigcap A_{>0} \right).  \label{eq:proof_converge_inter} 
\end{align}

Therefore, define $C_1 := \max\{\tilde{C_0}, \tilde{C_1} \}$, for any $n \ge  C_1 \varepsilon_1^{-2} \logone$, from \eqref{eq:A^n_Def}, we have 
\begin{align} 
    & {\FAn}\left( x \right) - {\FA}\left( x \right)  \nonumber\\
    = &  \prob\left( A_{=0} \right) + \prob \left( \left\{ \frac{A^*_i \cdot X_A }{\sum \nolimits_{j\in[n]} A^*_j} \le x \right\} \bigcap  A_{>0} \right) - {\FA}\left( x \right) \nonumber\\
     \le & \frac{\varepsilon_1}{4} + \prob \left(\{A^*_i \le x\! +\! \epsilon_n \}   \right) +  \prob \left( \left\{  A_i^* \! - \! \frac{A^*_i  X_A}{\sum \limits_{j\in[n]} A^*_j} \! > \! \epsilon_n \right\} \bigcap A_{>0} \right) - {\FA}\left( x \right) && (\textrm{due to } \eqref{eq:proofB4_A=0} \textrm{ and } \eqref{eq:proof_converge_inter}) \nonumber\\ 
     \le &  \frac{\varepsilon_1}{4}  + \FA(x+ \epsilon_n) + \frac{\varepsilon_1}{4} - \FA(x)   && (\textrm{due to } \eqref{eq:bound_E_non_0}) . \label{eq:final_lem_converge}
\end{align}

The final step is to bound the term $\FA(x + \epsilon_n) - \FA(x)$; we present this as the following lemma.
\begin{lemma}
\label{remark_prepare}
     For any $\epsilon >0$, $n>0$, $i \in [n]$ and $x \in [0,\infty)$, we have ${\FA}\left( {x + {\epsilon}} \right)  - {\FA}\left( x \right)  \le \frac{\epsilon \LdaB}{v^B_i}$.
\end{lemma}

\begin{proof}

If $i\in \Ostar$, then $A^*_i = \AS$ and 
\begin{equation}
\label{inequaFAS}
F_{\AS}(x+ \epsilon) - F_{\AS}(x) = \left\{ \begin{array}{l}
\frac{(x+\epsilon)\LdaB}{v^B_i} - \frac{x \LdaB}{v^B_i} = \frac{\epsilon \LdaB}{v^B_i}, \textrm{ if } 0 \le x < \frac{v^B_i}{\LdaB} - \epsilon \\
    1 - \frac{x\LdaB}{v^B_i} \le \frac{\epsilon \LdaB}{v^B_i}, \qquad \quad \textrm{ if } \frac{v^B_i}{\LdaB} - \epsilon \le x \le \frac{v^B_i}{\LdaB}\\
1 - 1 \le \frac{\epsilon v^B_i}{\LdaB}, \qquad \qquad \quad \textrm{ if } x > \frac{v^B_i}{\LdaB}\\
\end{array} \right..
\end{equation}
On the other hand, if $i\notin \Ostar$, then $A^*_i = \AW$ and we have 
\begin{equation}
\label{inequaFAW}
F_{\AW}(x+ \epsilon) - F_{\AW}(x) = \left\{ \begin{array}{l}
\frac{(x+\epsilon)\LdaB}{v^B_i} - \frac{x \LdaB}{v^B_i} = \frac{\epsilon \LdaB}{v^B_i}, \quad \textrm{ if } 0 \le x < \frac{v^A_i}{\LdaA} - \epsilon\\
1 - \frac{\frac{v^B_i}{\LdaB} - \frac{v^A_i}{\LdaA}}{ \frac{v^B_i}{\LdaB}}-\frac{x\LdaB}{v^B_i} \le \frac{\epsilon \LdaB}{v^B_i},  \textrm{ if } \frac{v^A_i}{\LdaA} - \epsilon \le x \le \frac{v^A_i}{\LdaA}\\
1 - 1 \le \frac{\epsilon v^B_i}{\LdaB}, \quad \qquad \qquad \quad \textrm{ if } x > \frac{v^A_i}{\LdaA}\\
\end{array} \right.. 
\end{equation} \qed
\end{proof}

Combine this lemma with ~\eqref{eq:final_lem_converge} and recall the definition of $\epsilon_n$ (which induces that \mbox{$\epsilon_n \LdaB / v^B_i \le \varepsilon_1 /2$}), we conclude that $\FAn(x)\! - \! \FA(x) \le \varepsilon_1$ for any $n \ge C_1 \varepsilon_1^{-2} \logone$. Similarly, for \mbox{$n \ge C_1 \varepsilon_1^{-2} \logone$} and $i \in [n]$, we can deduce that \mbox{${\FAn}\left( x \right)\! -\! {\FA}\left( x \right) \! \ge\! \!-\varepsilon_1$} for any $x \in [0, \infty)$. We conclude that for any $n \ge C_1 \varepsilon_1^{-2} \logone$, \mbox{$\mathop {\sup }\limits_{ x \in [0, \infty) } \left|\FAn(x) - \FA(x)\right| \le \varepsilon_1$}. The inequality \mbox{$\mathop {\sup }\limits_{ x \in [0, \infty) } \left|\FBn(x) \! - \! \FB(x)\right| \le \varepsilon_1$} can be proved in a similar way. \qed
%
%
    \subsection{Proof of Lemma~\ref{lem:portmanteau}}
    \label{sec:appen_proof_lem_portmanteau}
In this proof, we will use the notation $\Ex f(X):= \int_{0}^\infty {f(z) \de F_Z(x)} $ and $\Ex_{\mathcal{I}} f(X):= \int_{\mathcal{I}} {f(z) \de F_Z(x)}$ for any function $f$, random variable $Z$ and interval $\mathcal{I}$. To simplify the notation, let us define $M:= \frac{\Lmax}{\Lmin} \frac{\wmax^2}{\wmin^2}$ and we denote by $\mathcal{I}_i$ the interval $\left[ 0,{v^B_i}/{\LdaB} \right]$. For any $\varepsilon_2 \in (0,1]$, we define $\delta_2:=\frac{\varepsilon_2}{6 + 2M}$. We first consider the case where $i \in \Ostar$, i.e., $B^*_i = \BW$. Note that \mbox{$\FAn(x) = \FA(x) = 1$} for any $x \ge 2X_B$ (see Lemma~\ref{lem:Preliminary}-$(iv)$); the left-hand-side of~\eqref{eq:portmanteau_main} can be rewritten as follows.
\begin{align}
    & \left|\Ex \FB(A^n_i) - \Ex \FB(A^*_i)  \right|  \nonumber  \\ 
    = & \left| {\int_{[0,2X_B]}  { \FB \left( x \right)}  \de {\FAn}\left( x \right) - \int_{[0,2X_B]}  {\FB \left( x \right)}  \de {\FA}\left( x \right)} \right| \nonumber \\
    \le & \left| \Ex_{[0, v^B_i/ \LdaB]} F_{\BW} (A^n_i) - \Ex_{[0, v^B_i/ \LdaB]} F_{\BW} (A^*_i) \right|  +  \left| {\int \limits_{\left({v^B_i}/{\LdaB}, 2X_B \right]}   \de {\FAn}\left( x \right)  \! -\! \int \limits_{\left({v^B_i}/{\LdaB}, 2X_B \right]}   \de {\FA}\left( x \right)} \right| \nonumber \\
    = & \left| \Ex_{\mathcal{I}_i} F_{\BW} (A^n_i) - \Ex_{\mathcal{I}_i} F_{\BW} (A^*_i) \right|  +  \left|\FAn(2X_B) - \FAn(v^B_i/\LdaB) -  \FA(2X_B) + \FA(v^B_i/\LdaB) \right| \nonumber \\
    \le & \left| \Ex_{\mathcal{I}_i} F_{\BW} (A^n_i) - \Ex_{\mathcal{I}_i} F_{\BW} (A^*_i) \right|  \! + 2 \sup \limits_{x \in [0, \infty)} \left| \FAn(x) - \FA(x) \right|. \label{eq:portmanteau_case1}
\end{align}
We now focus on bounding the first term in~\eqref{eq:portmanteau_case1}. Let us define $K := \big \lceil  \frac{M}{\delta_2} \big \rceil$ and $K + 1$ points ${x}_j$ such that ${x_0} := 0$ and ${x}_j := {x}_{j-1} +  \frac{v^B_i }{\LdaB  K}, \forall j \in [K]$. In other words, we have the partitions \mbox{$\mathcal{I}_i = \bigcup\nolimits_{j = 1}^{K} {P_j}$} where we denote by $P_1$ the interval $[{x}_0,{x}_1]$ and by $P_j$ the interval $({x}_{j-1}, {x}_{j}]$ for $j = 2, \ldots, K$. For any $ x, x^{\prime} \in P_j, \forall j \in [K]$, from the definition of $\BW$, we have 
\begin{equation}
    |F_{\BW}(x)- F_{\BW}(x^{\prime}) | =  \frac{\LdaA}{v^A_i}|x - x^{\prime}| \le \frac{\LdaA}{v^A_i} \cdot \frac{v^B_i}{\LdaB} \cdot \frac{1}{K}  \le \frac{\Lmax n \wmax}{\wmin} \cdot \frac{\wmax}{n \wmin \Lmin} \cdot \frac{1}{K}  = \frac{M}{K} \le \delta_2 . \label{eq:port_delta_2}    
\end{equation}
Now, we define the function $g(x)\!:=\! \sum \limits_{j=1}^{K} {F_{\BW}(x_j) \boldsymbol{1}_{P_j} (x)}$. Here, $\boldsymbol{1}_{P_j}$ is the indicator function of the set ${P_j}$. From this definition and Inequality~\eqref{eq:port_delta_2}, we trivially have \mbox{$| F_{\BW}(x) - g(x) | \le \delta_2$}, $\forall x \in \mathcal{I}_i$. Therefore, 
\begin{align}
& \left| \Ex_{\mathcal{I}_i} F_{\BW} (A^n_i) - \Ex_{\mathcal{I}_i} g  (A^n_i)  \right| \le  \int _{I_i} \left| F_{\BW}(x) - g(x) \right| \de \FAn(x) \le \int_{\mathcal{I}_i} \delta_2 \de \FAn(x) \le \delta_2, \label{eq:triangle1}\\
& \left| \Ex_{\mathcal{I}_i} F_{\BW} (A^*_i) - \Ex_{\mathcal{I}_i} g  (A^*_i)  \right|  \le  \int _{I_i} \left| F_{\BW}(x) - g(x) \right| \de \FA(x) \le \int_{\mathcal{I}_i} \delta_2 \de \FA(x) \le \delta_2 .\label{eq:triangle2}
\end{align}
Now, we note that for any $j \in [K]$, $F_{\BW}(x_j) = \sum \limits_{m=0}^{j} \left[{F_{\BW}(x_m) -  F_{\BW}(x_{m-1})} \right]$; here, for the sake of notation, we denote by $x_{-1}$ an arbitrary negative number (that is $F_{\BW}(x_{-1}) = 0$). Using this, we have:
\begin{align}
     &\left|\Ex_{\mathcal{I}_i} g \left(A^n_i \right) - \Ex_{\mathcal{I}_i} g \left(A^*_i \right) \right|\nonumber \\
    = & \left| \sum \limits_{j=1}^K F_{\BW}(x_j) \left[{  \Ex_{\mathcal{I}_i} \boldsymbol{1} _{P_j}\left( A^n_i\right)} -  {\Ex_{\mathcal{I}_i} \boldsymbol{1} _{P_j}\left( A^*_i\right)}  \right] \right| \nonumber \\ %
    = & \left| \sum \limits_{j=1}^K  F_{\BW}(x_j) \left[  \prob\left(A^n_i \in P_j \right)\!- \! \prob\left(A^*_i \in P_j \right) \right] \right| \nonumber \\
    = & \left| \sum \limits_{j=1}^K \left(  \sum \limits_{m=0}^{j}
    \left[   {F_{\BW}(x_m) -  F_{\BW}(x_{m-1})} \right]   \left[\prob\left(A^n_i \in P_j \right)\!- \! \prob\left(A^*_i \in P_j \right) \right]  \right)      \right| \nonumber \\
    \le & \left| \left[  F_{\BW}(x_0) -  F_{\BW}(x_{-1}) \right] \sum \limits_{j=1}^K  \left[ \prob\left(A^n_i \in P_j \right)\!- \! \prob \left(A^*_i \in P_j \right)  \right]  \right|   \nonumber \\
        &  \qquad + \left| \sum \limits_{m=1}^K \left( \left[  F_{\BW}(x_m) -  F_{\BW}(x_{m-1}) \right] \sum \limits_{j=m}^K  \left[ \prob\left(A^n_i \in P_j \right)\!- \! \prob\left(A^*_i \in P_j \right)  \right] \right) \right|. \label{eq:portmanteau_bridge}
\end{align}
Note that \mbox{$\prob\left(A^n_i \in P_j \right)\!- \! \prob \left(A^*_i \in P_j \right) \!=\!  \FAn(x_j) \!-\!\FAn(x_{j-1})\! -\! \FA(x_j)\! +\! \FA(x_{j-1}) $}.\footnote{For any $j \! \ge \! 2$, this is trivially since $P_j \!:= \! (x_{j-1}, x_j]$. For $P_1 = [0,x_1]$,~we have~that \mbox{$\prob\left(A^n_i \in P_1 \right)\!- \! \prob \left(A^*_i \in P_1 \right) =   \prob(A^n_i \in (0, x_1] ) \!- \!\prob(A^*_i \in (0, x_1] )  \!+\! \prob(A^n_i = 0)\! - \!\prob(A^*_i = 0) $}; moreover, due to Lemma~\ref{lem:continuity_Ani_and_Bni}-$(i)$, we also note that \mbox{$\prob(A^n_i \!= \!0) \!=\! \prob(A^*_i\!=\! 0)$}.}
Moreover, due to the fact that \mbox{$x_0 =0$} and $F_{\BW}(x_{-1}) = 0$, we can rewrite the first term in \eqref{eq:portmanteau_bridge} as follows:
\begin{align}
    & \left| \left[  F_{\BW}(x_0) -  F_{\BW}(x_{-1}) \right] \sum \limits_{j=1}^K  \left[ \prob\left(A^n_i \in P_j \right)\!- \! \prob \left(A^*_i \in P_j \right)  \right]  \right| \nonumber \\
    = & \left| F_{\BW}(0) \cdot \left[ \sum \limits_{j=1}^K  \left(\FAn(x_j) - \FAn(x_{j-1}) - \FA(x_j) + \FA(x_{j-1}) \right)   \right]  \right| \nonumber \\
    = & \left| F_{\BW}(0) \cdot \left[\FAn(x_K) - \FAn(x_0) - \FA(x_K) + \FA(x_{0})   \right]  \right| \nonumber \\
    \le & F_{\BW}(0) \cdot 2 \sup \limits_{x \in [0, \infty)} {\left| \FAn(x) - \FA(x) \right| } \nonumber \\
    \le & 2 \sup \limits_{x \in [0, \infty)} {\left| \FAn(x) - \FA(x) \right| } . \label{eq:portmanteau_final1}
\end{align}
Now, recall that $x_m = x_{m-1} +  {v^B_i}/{(\LdaB \cdot K)}, \forall m \in [K]$, by the definition of $F_{\BW}$, we deduce that \mbox{$ F_{\BW}(x_m) -  F_{\BW}(x_{m-1}) =  \frac{v^B_i}{\LdaB  K}\frac{\LdaA}{v^A_i} \le  \frac{\Lmax}{\Lmin} \frac{\wmax^2}{\wmin^2} \frac{1}{K}= \frac{M}{K}, \forall m \in [K] $}. Therefore, the second term in \eqref{eq:portmanteau_bridge}~is
\begin{align}
     &\left| \sum \limits_{m=1}^K \left( \left[  F_{\BW}(x_m) -  F_{\BW}(x_{m-1}) \right] \sum \limits_{j=m}^K  \left[ \prob\left(A^n_i \in P_j \right)\!- \! \prob\left(A^*_i \in P_j \right)  \right] \right) \right| \nonumber \\
     = & \left| \sum \limits_{m=1}^K \left(  \left[  F_{\BW}(x_m) -  F_{\BW}(x_{m-1}) \right] \sum \limits_{j=m}^K  \left[\FAn(x_j) - \FAn(x_{j-1}) - \FA(x_j) + \FA(x_{j-1})   \right]  \right) \right| \nonumber \\
     = & \left| \sum \limits_{m=1}^K \left(  \left[  F_{\BW}(x_m) -  F_{\BW}(x_{m-1}) \right] \left[\FAn(x_K) - \FAn(x_{m-1}) - \FA(x_K) + \FA(x_{m-1})   \right]  \right) \right| \nonumber \\
     \le & \sum \limits_{m=1}^K \left( F_{\BW}(x_m) -  F_{\BW}(x_{m-1}) \right) \cdot  2 \sup \limits_{x \in [0, \infty)} {\left| \FAn(x) - \FA(x) \right| } \nonumber \\
     \le & \sum \limits_{m=1}^K  \frac{M}{K} \cdot  2 \sup \limits_{x \in [0, \infty)} {\left| \FAn(x) - \FA(x) \right| } \nonumber \\
     = & 2 M \sup \limits_{x \in [0, \infty)} {\left| \FAn(x) - \FA(x) \right| }. \label{eq:portmanteau_final2}
\end{align}
Inject~\eqref{eq:portmanteau_final1} and~\eqref{eq:portmanteau_final2} into~\eqref{eq:portmanteau_bridge}, we obtain that
\begin{equation}
    \left|\Ex_{\mathcal{I}_i} g \left(A^n_i \right) - \Ex_{\mathcal{I}_i} g \left(A^*_i \right) \right| \le \left( 2 + 2 M \right) \sup \limits_{x \in [0, \infty)} {\left| \FAn(x) - \FA(x) \right| }. \label{eq:triangle3}
\end{equation}
Apply the triangle inequality and combine~\eqref{eq:triangle1},~\eqref{eq:triangle2}, \eqref{eq:triangle3}, we have that:
\begin{equation*}
    \left| \Ex_{\mathcal{I}_i} F_{\BW} (A^n_i) - \Ex_{\mathcal{I}_i} F_{\BW} (A^*_i) \right| \le 2 \delta_2 + \left( 2 + 2M \right) \sup \limits_{x \in [0, \infty)} {\left| \FAn(x) - \FA(x) \right| }.
\end{equation*}
From this and~\eqref{eq:portmanteau_case1}, we obtain that 
\begin{equation}
    |\Ex \FB(A^n_i) - \Ex \FB(A^*_i) | \le 2 \delta_2 + \left( 4 + 2M \right) \sup \limits_{x \in [0, \infty)} {\left| \FAn(x) - \FA(x) \right| }. \label{eq:portmanteau_the_end}
\end{equation}
Recall the constant $C_1$ indicated in Lemma~\ref{lem:convergence}, we define \mbox{$C_2: = C_1 \cdot \left(6 + 2 M \right)^2 \left[\ln(6 + 2 M) + 1 \right]$} (note that $C_2$ does not depend on $n$ nor $\varepsilon_2$) and deduce that $C_2 \varepsilon_2^{-2} \logtwo \ge C_1 \delta_2^{-2} \ln\left(\frac{1}{\min\{\delta_2, 1/\e\}} \right)$.\footnote{Apply Lemma~\ref{lem:log_pre}, \mbox{$C_2 \varepsilon_2^{-2} \logtwo \!= \!C_1 \left( \frac{6 \! +\! 2M}{\varepsilon_2} \right)^2 \left[\ln(6 \!+\! 2M) \!+\! 1 \right] \logtwo  \!\ge\! C_1  \left( \frac{6 \! + \! 2M}{\varepsilon_2} \right)^2 \ln\left(\frac{6 \! + \! 2M}{\varepsilon_2} \right)$}. Moreover, since $\frac{\varepsilon_2}{6+2M} =  \min\left\{\frac{\varepsilon_2}{6+2M}, \frac{1}{\e} \right\} = \min\left\{\delta_2, \frac{1}{\e} \right\} $ (due to the fact that $\delta_2 = \varepsilon_2 /(6+2M) < 1/\e$). Therefore, we have \mbox{$C_2 \varepsilon_2^{-2} \logtwo \ge C_1 \delta_2^{-2} \ln\left(\frac{1}{\min\{\delta_2, 1/\e\}} \right)$}.}
Take $\varepsilon_1:= \delta_2$, for any $n \ge C_2 \varepsilon_2^{-2} \logtwo$, we have \mbox{$n \!\ge\!  C_1 \varepsilon_1^{-2} \logone$} and by applying Lemma~\ref{lem:convergence}, we obtain \mbox{$\sup \limits_{x \in [0, \infty)} {\left| \FAn(x) - \FA(x) \right| } \le \varepsilon_1 = \delta_2 $} and thus by~\eqref{eq:portmanteau_the_end}, we have:
\begin{equation*}
    |\Ex \FB(A^n_i) - \Ex \FB(A^*_i) | \le 2 \delta_2 + \left( 4 + 2M \right) \delta_2 = (6 + 2M) \delta_2 = \varepsilon_2. 
\end{equation*}
This is exactly~\eqref{eq:portmanteau_main}. We can have a similar result in the case where $i \notin \Ostar$ (its proof is omitted here) and we conclude the proof of this lemma.
\qed

\section{Proof of Results in Section~\ref{sec:LotteryApproximation}}
\label{sec:Appendix_Lotte}

    \subsection{Proof of Lemma~\ref{lem:deltalemma}}
    \label{sec:appen_proof_lem_delta}
\deltalemma*
Fix $y^* \in [0,2X_B]$, we look for the condition on $n$ such that \mbox{$\int_{\X(y^*, \varepsilon)} \de \FAn(x) \le \delta \!+ \!\varepsilon$} holds. The condition corresponding to the inequality $\int_{\Y(x^*, \varepsilon)} \de \FBn(x) \le \varepsilon + \delta$ with $x^* \in [0,2X_B]$ can be proved similarly and thus is omitted in this section. 

First, we note that if $\X(y^*, \varepsilon)$ is empty, $\int_{\X(y^*, \varepsilon)} \de \FAn(x)=0$ and the result trivially holds. Now, let us assume that $\X(y^*, \varepsilon) \neq \emptyset$, we can write $\X(y^*, \varepsilon) = I_1 \bigcup I_2 \bigcup I_3$ with\footnote{Recall that by definition, $\beta_A(x,y^*) = \alpha$ if $x=y^*$, $\beta_A(x,y^*) = 0$ if $x<y^*$ and $\beta_A(x,y^*) = 1$ if $x>y^*$}
\begin{align*}
    I_1 &:= \{x \in [0,2X_B]: x=y^*, |\zeta_A(x,y^*) - \alpha| \ge \varepsilon \},  \\
    I_2 &:= \{x \in [0,2X_B]: x < y^*,  \zeta_A(x,y^*) \ge \varepsilon \}, \\
    I_3 &:= \{x \in [0,2X_B]: x>y^*,  1 -\zeta_A(x,y^*) \ge \varepsilon \}.
\end{align*}
It is trivial that $I_1$ is either an empty set or a singleton; on the other hand, due to the monotonicity of the CSF $\zeta_A$ (see $(C2)$, Definition~\ref{def:CSF_general}), $I_2$ and $I_3$ are either empty sets or half intervals. Moreover, for any arbitrary distribution $F$, we have that

\begin{equation*}
    \int_{ x \in I^{\prime}} \de F(x) = \left\{ \begin{array}{l}
    0 \text{ , if } I^{\prime} = \emptyset, \\
    F(a)  \text{ , if } I^{\prime} = \{a\}, \textrm{i.e., } I^{\prime} \textrm{ is a singleton},\\
    F(b) - F(a) \text{ , if } I^{\prime} = (a,b], \textrm{i.e., } I^{\prime} \textrm{ is a half interval}.
    \end{array} \right. \quad
\end{equation*}

Therefore, we can deduce that

\begin{equation*}
    \int_{\X(y^*, \varepsilon)} \de \FAn(x)  -  \int_{\X(y^*, \varepsilon)} \de \FA(x) =  \sum_{j=1}^3 \left( \int_{I_j} \de \FAn(x) - \int_{I_j} \de \FA(x) \right)     \le 5 \sup_{x \in [0,\infty)}{|\FAn(x) - \FA(x)| } 
\end{equation*}
Recall the constant $C_1$ indicated in Lemma~\ref{lem:convergence}, we define $L_0:= C_1 5^2 (\ln(5) + 1)$. Note that $L_0$ does not depend on the choice of $y^*$. Take $\varepsilon_1:= \varepsilon /5$, we can deduce that \mbox{$L_0 \varepsilon^{-2}  \logep  \ge C_1 \varepsilon_1^2 \logone$}.\footnote{Note that $\varepsilon_1 = \frac{\varepsilon}{5}$ and apply Lemma~\ref{lem:log_pre}, $L_0 \cdot \varepsilon^{-2} \logep = C_1 \left(\frac{5}{\varepsilon} \right)^2  \cdot (\ln(5)+1)\logep \ge  C_1 \cdot  \left(\frac{5}{\varepsilon} \right)^2 \ln \left(\frac{5}{\varepsilon} \right)$; moreover, $\frac{\varepsilon}{5} =  \min\{\frac{\varepsilon}{5}, \frac{1}{e}\}$ since $\varepsilon \le 1$; thus, we can rewrite $\ln\left(\frac{5}{\varepsilon} \right) =  \ln \left(\frac{1}{\min\{\varepsilon/5, 1/\e\}} \right) = \ln \left(\frac{1}{\min\{\varepsilon_1, 1/\e\}} \right)$.} Therefore, for any $n \ge  L_0 \varepsilon^{-2}  \logep$, we have $n \ge C_1 \varepsilon_1^2 \logone$ and by Lemma~\ref{lem:convergence}, $\sup \limits_{x \in [0,\infty)}{|\FAn(x) - \FA(x)| } \le \varepsilon_1 = \varepsilon/5 $. Hence, for any \mbox{$n \ge  L_0 \varepsilon^{-2}  \logep$} and $\delta \in \Del$, 
\begin{equation*}
    \int_{\X(y^*, \varepsilon)} \de \FAn(x)  \le  \int_{\X(y^*, \varepsilon)} \de \FA(x) + 5 \cdot \varepsilon/5 \le \delta +  \varepsilon.    
\end{equation*}

\qed


%
%
%

        \subsection{Proof of Theorem~\ref{theo:Lottery_generic_approx}}
        \label{sec:appen_proof_theoLottery}

\LotteTheo*
\begin{proof}

We first give the proof of Result~$(ii)$. For the sake of brevity, we only focus on \eqref{eq:lottery_theo_A}. The proof that \eqref{eq:lottery_theo_B} holds under the same condition can be done similarly and thus is omitted. Note that in this proof, we often use the Fubini's Theorem to exchange the order of the double integrals. 

Recall that $\boldsymbol{x}^A = (x^A_i)_{i \in [n]}$, by the definition of the payoff functions in $\LB(\zeta)$, \eqref{eq:lottery_theo_A} can be rewritten as
\begin{equation}
    \sum\limits_{i = 1}^n {\left( {{v^A_i}\int\limits_0^\infty  {\zeta_A\left( {x_i^A,y} \right) \de {\FBn}\left( y \right)} } \right)}  - \sum\limits_{i = 1}^n {\left( {{v^A_i}\int\limits_0^\infty  {\int\limits_0^\infty  {\zeta_A\left( {x,y} \right) \de {\FAn}\left( x \right) \de {\FBn}\left( y \right)} } } \right)}  \le 8\delta  + 13 \varepsilon. \label{eq:lottery_proof_rewrite_A}
\end{equation} 
We now prove that~\eqref{eq:lottery_proof_rewrite_A} holds under appropriate parameters values. To do this, we prepare two useful lemmas as follows. 

\begin{lemma}
\label{lem:proof_lottery_LBn_prepare1}
For any pair of CSFs $\zeta = (\zeta_A, \zeta_B)$, any $\varepsilon \in (0,1]$ and $x^* \in [0,2X_B]$, the following~results~hold:
    \begin{itemize}
        \item[$(i)$] For any $n$, $i \in [n]$ and $\delta \in \Del$,
        \begin{equation}
            \left| \int\limits_0^\infty  {\zeta_A\left( {x^*,y} \right) \de {\FB}\left( y \right)} - \int\limits_0^\infty  {\beta_A\left( {x^*,y} \right) \de {\FB}\left( y \right)} \right| \le \delta + \varepsilon \label{eq:lottery_delta_prepare1*}.
        \end{equation}
        \item[$(ii)$]  There exists a constant $L_1 >0$ such that for any $n \ge L_1 \varepsilon^{-2} \logep$, $i \in [n]$ and $\delta \in \Del$,
        \begin{equation}
            \left| \int\limits_0^\infty  {\zeta_A\left( {x^*,y} \right) \de {\FBn}\left( y \right)} - \int\limits_0^\infty  {\beta_A\left( {x^*,y} \right) \de {\FBn}\left( y \right)} \right| \le \delta + 2\varepsilon. \label{eq:lottery_delta_prepare1n}
        \end{equation}
    \end{itemize}
\end{lemma}
\begin{lemma}
\label{lem:proof_lottery_LBn_prepare2}
Given $\wmin, \wmax, X_A, X_B >0$ ($\wmin \le \wmax$, $X_A \le X_B$), there exists a constant $L_2>0$ such that for any $\varepsilon \in (0,1]$ and \mbox{$ n \ge L_2 \varepsilon^{-2} \logep$}, in any game $\LB(\zeta)$ with any \mbox{$\delta \in \Del$} and \mbox{$i \in [n]$}, we~have:
\begin{align}
    & \left|{\int\limits_0^\infty  {\zeta_A\left( {x ,y} \right) \de {\FBn}\left( y \right)}  - \int\limits_0^\infty  {\zeta_A\left( {x,y} \right) \de {\FB}\left( y \right)} }  \right|  \le  2\delta + 4\varepsilon , \forall x \ge 0, \label{eq:propo_appen_lottery_1} \\
    &\left|\int\limits_0^\infty  {\int\limits_0^\infty  {\zeta_A\left( {x,y} \right) \de {\FA}\left( x \right)}  \de {\FB}\left( y \right)} - \int\limits_0^\infty  {\int\limits_0^\infty  {\zeta_A\left( {x,y} \right) \de {\FA}\left( x \right)}  \de {\FBn}\left( y \right)} \right|  \le 2\delta + 3\varepsilon, \label{eq:propo_appen_lottery_2} \\
    & \left|\int\limits_0^\infty  {\int\limits_0^\infty  {\zeta_A\left( {x,y} \right) \de {\FBn}\left( y \right)}  \de {\FA}\left( x \right)} - \int\limits_0^\infty  {\int\limits_0^\infty  {\zeta_A\left( {x,y} \right) \de {\FBn}\left( y \right)}  \de {\FAn}\left( x \right)} \right|  \le  2\delta + 4\varepsilon. \label{eq:propo_appen_lottery_3}
\end{align}
\end{lemma}
Lemma~\ref{lem:proof_lottery_LBn_prepare1} states the relation between the first term appearing in the left-hand-side of~\eqref{eq:lottery_proof_rewrite_A} and the corresponding terms when we replace the CSF $\zeta$ by the Blotto functions $\beta$ and replace $\FBn$ by $\FB$. These relations are useful to connect the statement we want to prove and the results obtained in Section~\ref{sec:ApproximateBlotto}. A proof of Lemma~\ref{lem:proof_lottery_LBn_prepare1} is given in~\ref{sec:lemma_proof_lottery_LBn_prepare1}. On the other hand, Lemma~\ref{lem:proof_lottery_LBn_prepare2} indicates several useful inequalities involving the players' payoffs in the game $\LB$ (when they play according to the $\IU$ strategy or playing such that the marginals are $\FA, \FB$). Its proof is given in~\ref{sec:lem:proof_lottery_LBn_prepare2} that is based on Lemma~\ref{lem:proof_lottery_LBn_prepare1} and the convergence of the distributions $\FAn,\FBn$ toward $\FA,\FB$ (i.e., Lemma~\ref{lem:convergence}). 

We have another remark: for any $n$ and $i \in [n]$,
\begin{equation}
    \prob(A^*_i = B^*_i = x) = 0, \forall x \ge 0.    \label{eq:A^*=B^*}
\end{equation}
This can be trivially proved as follows: first, $\prob(A^*_i = B^*_i = x) = \prob(A^*_i = x) \prob(B^*_i =x)$ since they are independent; now, if $x > 0$, $\FA$ and $\FB$ are continuous at $x$ and thus \mbox{$\prob(A^*_i \!= \!x) \!=\! \prob(B^*_i \!=\!x)\! =\! 0$}; on the other hand, if $x =0$, in the case where $i \in \Ostar$, since $A^*_i = \AS$, we have $\prob(A^*_i = x) = 0 $, in the case where $i \notin \Ostar$, since $B^*_i = \BS$, we have $\prob(B^*_i = x) = 0 $.

Finally, use Lemma~\ref{lem:proof_lottery_LBn_prepare1} and \ref{lem:proof_lottery_LBn_prepare2} and take $L^* = \max\{L_1, L_2\}$, for any $n\ge L^* \varepsilon^{-2 } \logep$, \mbox{$\delta \in \Del$} and any pure strategy $\boldsymbol{x}^A$ of Player A, we~have:
\begin{align}
    & \sum\limits_{i = 1}^n {\left( {{v^A_i}\int\limits_0^\infty  {\zeta_A\left( {x_i^A,y} \right) \de {\FBn}\left( y \right)} } \right)}  \nonumber\\
    \le & \sum\limits_{i = 1}^n {\left( {{v^A_i}\int\limits_0^\infty  {\zeta_A\left( {x_i^A,y} \right) \de {\FB}\left( y \right)} } \right)}  + \sum\limits_{i = 1}^n { v^A_i (2\delta + 4\varepsilon)} & \text{ (due to \eqref{eq:propo_appen_lottery_1}) }  \nonumber\\
    = & \sum\limits_{i = 1}^n {\left( {{v^A_i}\int\limits_0^\infty  {\zeta_A\left( {x_i^A,y} \right) \de {\FB}\left( y \right)} } \right)}  + 2\delta + 4\varepsilon  & (\textrm{note that } \sum\nolimits_{i=1}^n{v^A_i} = 1 )\nonumber\\
    \le & \sum\limits_{i = 1}^n {\left( {{v^A_i}\int\limits_0^\infty  {\beta_A\left( {x_i^A,y} \right) \de {\FB}\left( y \right)} } \right)}  +  3\delta + 5\varepsilon  & ( \text{due to } \eqref{eq:lottery_delta_prepare1*})\nonumber \\
    = & \sum\limits_{i = 1}^n {\left[ {v^A_i}  \left(\alpha \prob(B^*_i =x^A_i) + \prob(B^*_i < x^A_i) \right) \right]}  +  3\delta + 5\varepsilon  \nonumber \\
    \le & \sum\limits_{i = 1}^n {{v_i^A}{\FB}\left( {x_i^A} \right)}  +  3\delta + 5\varepsilon  & (\textrm{since } \alpha \le 1)  \nonumber\\
    \le &  \sum\limits_{i = 1}^n {\left( {{v_i^A}\int\limits_0^\infty  {{\FB}(x) \de {\FA}(x)} } \right)} +   3\delta + 5\varepsilon       & (\text{due to Lemma~\ref{lem:best_response}}) \nonumber \\
    = &  \sum\limits_{i = 1}^n {\left( {v_i^A} \int\limits_0^\infty  \prob(B^*_i < x) \de {\FA}(x)  \right)} +   3\delta + 5\varepsilon       & \text{ (due to } \eqref{eq:A^*=B^*}) \nonumber \\
    \le & \sum\limits_{i = 1}^n {\left( {{v^A_i}\int\limits_0^\infty  {\int\limits_0^\infty  {{\beta_A}\left( {x,y} \right) \de {\FB}\left( y \right) \de {\FA}\left( x \right)} } } \right)} + 3\delta + 5\varepsilon   \nonumber \\
    \le & \sum\limits_{i = 1}^n {\left( {{v^A_i}\int\limits_0^\infty  {\int\limits_0^\infty  {{\zeta_A}\left( {x,y} \right) \de {\FB}\left( y \right) \de {\FA}\left( x \right)} } } \right)} + 4\delta + 6\varepsilon   	& \text{ (due to } \eqref{eq:lottery_delta_prepare1*}) \nonumber \\
    \le &  \sum\limits_{i = 1}^n {\left( {v^A_i}\int\limits_0^\infty  {\int\limits_0^\infty  {\zeta_A\left( {x,y} \right) \de {\FA}\left( x \right)}  \de {\FBn}\left( y \right)} \right)} +   6\delta + 9\varepsilon   & \text{ (due to \eqref{eq:propo_appen_lottery_2}) } \nonumber \\
    \le & \sum\limits_{i = 1}^n {\left( {v^A_i}\int\limits_0^\infty  {\int\limits_0^\infty  {\zeta_A \left( {x,y} \right) \de {\FBn}\left( y \right)} \de {\FAn}\left( x \right)}  \right)}+   8\delta + 13\varepsilon   & \text{ (due to \eqref{eq:propo_appen_lottery_3})}. \nonumber
\end{align}
Hence, we conclude that for $n \ge L^* \varepsilon^{-2} \logep$ \eqref{eq:lottery_proof_rewrite_A} holds and thus, \eqref{eq:lottery_theo_A} also holds.

To prove that Result~$(ii)$ implies Result~$(i)$, we can proceed similarly to the proof that Theorem~\ref{TheoMainBlotto}-$(ii)$ implies Theorem~\ref{TheoMainBlotto}-$(i)$ (see~\ref{sec:Appen_Proof_TheoBlotto}). We conclude this proof. \qed
\end{proof}
%
%
%
%


    \subsection{Proof of Lemma~\ref{lem:proof_lottery_LBn_prepare1}}
    \label{sec:lemma_proof_lottery_LBn_prepare1}

First, we prove~\eqref{eq:lottery_delta_prepare1*}. Note that $\FB(y)=1, \forall y > 2X_B$ (see Lemma~\ref{lem:Preliminary}-$(iv))$, for any $n$, $i \in [n]$ and $\delta \in \Del$, we have
\begin{align}
    & \left| \int\limits_0^\infty  {\zeta_A\left( {x^*,y} \right) \de {\FB}\left( y \right)} - \int\limits_0^\infty  {\beta_A\left( {x^*,y} \right) \de {\FB}\left( y \right)} \right| \nonumber\\
    \le & \int\nolimits_{\Y(x^*, \varepsilon )}{\left|\zeta_A(x^*,y) - \beta_A(x^*,y)\right|  \de {\FB}\left( y \right) } + \int\nolimits_{[0,\infty) \backslash \Y(x^*,\varepsilon)}{\left|\zeta_A(x^*,y) - \beta_A(x^*,y)\right|  \de {\FB}\left( y \right) } \nonumber\\
   = & \int\nolimits_{\Y(x^*, \varepsilon )}{\left|1 \!-\! \zeta_B(x^*,y) \!-\! 1 \!+\! \beta_B(x^*,y)\right|  \de {\FB}\left( y \right) } \!+\! \int\nolimits_{[0,2X_B] \backslash \Y(x^*,\varepsilon)}{\left|1 \!-\! \zeta_B(x^*,y) \!-\! 1 \!+\! \beta_B(x^*,y)\right|  \de {\FB}\left( y \right) } \nonumber\\
    = & \int\nolimits_{\Y(x^*, \varepsilon )}{\left|\zeta_B(x^*,y) - \beta_B(x^*,y)\right|  \de {\FB}\left( y \right) } + \int\nolimits_{[0,2X_B] \backslash \Y(x^*,\varepsilon)}{\left|\zeta_B(x^*,y) - \beta_B(x^*,y)\right|  \de {\FB}\left( y \right) } \nonumber\\
    \le & \int\nolimits_{\Y(x^*, \varepsilon )}{  \de {\FB}\left( y \right) } + \int\nolimits_{[0,2X_B] \backslash \Y(x^*,\varepsilon)}{\varepsilon  \de {\FB}\left( y \right) }  \nonumber \\
    \le & \delta  + \varepsilon. \label{eq:prepare1_first}
\end{align}
Here, the second-to-last inequality comes from the fact that $0 \le \zeta_B(x,y), \beta_B(x,y) \le 1$ for any $x,y$ and the definition of $\Y(x^*, \varepsilon)$ while the last inequality is due to the definition of $\Del$.

Now, in order to prove~\eqref{eq:lottery_delta_prepare1n}, we proceed similarly as in \eqref{eq:prepare1_first} to show that
\begin{align}
    & \left| \int\limits_0^\infty  {\zeta_A\left( {x^*,y} \right) \de {\FBn}\left( y \right)} - \int\limits_0^\infty  {\beta_A\left( {x^*,y} \right) \de {\FBn}\left( y \right)} \right| \nonumber\\
    \le & \int\nolimits_{\Y(x^*, \varepsilon )}{  \de {\FBn}\left( y \right) } + \int\nolimits_{[0,2X_B] \backslash \Y(x^*,\varepsilon)}{\varepsilon  \de {\FBn}\left( y \right) }  \nonumber \\
    \le & \int\nolimits_{\Y(x^*, \varepsilon )}{  \de {\FBn}\left( y \right) } + \varepsilon. \label{eq:LemC1_2}
\end{align}
Finally, by Lemma~\ref{lem:deltalemma}, for any $n \ge L_0 \varepsilon^{-2} \logep$ and $\delta \!\in\! \Del$, we have \mbox{$\int\nolimits_{\Y(x^*, \varepsilon )}{  \de {\FBn}\left( y \right) }\! \le\! \varepsilon \!+\! \delta $}. Combine this with~\eqref{eq:LemC1_2}, we conclude that \eqref{eq:lottery_delta_prepare1n} holds for any \mbox{$n \ge L_0 \varepsilon^{-2} \logep$ and $\delta \in \Del$}. Take $L_1:= L_0$, we conclude the proof.

\qed
%

    \subsection{Proof of Lemma~\ref{lem:proof_lottery_LBn_prepare2}}
    \label{sec:lem:proof_lottery_LBn_prepare2}
In this proof, we use the notation $\Ex h(X,y):= \int \nolimits_{0}^{\infty}h(x,y) \de F_{X}(x)$ and \mbox{$\Ex h(x,Y):= \int \nolimits_{0}^{\infty}h(x,y) \de F_{Y}(y)$} where $X, Y$ are arbitrary non-negative random variables and $h$ is any function. 

\underline{Proof of \eqref{eq:propo_appen_lottery_1}:} 
For any $i \in [n]$ and $x \ge 0$, we have
\begin{align}
    &\left|  {\int\limits_0^\infty  {\zeta_A\left( {x,y} \right) \de {\FBn}\left( y \right)}  - \int\limits_0^\infty  {\zeta_A\left( {x,y} \right) \de {\FB}\left( y \right)} }\right|  \nonumber \\
    \le &\left| \Ex\zeta_A(x,B^n_i) \! -\! \Ex\beta_A(x,B^n_i)\right|\! +\! \left|\Ex\beta_A(x,B^n_i) \!-\! \Ex \beta_A(x,B^*_i)\right|\! +\! \left|\Ex \beta_A(x,B^*_i) \!-\! \Ex\zeta_A(x,B^*_i) \right|. \label{eq:proof_propo_LBn_1}
\end{align}
We notice that upper-bounds of the first and third terms in the right-hand-side of~\eqref{eq:proof_propo_LBn_1} are given by~\eqref{eq:lottery_delta_prepare1n} and~\eqref{eq:lottery_delta_prepare1*} from Lemma~\ref{lem:proof_lottery_LBn_prepare1}. We focus on finding an upper-bound of the second term of~\eqref{eq:proof_propo_LBn_1}; to do this, we rewrite this term as follows.
\begin{align}
    &  \Ex \beta_A(x,B^n_i) = \int_{y<x} \de \FBn(y) + \alpha \prob(B^n_i = x) = \FBn (x )  - (1-\alpha)\prob(B^n_i = x),\label{eq:Ex_betaBn}\\
    \textrm{and~} & \Ex \beta_A(x,B^*_i) = \int_{y<x} \de \FB(y) + \alpha \prob(B^*_i = x) = \FB(x)  - (1-\alpha)\prob(B^*_i = x).\label{eq:Ex_betaB*}
\end{align}
If $\alpha =1$, we trivially have $  \left| \Ex \beta_A(x,B^n_i) - \Ex \beta_A(x,B^*_i) \right| = \left| \FBn(x) -  \FB(x) \right| $. In the following, we assume that $\alpha < 1$ and consider three~cases:

\emph{Case 1:} If $x =0$. From Lemma~\ref{lem:continuity_Ani_and_Bni}-$(i)$, we have $\prob (B^n_i = 0) = \prob (B^*_i = 0)$ and thus 
\begin{align*}
       \left| \Ex \beta_A(0,B^n_i) - \Ex \beta_A(0,B^*_i) \right|   =  \left|\int_{y<0} \de \FBn(y) - \int_{y<0} \de \FB(y)+ \alpha \prob(B^n_i = 0) - \alpha \prob(B^*_i =0)\right|   =  0  .
\end{align*}

\emph{Case 2:} If $x>0$, $\prob(B^*_i = x) =0$ by definition. On the other hand, from Results~$(ii)$ and $(iii)$ of Lemma~\ref{lem:continuity_Ani_and_Bni}, we have \mbox{$\prob(B^n_i = x) \le D^{n-1}$} where we define $D:=\left(1 - \frac{\Lmin}{\Lmax} \frac{\wmin^2}{\wmax^2} \right)$. Following~\eqref{eq:probAn=Bn}, for any $n \ge C_0 \logep$ (here, $C_0$ is defined as in~\ref{sec:appen_proof_lem:SufCon}), we have $D^{n-1 } \le \frac{\varepsilon}{2(1-\alpha)}$. Therefore, for any \mbox{$n \ge C_0 \logep$}, we have
\begin{align*}
        & \left|   \Ex \beta_A(x,B^n_i) - \Ex \beta_A(x,B^*_i)  \right|  \\
    \le & |\FBn(x) -  \FB(x)| + (1-\alpha)\left|\prob(B^n_i = x)\right| \qquad (\textrm{due to }\eqref{eq:Ex_betaBn}-\eqref{eq:Ex_betaB*})\\ 
    \le & \sup \limits_{x \in [0,\infty)}{|\FBn(x) - \FB(x)| }   + (1-\alpha)\frac{\varepsilon}{2 (1-\alpha)} \\
     = & \sup \limits_{x \in [0,\infty)}{|\FBn(x) - \FB(x)| }   + \frac{\varepsilon}{2}.
\end{align*}
  
 %

In conclusion, \mbox{$ \left|   \Ex \beta_A(x,B^n_i) - \Ex \beta_A(x,B^*_i)  \right|   \le \sup \limits_{x \in [0,\infty)}{|\FBn(x) - \FB(x)| } + \varepsilon/2$} for any $x \ge 0$, \mbox{$\alpha \in [0,1]$} and \mbox{$n \ge C_0 \logep$}. Now, let us define $C^{\prime}_1 = C_1 \cdot 4 (\ln(2) + 1)$ (where $C_1$ is indicated in~Lemma~\ref{lem:convergence}); take $\varepsilon_1:= \varepsilon/2$, we have \mbox{$C^{\prime}_1 \varepsilon^{-2} \logep \ge C_1 \varepsilon_1^{-2} \logone $}. Thus, for any $n \ge C^{\prime}_1 \varepsilon^{-2} \logep$, we have $n \ge  C_1 \varepsilon_1^{-2} \logone $ and apply Lemma~\ref{lem:convergence}, we have $\sup \limits_{x \in [0,\infty)}{|\FBn(x) - \FB(x)| } \le \varepsilon_1 = \varepsilon/2$.

We deduce that for any $x \ge 0$, for any $n \ge \max\{C_0, C^{\prime}_1\} \varepsilon^{-2} \logep$ and $i \in [n]$, we have:
\begin{equation}
    \left|   \Ex \beta_A(x,B^n_i) - \Ex \beta_A(x,B^*_i)  \right|   \le   \varepsilon/2 + \varepsilon/2  = \varepsilon. \label{eq:proof_C5}
\end{equation}

Finally, apply Lemma~\ref{lem:proof_lottery_LBn_prepare1} to \eqref{eq:proof_propo_LBn_1} to bounds the first and third term of its right-hand-side, use~\eqref{eq:proof_C5} to bound its second-term and take $L_{\eqref{eq:propo_appen_lottery_1}} = \max\{L_1, C_0, C^{\prime}_1\}$, we deduce that for any \mbox{$n \ge L_{\eqref{eq:propo_appen_lottery_1}}\varepsilon^{-2} \logep $} and $\delta \in \Del$, 
\begin{equation*}
    \left|  {\int\limits_0^\infty  {\zeta_A\left( {x,y} \right) \de {\FBn}\left( y \right)}  - \int\limits_0^\infty  {\zeta_A\left( {x,y} \right) \de {\FB}\left( y \right)} }\right| \le (\delta+ 2\varepsilon) + \varepsilon + (\delta + \varepsilon) = 2\delta + 4 \varepsilon.
\end{equation*}
%
%
%
\underline{Proof of \eqref{eq:propo_appen_lottery_2}:} 
To prove this inequality, we note that similar to the proof of~\eqref{eq:lottery_delta_prepare1*} in Lemma~\ref{lem:proof_lottery_LBn_prepare1} (by replacing $\FB$ by $\FA$ and replacing $\zeta_A(x^*,y), \beta_A(x^*,y)$ by $\zeta_A(x,y^*), \beta_A(x,y^*)$), we can prove that for any $n$, $i \in [n]$, $\delta \in \Del$ and $y^* \in [0,2X_B]$, the following inequality~holds
        \begin{equation}
            \left| \int\limits_0^\infty  {\zeta_A\left( {x,y^*} \right) \de {\FA}\left( x \right)} - \int\limits_0^\infty  {\beta_A\left( {x,y^*} \right) \de {\FA}\left( x \right)} \right| \le \delta + \varepsilon.\label{eq:zeta_A*}
        \end{equation}
Using this, we have
\begin{align*}
    & \left|\int\limits_0^\infty  {\int\limits_0^\infty  {\zeta_A\left( {x,y} \right) \de {\FA}\left( x \right)}  \de {\FB}\left( y \right)} - \int\limits_0^\infty  {\int\limits_0^\infty  {\zeta_A\left( {x,y} \right) \de {\FA}\left( x \right)}  \de {\FBn}\left( y \right)} \right| \\
    \le & \int \limits_{0}^{\infty} \left| \int \limits_0^\infty \zeta_A(x,y) \de \FA(x) \!- \! \int \limits_0^\infty \beta_A(x,y) \de \FA(x) \right|  \de \FB(y) \!\\
        & \qquad \qquad +\! \left|  \int \limits_{0}^{\infty} \int \limits_0^\infty \beta_A(x,y) \de \FA(x) \de \FB(y)\! -\! \int \limits_{0}^{\infty} \int\limits_0^\infty \beta_A(x,y) \de \FA(x) \de \FBn(y) \right|\\
            &  \qquad \qquad \qquad \qquad + \int\limits \limits_{0}^{\infty} \left| \int \limits_0^\infty \beta_A(x,y) \de \FA(x) - \int \limits_0^\infty \zeta_A(x,y) \de \FA(x)\right|  \de \FBn(y) \\
    \le & \int \nolimits_{0}^{\infty} (\delta \!+\! \varepsilon)  \de \FB(y)  \!+\! \left|  \int \nolimits_{0}^{\infty} \Ex \beta_A(x,B^*_i) \de \FA(x) \! -\! \int \nolimits_{0}^{\infty} \Ex \beta_A(x,B^n_i) \de \FA(x) \right| \! +\! \int \nolimits_{0}^{\infty} (\delta \! +\! \varepsilon)  \de \FBn(y) \\
    \le & 2\delta \!+\! 2\varepsilon  \!+\!   \int \nolimits_{0}^{\infty} \left|\Ex \beta_A(x,B^n_i)- \Ex \beta_A(x,B^*_i) \right|\de \FA(x).
\end{align*}
Finally, take $L_\eqref{eq:propo_appen_lottery_2} =  \max\{C_0, C^{\prime}_1\}$ and apply~\eqref{eq:proof_C5}, we deduce that for any \mbox{$n \ge L_\eqref{eq:propo_appen_lottery_2} \varepsilon^{-2} \logep$}, \eqref{eq:propo_appen_lottery_2}~holds.

\underline{Proof of \eqref{eq:propo_appen_lottery_3}}
To prove this inequality, we note that similar to the proof of~\eqref{eq:lottery_delta_prepare1n} in Lemma~\ref{lem:proof_lottery_LBn_prepare1} (by replacing $\FBn$ by $\FAn$ and replacing $\zeta_A(x^*,y), \beta_A(x^*,y)$ by $\zeta_A(x,y^*), \beta_A(x,y^*)$), we can prove that for \mbox{$n \ge L_1 \varepsilon^{-2} \logep$}, $i \in [n]$ and $\delta \in \Del$,
 \begin{equation}
 \label{eq:last}
            \left| \int\limits_0^\infty  {\zeta_A\left( {x,y^*} \right)\de {\FAn}\left( x \right)} - \int\limits_0^\infty  {\beta_A\left( {x,y^*} \right) \de {\FAn}\left( x \right)} \right| \le \delta + 2\varepsilon. 
\end{equation}
Now, similar to the proof leading to~\eqref{eq:proof_C5}, we can prove that
\begin{equation} 
\label{eq:mid}
    \left|   \Ex \beta_A(A^*_i,y) - \Ex \beta_A(A^n_i,y)  \right| \le   \varepsilon,
\end{equation}
for any \mbox{$n \ge \max\{C_0, C^{\prime}_1\} \varepsilon^{-2} \logep$}, $i \in [n]$ and $y \ge 0$.

Finally, take $L_\eqref{eq:propo_appen_lottery_3} = \max\{L_1, C_0, C^{\prime}_1\}$, for any $n \ge L_\eqref{eq:propo_appen_lottery_3} \varepsilon^{-2} \logep$, $i \in [n]$ and \mbox{$\delta \in \Del$}, we have
\begin{align*}
    & \left|\int\limits_0^\infty  {\int\limits_0^\infty  {\zeta_A\left( {x,y} \right) \de {\FBn}\left( y \right)}  \de {\FA}\left( x \right)} - \int\limits_0^\infty  {\int\limits_0^\infty  {\zeta_A\left( {x,y} \right) \de {\FBn}\left( y \right)}  \de {\FAn}\left( x \right)} \right| \\
     \le & \int \limits_{0}^{\infty} \left| \int \limits_0^\infty \zeta_A(x,y) \de \FA(x) \!- \! \int \limits_0^\infty \beta_A(x,y) \de \FA(x) \right|  \de \FBn(y) \!\\
            & \qquad \qquad +\! \left|  \int \limits_{0}^{\infty} \int \limits_0^\infty \beta_A(x,y) \de \FA(x) \de \FBn(y) \! -\! \int \limits_{0}^{\infty} \int\limits_0^\infty \beta_A(x,y) \de \FAn(x) \de \FBn(y)\right|  \\
                &  \qquad \qquad \qquad \qquad + \int\limits \limits_{0}^{\infty} \left| \int \limits_0^\infty \beta_A(x,y) \de \FAn(x) - \int \limits_0^\infty \zeta_A(x,y) \de \FAn(x)\right|  \de \FBn(y) \\
    \le  & \int \limits_{0}^{\infty} (\delta  \!+\! \varepsilon)  \de \FA(x)  \!+\! \int \limits_0^\infty  \left| \Ex \beta_A(A^*_i,y) \!-\! \Ex \beta_A(A^n_i,y)   \right| \de  {\FBn}(y) \!+ \! \int \limits_{0}^{\infty} (\delta  \!+\! 2\varepsilon)  \de \FAn(x)  &  \textrm{(due to } \eqref{eq:zeta_A*} \textrm{, } \eqref{eq:last} )\\
     \le & 2\delta + 4 \varepsilon & \textrm{(due to } \eqref{eq:mid} ).
\end{align*} 

In conclusion, take $L_2:= \max\{L_{\eqref{eq:propo_appen_lottery_1}},L_{\eqref{eq:propo_appen_lottery_2}},L_{\eqref{eq:propo_appen_lottery_3}} \}$, we conclude the proof of this lemma.
\qed


        \subsection{Remark on the Generalized Lottery Blotto games with continuous CSFs}
        \label{sec:appen_remark_conti_CSF}
In this section, we present and prove the remark stating that under the additional assumption that the CSFs $\zeta_A$ and $\zeta_B$ are Lipschitz continuous on $[0,2X_B] \times [0,2X_B]$, the statements in Theorem~\ref{theo:Lottery_generic_approx} also hold with~\eqref{eq:lottery_theo_A_improve} and~\eqref{eq:lottery_theo_B_improve} (see below) in places of~\eqref{eq:lottery_theo_A} and~\eqref{eq:lottery_theo_B}. For the sake of completeness, we formally state this result as follows.
        
\begin{remark}
    \label{remark_conti_CSF}
   For any CSF $\zeta_A$ and $\zeta_B$ that are Lipschitz continuous on $[0,2X_B] \times [0,2X_B]$, the following results hold (here, we denote $\zeta:=(\zeta_A, \zeta_B)$):
     \begin{itemize}
        \item[(i)] In any game $\LB(\zeta)$, there exists a positive number \mbox{$\varepsilon \le \tilde{\mO} (n^{-1/2})$} such that for any $\gam  \in \Sn$ and $\delta \in \Del$, the following inequalities hold for any pure strategy $\boldsymbol{x}^A$ and $\boldsymbol{x}^B$ of players A and B:

        \begin{align}
            & \Pi^{\zeta}_A(\boldsymbol{x}^A,{\IU_B}) \le \Pi^{\zeta}_A({\IU_A},{\IU_B}) + \left(2\delta + 5{\varepsilon} \right) W_A,\label{eq:lottery_theo_A_improve}\\
            & \Pi^{\zeta}_B({\IU_A},\boldsymbol{x}^B) \le \Pi^{\zeta}_B({\IU_A},{\IU_B})  + \left(2\delta + 5{\varepsilon} \right) W_B. \label{eq:lottery_theo_B_improve}
        \end{align}

        \item[(ii)] For any $\varepsilon \in (0,1]$, there exists a constant $L_{\zeta} >0 $ (that depends on $\zeta$ but does not depend on $\varepsilon$) such that in any game $\LB(\zeta)$ where $ n \ge L_{\zeta} \varepsilon^{-2} \logep$, \eqref{eq:lottery_theo_A_improve} and~\eqref{eq:lottery_theo_B_improve} hold for any $\gam  \in \Sn$, $\delta \in \Del$ and any pure strategy $\boldsymbol{x}^A, \boldsymbol{x}^B$ of players A and~B.
    \end{itemize}         
\end{remark}

%
%
%
\begin{proof}

We define the Lipschitz constant of $\zeta_A, \zeta_B$ respectively by $\mathcal{L}_{\zeta_A}, \mathcal{L}_{\zeta_B}$ and let $\Lip:= \max\{\mathcal{L}_{\zeta_A}, \mathcal{L}_{\zeta_B}\}$. We focus on proving Result~$(ii)$ of this Remark; Result~$(i)$ can be deduced from~Result~$(ii)$ and thus is~omitted. 

\underline{\textit{Step 1}}: We prove that for any $x^*, y^* \in [0,2X_B]$, there exists a constant ${C_{\zeta}}$ (that does not depend on $\varepsilon$ nor $x^*,y^*$) such that for any $n \ge C_{\zeta} \varepsilon^{-2} \logep $, the following inequalities hold:

\begin{align}
    &\left|{\int\limits_0^\infty  {\zeta_A\left( {x ,y^*} \right) \de {\FAn}\left( x \right)}  - \int\limits_0^\infty  {\zeta_A\left( {x,y^*} \right) \de {\FA}\left( x \right)} }  \right|  \le \varepsilon, \label{eq:CSF_intergal_1} \\
    & \left|{\int\limits_0^\infty  {\zeta_A\left( {x^* ,y} \right) \de {\FBn}\left( y \right)}  - \int\limits_0^\infty  {\zeta_A\left( {x^*,y} \right) \de {\FB}\left( y \right)} }  \right| \le \varepsilon \label{eq:CSF_intergal_2}.
\end{align}

The proof of this statement is quite similar to the proof of Lemma~\ref{lem:portmanteau} (see~\ref{sec:appen_proof_lem_portmanteau}). We present here the proof of~\eqref{eq:CSF_intergal_1}; the proof of~\eqref{eq:CSF_intergal_2} can be done similarly. 

Fix $y^* \in [0,2X_B]$; to simplify the notation, we define $f(x):= \zeta_A(x,y^*)$ and $\tep_1:=\frac{\varepsilon}{4 + 4X_B \Lip}$. From~Lemma~\ref{lem:Preliminary}, $\FAn(x) = \FA(x)=1, \forall x > 2X_B$; therefore, the left-hand-side of~\eqref{eq:CSF_intergal_1} can be rewritten as follows.
\begin{equation}
    \left|{\int\limits_0^\infty  {\zeta_A\left( {x ,y^*} \right) \de {\FAn}\left( x \right)}  - \int\limits_0^\infty  {\zeta_A\left( {x,y^*} \right) \de {\FA}\left( x \right)} }  \right|  =  \left| {\int_{0}^{2X_B}  {f(x)}  \de {\FAn}\left( x \right) - \int_{0}^{2X_B}  {f(x)}  \de {\FA}\left( x \right)} \right|. \label{eq:CSF_integral_first}
\end{equation}
Let us define $K := \big \lceil  \frac{2 X_B \Lip}{\tep_1} \big \rceil$ and $K + 1$ points ${x}_j$ such that ${x_0} := 0$ and ${x}_j := {x}_{j-1} +  \frac{2X_B}{K}, \forall j \in [K]$. In other words, we have the partitions \mbox{$[0,2X_B] = \bigcup\nolimits_{j = 1}^{K} {P_j}$} where we denote by $P_1$ the interval $[{x}_0,{x}_1]$ and by $P_j$ the interval $({x}_{j-1}, {x}_{j}]$ for $j = 2, \ldots, K$. For any $ x, x^{\prime} \in P_j, \forall j \in [K]$, since $f$ is Lipschitz continuous, we~have 
\begin{equation}
    |f(x)- f(x^{\prime}) | \le  \Lip|x - x^{\prime}| \le   \Lip \frac{2X_B}{K} \le \tep_1 . \label{eq:CSF_intergral_lipsc}    
\end{equation}
Now, we define the function $g(x)\!:=\! \sum \limits_{j=1}^{K} {f(x_j) \boldsymbol{1}_{P_j} (x)}$. Here, $\boldsymbol{1}_{P_j}$ is the indicator function of the set ${P_j}$. From this definition and Inequality~\eqref{eq:CSF_intergral_lipsc}, we have \mbox{$| f(x) - g(x) | \le \tep_1$}, $\forall x \in [0,2X_B]$. Therefore, 
\begin{align}
& \left| {\int_{0}^{2X_B}  {f(x)}  \de {\FAn}\left( x \right) - \int_{0}^{2X_B}  {g(x)}  \de {\FAn}\left( x \right)} \right| \le \int_{0}^{2X_B} \tep_1 \de \FAn(x) \le \tep_1, \label{eq:CSF_intergral_triangle1}\\
& \left| {\int_{0}^{2X_B}  {f(x)}  \de {\FA}\left( x \right) - \int_{0}^{2X_B}  {g(x)}  \de {\FA}\left( x \right)} \right| \le \int_{0}^{2X_B} \tep_1 \de \FA(x) \le \tep_1 .\label{eq:CSF_integral_triangle2}
\end{align}
Now, we note that for any $j \in [K]$, $f(x_j) = \sum \limits_{m=0}^{j} \left[{f(x_m) -  f(x_{m-1})} \right]$; here, by convention, we denote by $x_{-1}$ an arbitrary negative number and set $f(x_{-1}) = 0$. Using this, we have:
\begin{align}
     & \left| {\int_{0}^{2X_B}  {g(x)}  \de {\FAn}\left( x \right) - \int_{0}^{2X_B}  {g(x)}  \de {\FA}\left( x \right)} \right| \nonumber \\
    = & \left| \sum \limits_{j=1}^K f(x_j) \left[{  \int_{0}^{2X_B} \boldsymbol{1} _{P_j}\left( x\right)} \de \FAn(x) -  {\int_{0}^{2X_B}   \boldsymbol{1} _{P_j}\left(x\right) \de \FA(x) }  \right] \right| \nonumber \\ %
    = & \left| \sum \limits_{j=1}^K  f(x_j) \left[  \prob\left(A^n_i \in P_j \right)\!- \! \prob\left(A^*_i \in P_j \right) \right] \right| \nonumber \\
    = & \left| \sum \limits_{j=1}^K \left(  \sum \limits_{m=0}^{j}
    \left[   {f(x_m) -  f(x_{m-1})} \right]   \left[\prob\left(A^n_i \in P_j \right)\!- \! \prob\left(A^*_i \in P_j \right) \right]  \right)      \right| \nonumber \\
    \le & \left| \left[  f(x_0) -  f(x_{-1}) \right] \sum \limits_{j=1}^K  \left[ \prob\left(A^n_i \in P_j \right)\!- \! \prob \left(A^*_i \in P_j \right)  \right]  \right|   \nonumber \\
        &  \qquad + \left| \sum \limits_{m=1}^K \left( \left[  f(x_m) -  f(x_{m-1}) \right] \sum \limits_{j=m}^K  \left[ \prob\left(A^n_i \in P_j \right)\!- \! \prob\left(A^*_i \in P_j \right)  \right] \right) \right|. \label{eq:CSF_intergral_bridge}
\end{align}
Note that \mbox{$\prob\left(A^n_i \in P_j \right)\!- \! \prob \left(A^*_i \in P_j \right) \!=\!  \FAn(x_j) \!-\!\FAn(x_{j-1})\! -\! \FA(x_j)\! +\! \FA(x_{j-1}) $}.\footnote{For any $j \! \ge \! 2$, this is trivially since $P_j \!:= \! (x_{j-1}, x_j]$. For $P_1 = [0,x_1]$,~we have~that \mbox{$\prob\left(A^n_i \in P_1 \right)\!- \! \prob \left(A^*_i \in P_1 \right) =   \prob(A^n_i \in (0, x_1] ) \!- \!\prob(A^*_i \in (0, x_1] )  \!+\! \prob(A^n_i = 0)\! - \!\prob(A^*_i = 0) $}; moreover, due to Lemma~\ref{lem:continuity_Ani_and_Bni}-$(i)$, we also note that \mbox{$\prob(A^n_i \!= \!0) \!=\! \prob(A^*_i\!=\! 0)$}.} Now, we can rewrite the first term in \eqref{eq:CSF_intergral_bridge} as follows.
\begin{align}
    & \left| \left[  f(x_0) -  f(x_{-1}) \right] \sum \limits_{j=1}^K  \left[ \prob\left(A^n_i \in P_j \right)\!- \! \prob \left(A^*_i \in P_j \right)  \right]  \right| \nonumber \\
    = & \left| f(0) \cdot \left[ \sum \limits_{j=1}^K  \left(\FAn(x_j) - \FAn(x_{j-1}) - \FA(x_j) + \FA(x_{j-1}) \right)   \right]  \right| \nonumber \\
    = & \left| f(0) \cdot \left[\FAn(x_K) - \FAn(x_0) - \FA(x_K) + \FA(x_{0})   \right]  \right| \nonumber \\
    \le &2 \sup \limits_{x \in [0, \infty)} {\left| \FAn(x) - \FA(x) \right| } \label{eq:CSF_intergral_final1}.
\end{align}
Here, the last inequality comes from the fact that $f(x) \le 1, \forall x \in [0,2X_B]$ (since it is a CSF).

Now, we recall that for any $m \in [K]$, \mbox{$ f(x_m) -  f(x_{m-1}) \le  \frac{2X_B \Lip}{K}$}. Therefore, the second term in \eqref{eq:CSF_intergral_bridge}~is
\begin{align}
     &\left| \sum \limits_{m=1}^K \left( \left[  f(x_m) -  f(x_{m-1}) \right] \sum \limits_{j=m}^K  \left[ \prob\left(A^n_i \in P_j \right)\!- \! \prob\left(A^*_i \in P_j \right)  \right] \right) \right| \nonumber \\
     = & \left| \sum \limits_{m=1}^K \left(  \left[  f(x_m) -  f(x_{m-1}) \right] \sum \limits_{j=m}^K  \left[\FAn(x_j) - \FAn(x_{j-1}) - \FA(x_j) + \FA(x_{j-1})   \right]  \right) \right| \nonumber \\
     = & \left| \sum \limits_{m=1}^K \left(  \left[  f(x_m) - f(x_{m-1}) \right] \left[\FAn(x_K) - \FAn(x_{m-1}) - \FA(x_K) + \FA(x_{m-1})   \right]  \right) \right| \nonumber \\
     \le & \sum \limits_{m=1}^K  \frac{2X_B \Lip}{K} \cdot  2 \sup \limits_{x \in [0, \infty)} {\left| \FAn(x) - \FA(x) \right| } \nonumber \\
     = & 4 X_B \Lip \sup \limits_{x \in [0, \infty)} {\left| \FAn(x) - \FA(x) \right| }. \label{eq:CSF_intergral_final2}
\end{align}
Inject~\eqref{eq:CSF_intergral_final1} and~\eqref{eq:CSF_intergral_final2} into~\eqref{eq:CSF_intergral_bridge}, we obtain that
\begin{equation}
    \left| {\int_{0}^{2X_B}  {g(x)}  \de {\FAn}\left( x \right) - \int_{0}^{2X_B}  {g(x)}  \de {\FA}\left( x \right)} \right| \le \left( 2 + 4 X_B \Lip \right) \sup \limits_{x \in [0, \infty)} {\left| \FAn(x) - \FA(x) \right| }. \label{eq:CSF_intergral_triangle3}
\end{equation}
Apply the triangle inequality and combine~\eqref{eq:CSF_intergral_triangle1},~\eqref{eq:CSF_integral_triangle2}, \eqref{eq:CSF_intergral_triangle3}, we have that:
\begin{equation}
   \left| {\int_{0}^{2X_B}  {f(x)}  \de {\FAn}\left( x \right) - \int_{0}^{2X_B}  {f(x)}  \de {\FA}\left( x \right)} \right| \le 2 \tep_1 + (2+ 4 X_B \Lip)\sup \limits_{x \in [0, \infty)} {\left| \FAn(x) - \FA(x) \right|}. \label{eq:CSF_intergral_the_end}
\end{equation}
Recall the constant $C_1$ indicated in Lemma~\ref{lem:convergence}, we define \mbox{$C_{\zeta} \!:= \!  C_1 \left(4 \!+\! 4 X_B \Lip \right)^2 \left[\ln(4 \!+\! 4 X_B \Lip) \!+\! 1 \right]$} (note that $C_{\zeta}$ does not depend on $n$ nor $\varepsilon$) and deduce that\footnote{Apply Lemma~\ref{lem:log_pre}, \mbox{$C_{\zeta} \varepsilon^{-2} \logep \!= \!C_1 \left( \frac{4 \! +\! 4X_B \Lip}{\varepsilon} \right)^2 \left[\ln(4 \!+\! 4X_B \Lip) \!+\! 1 \right] \logep $} \mbox{$\ge\! C_1  \left( \frac{4 \! + \! 4X_B \Lip}{\varepsilon} \right)^2 \ln\left(\frac{4 \! + \! 4X_B \Lip}{\varepsilon} \right)$}. Moreover, since $\frac{\varepsilon}{4+4X_B \Lip} =  \min\left\{\frac{\varepsilon}{4+4X_B \Lip}, \frac{1}{\e} \right\} = \min\left\{\tep_1, \frac{1}{\e} \right\} $ (due to the fact that $\tep_1 = \frac{\varepsilon}{4+4X_B \Lip} < \frac{1}{\e}$).} 
\begin{equation*}
    C_{\zeta} \varepsilon^{-2} \logep \ge C_1 \tep_1^{-2} \ln\left(\frac{1}{\min\{\tep_1, 1/\e\}} \right).    
\end{equation*}
Take $\varepsilon_1:= \tep_1$, for any $n \ge C_{\zeta} \varepsilon^{-2} \logep$, we have $n \ge  C_1 \varepsilon_1^{-2} \logone$ and by applying Lemma~\ref{lem:convergence}, we obtain that $\sup \limits_{x \in [0, \infty)} {\left| \FAn(x) - \FA(x) \right| } \le \varepsilon_1 = \tep_1 $ and thus by~\eqref{eq:CSF_integral_first} and~\eqref{eq:CSF_intergral_the_end}, we~have:
\begin{equation*}
    \left|{\int\limits_0^\infty  {\zeta_A\left( {x ,y^*} \right) \de {\FAn}\left( x \right)}  - \int\limits_0^\infty  {\zeta_A\left( {x,y^*} \right) \de {\FA}\left( x \right)} }  \right| \le 2 \tep_1 + \left( 2 + 4 X_B \Lip \right) \tep_1 = (4 + 4 X_B \Lip) \tep_1 = \varepsilon. 
\end{equation*}
This is exactly~\eqref{eq:CSF_intergal_1}.

\underline{\textit{Step 2}}: Based on~\eqref{eq:CSF_intergal_1} and~\eqref{eq:CSF_intergal_2}, we can trivially deduce that the following inequalities hold for any $n \ge C_{\zeta} \varepsilon^{-2} \logep$ and $i \in [n]$:
    \begin{align}
    & \left|{\int\limits_0^\infty  {\zeta_A\left( {x ,y} \right) \de {\FBn}\left( y \right)}  - \int\limits_0^\infty  {\zeta_A\left( {x,y} \right) \de {\FB}\left( y \right)} }  \right|  \le   \varepsilon , \forall x \ge 0, \label{eq:1} \\
    &\left|\int\limits_0^\infty  {\int\limits_0^\infty  {\zeta_A\left( {x,y} \right) \de {\FB}\left( y \right)  \de {\FA}\left( x \right)}  } - \int\limits_0^\infty  {\int\limits_0^\infty  {\zeta_A\left( {x,y} \right)  \de {\FBn}\left( y \right) \de {\FA}\left( x \right)} } \right|  \le  \varepsilon,  \label{eq:2} \\
    & \left|\int\limits_0^\infty  {\int\limits_0^\infty  {\zeta_A\left( {x,y} \right)  \de {\FA}\left( x \right) \de {\FBn}\left( y \right)} } - \int\limits_0^\infty  {\int\limits_0^\infty  {\zeta_A\left( {x,y} \right)  \de {\FAn}\left( x \right) \de {\FBn}\left( y \right)} } \right|  \le   \varepsilon. \label{eq:3}
\end{align}
We notice that the left-hand-sides of these inequalities are exactly the terms considered in Lemma~\ref{lem:proof_lottery_LBn_prepare2}; moreover, the upper-bounds given in\eqref{eq:1}, \eqref{eq:2} and \eqref{eq:3}  are smaller than that in \eqref{eq:propo_appen_lottery_1}, \eqref{eq:propo_appen_lottery_2} and \eqref{eq:propo_appen_lottery_3} of Lemma~\ref{lem:proof_lottery_LBn_prepare2}.

\underline{\textit{Step 3}}: To complete the proof of Remark~\ref{remark_conti_CSF}, we follow the proof of Theorem~\ref{theo:Lottery_generic_approx} where we use \eqref{eq:1}, \eqref{eq:2} and \eqref{eq:3} instead of \eqref{eq:propo_appen_lottery_1}, \eqref{eq:propo_appen_lottery_2} and \eqref{eq:propo_appen_lottery_3}. By doing this, we obtain \eqref{eq:lottery_theo_A_improve} and \eqref{eq:lottery_theo_B_improve}.     \qed
\end{proof}

\section{Proof of Lemma~\ref{lem:delta_mu_nu} and Theorem~\ref{theoratioform}}

        \subsection{Proof of Lemma~\ref{lem:delta_mu_nu}}
        \label{sec:appen_proof_ratio-form}

\textit{$(i)$ We first consider the games $\LB(\mu^R)$.}

    \textit{\underline{Step 1:} We want to prove that there exists $\delta_0 = \mO(\varepsilon^{-1/R} - 1)$ such that \mbox{$  \Xmu(y^*, \varepsilon) \subset [y^* - \delta_0, y^* + \delta_0]$}} for any $y^* \in [0,2X_B]$. Note that this is trivial if $\Xmu(y^*, \varepsilon) = \emptyset$. In the following, we consider the case where $\Xmu(y^*, \varepsilon) \neq \emptyset$. We denote by $f: [0, 2X_B]  \times [0,2X_B] \rightarrow [0,1]$ the~function:
    \begin{equation*}
        f(x, y^*) := | \mu^R_A(x,y^*) - \beta_A(x, y^*)| =  \left\{ \begin{array}{l}
    \frac{\alpha x^R}{\alpha x^R + (1-\alpha ) (y^*)^R }, \text{ if } x < y^* \\
    0 ,\text{ if } x = y^* \\
     1 - \frac{\alpha x^R}{\alpha x^R + (1-\alpha ) (y^*)^R }, \text{ if } x  > y^*
    \end{array} \right..
    \end{equation*}
    Trivially, $y^* \notin \Xmu(y^*, \varepsilon)$. Take an arbitrary $x \in \Xmu(y^*, \varepsilon)$. If $ x < y^*$, we have
    \begin{equation*}
        f(x,y^*) \ge \varepsilon \Rightarrow \frac{\alpha x^R}{\alpha x^R + (1-\alpha ) (y^*)^R } \ge \varepsilon \Rightarrow \frac{x}{y^*} \ge \left(\frac{\varepsilon}{1-\varepsilon} \frac{1- \alpha}{\alpha}\right)^{1/R}.
    \end{equation*}
    Therefore, $0< y^*- x \le y^* \left[ 1 -\left(\frac{\varepsilon}{1-\varepsilon} \frac{1- \alpha}{\alpha}\right)^{1/R}  \right]$. Here, we note that the right-hand side is positive (due to the condition $\varepsilon < \alpha$); moreover, it is upper-bounded by $\mO(1 - \varepsilon^{1/R}  ) \le \mO(\varepsilon^{-1/R} - 1) $.

    On the other hand, if $x > y^*$, we have:
     \begin{equation*}
        f(x,y^*) \ge \varepsilon  \Rightarrow   1 - \frac{\alpha  x^R}{\alpha x^R + (1-\alpha ) (y^*)^R } \ge \varepsilon  \Rightarrow \frac{x}{y^*} \le \left(\frac{1-\varepsilon}{\varepsilon} \frac{1- \alpha}{\alpha}\right)^{1/R}.
    \end{equation*}
      Therefore we have $0< x - y^*  \le y^* \left[\left(\frac{1-\varepsilon}{\varepsilon} \frac{1- \alpha}{\alpha}\right)^{1/R} - 1 \right]$. Here the right-hand side is positive (due to the condition $\alpha+\varepsilon < 1$) and is upper-bounded by $\mO(\varepsilon^{-1/R} - 1)$.

   In conclusion, for any $\varepsilon < \min\{\alpha, 1\!- \!\alpha\}$, there exists $\delta_0 \!= \!\mO(\varepsilon^{-1/R}\! -\! 1)$ such that \mbox{$\Xmu(y^*, \varepsilon) \subset [y^*\! -\! \delta_0, y^*\! +\! \delta_0]$}. Note that a similar proof can be done to prove that there exists \mbox{$\hat{\delta}_0 = \mO(\varepsilon^{-1/R} - 1) $ such that for any $x^* \in [0, 2X_B]$}, \mbox{$\Ymu(x^*, \varepsilon) \subset [x^* - \hat{\delta}_0, x^* + \hat{\delta}_0] $}.

\paragraph{\underline{Step 2:}} For any $y^* \in [0,2X_B]$ and $\delta_0 \ge 0$, let us define the set \mbox{$I_0 (y^*) \!:=\! [y^* \! -\! \delta_0, y^* \!+\! \delta_0] \bigcap [0, 2 X_B] $}; we want to show that $\int_{x \in I_0 (y^*)}  \de  \FA(x)  \le \frac{2 n \Lmax \delta_0 \wmax}{\wmin}, \forall i \in [n]$.
   
    \textit{{Case 1:}} For $i \in \Ostar$, then $A^*_i \AS$, we have that
    \begin{align*}
        \int_{x \in I_0 (y^*)}  \de  \FA(x)         \le & {F_{ \AS}}\left( {y^* + \delta_0 } \right) - {F_{ \AS}}\left( {y^* - \delta_0  } \right) \\
        =	& \left\{ \begin{array}{l}
        \frac{(y^* + \delta_0 ) \LdaB }{v^B_i}  \le \frac{2\delta_0 \LdaB }{v^B_i}  ,\text{~if~}0 \le y^* \le \delta_0 \\
        \frac{(y^* + \delta_0 ) \LdaB }{v^B_i} - \frac{(y^* - \delta_0 ) \LdaB }{v^B_i} =  \frac{ 2\delta_0 \LdaB }{v^B_i} , \text{~if~} \delta_0   \le y^* < \frac{v^B_i}{\LdaB} - \delta_0  \\
        1 - \frac{(y^* - \delta_0 ) \LdaB }{v^B_i} = \frac{v^B_i - y^* \LdaB + \delta_0 \LdaB }{v^B_i}\le \frac{ 2\delta_0 \LdaB }{v^B_i} ,\text{~if~}\frac{v^B_i}{\LdaB} - \delta_0  \le y^* \le \frac{v^B_i}{\LdaB} + \delta_0  \\
        1 - 1 = 0 , \text{~otherwise}
        \end{array} \right.\\
        \le &\frac{2 n \Lmax \delta_0 \wmax}{\wmin}.
    \end{align*}

\textit{{Case 2:}} For $i \notin \Ostar$, then $A^{*}_i= \AW$. We have
    \begin{align*}
        \int_{x \in I_0 (y^*)}  \de  \FA(x)         \le & {F_{ \AW}}\left( {y^* + \delta_0 } \right) - {F_{ \AW}}\left( {y^* - \delta_0  } \right) \\
         =	&  \left\{ \begin{array}{l}
         \frac{(y^* + \delta_0) \LdaB }{v^B_i} \le \frac{2 n \Lmax \delta_0 \wmax}{\wmin} ,\text{~if~}0 \le y^* \le \delta_0 
         \\
        \frac{(y^* + \delta_0) \LdaB }{v^B_i} - \frac{(y^* - \delta_0) \LdaB }{v^B_i} =  \frac{ 2\delta_0\LdaB }{v^B_i} , \text{~if~} \delta_0  < y^* < \frac{v^A_i}{\LdaA} -\delta_0\\
        1 - \frac{\frac{v^B_i}{\LdaB} - \frac{v^A_i}{\LdaA}}{\frac{v^B_i}{\LdaB}} -\frac{(y^* - \delta_0) \LdaB }{v^B_i} = \frac{v^A_i \frac{\LdaB}{\LdaA} - y^* \LdaB + \delta_0 \LdaB }{v^B_i}\le \frac{ 2\delta_0 \LdaB }{v^B_i} ,\text{~if~}\frac{v^A_i}{\LdaA} - \delta_0  \le y^* \le \frac{v^A_i}{\LdaA} + \delta_0\\
        1 - 1 = 0 , \text{~otherwise}
        \end{array} \right.\\
        \le &\frac{2 n \Lmax \delta_0 \wmax}{\wmin}.
    \end{align*}
Note that we also can similarly prove that for any $ x^* \in [0, 2X_B] $ and $\delta_0 \ge 0$, for any $i \in [n]$, we also have $\int_{y \in I_0 (x^*)}  \de  \FB(y) \le \frac{2 n \Lmax \delta_0 \wmax}{\wmin}$.

\paragraph{\underline{Step 3:} Conclusion.} We note that all random variable $ A^*_i, B^*_i, i \in [n]$ are bounded in $[0, 2X_B]$; therefore, for any $x^*, y^* \in [0, 2X_B]$ and $\delta_0 \ge 0$, we have:
\begin{align*}
    \int_{x \in [y^* - \delta, y^* + \delta_0]}  \de  \FA(x)   = \int_{x \in I_0 (y^*)}  \de  \FA(x)  & \textrm{ and }   & \int_{y \in [x^* - \delta, x^* + \delta_0]}  \de  \FB(y)   = \int_{y \in I_0 (x^*)}  \de  \FB(x).
\end{align*}

Let us define $\delta_\mu: = min\{ 1, \frac{2 n \Lmax \delta_0 \wmax}{\wmin} \}= \mO\left(n (\varepsilon^{-1/R}-1) \right)$ and we conclude that:
\begin{equation*}
    \max \left\{ \max_{y^* \in [0,2X_B]}{\int \nolimits_{\Xmu(y^*,\varepsilon)} { \de \FA(x)}}, \max_{x^* \in [0,2X_B]}{\int \nolimits_{\Ymu(x^*,\varepsilon)} { \de \FB(y)}} \right\}  \le \delta_\mu.
\end{equation*}
This implies that $\delta_\mu \in \Delmu$.

\paragraph{$(ii)$ We now turn our focus on the games $\LB(\nu^R)$.} We first prove the existence of $\delta_1>0$ such that $\Xnu(y^*, \varepsilon) \subset [y^* - \delta_1, y^* + \delta_1]$ for any $y^* \in [0,2X_B]$. Similar to step 1 in the above analysis for the game $\LB(\mu^R)$, we denote by $g: [0, 2X_B]  \times [0,2X_B] \rightarrow [0,1]$ the~function:
 \begin{equation*}
        g(x, y^*) := | \nu^R_A(x,y^*) - \beta_A(x, y^*)| =  \left\{ \begin{array}{l}
    \frac{\alpha e^{xR}}{\alpha e^{xR} + (1-\alpha ) e^{y^* R}} \text{ , if } x < y^*, \\
    0 \text{ , if } x = y^*, \\
     1 - \frac{\alpha e^{xR}}{\alpha e^{xR} + (1-\alpha ) e^{y^* R}}  \quad \text{ , if } x  > y^*.
    \end{array} \right.
    \end{equation*}
    Trivially, $y^* \notin \Xnu(y^*, \varepsilon)$. Take an arbitrary $x \in \Xmu(y^*, \varepsilon)$. If $ x < y^*$, we have
    \begin{equation*}
        g(x,y^*) \ge \varepsilon \Rightarrow \frac{\alpha e^{xR}}{\alpha e^{xR} + (1-\alpha ) e^{y^* R} } \ge \varepsilon.
    \end{equation*}
    Therefore, $0< y^*- x \le \frac{1}{R} \ln \left(\frac{1-\varepsilon}{\varepsilon} \frac{\alpha}{1-\alpha}\right)$. Here, we note that the right-hand side is positive (due to the condition $\varepsilon < \alpha$).

    On the other hand, if $x > y^*$, we have:
     \begin{equation*}
        g(x,y^*) \ge \varepsilon  \Rightarrow   1 - \frac{\alpha e^{xR}}{\alpha e^{xR} + (1-\alpha ) e^{y^* R} } \ge \varepsilon.
    \end{equation*}
      Therefore, $0< x - y^* \le  \frac{1}{R} \ln \left( \frac{1-\varepsilon}{\varepsilon} \frac{1-\alpha}{ \alpha}\right)$. Here, the right-hand side is positive (due to the condition $\alpha+\varepsilon < 1$).
    %
    %

    
    
    In conclusion, let us denote $\delta_1 = \mathcal{O} (R^{-1} \ln(\varepsilon^{-1}))$, we have proved that $\Xnu(y^*,\varepsilon) \subset [y^*- \delta_1, y^* + \delta_1]$ for any $ y^* \in [0,2X_B]$. Now, we define $I_1 (y^*) := [y^* - \delta_1, y^* + \delta_1] \bigcap [0, 2 X_B]$. Similar to step 2 of the above analysis regarding the game $\LB(\mu^R)$, we can prove that $\int \nolimits_{I_1(y^*)} \de \FA(x) \le 2 n \Lmax \delta_1 \wmax /\wmin$ for any $y^* \in [0,2X_B]$.~Therefore, 
    \begin{equation*}
    \max \left\{ \max_{y^* \in [0,2X_B]}{\int \nolimits_{\Xnu(y^*,\varepsilon)} { \de \FA(x)}}, \max_{x^* \in [0,2X_B]}{\int \nolimits_{\Ynu(x^*,\varepsilon)} { \de \FB(y)}} \right\}  \le \delta_\nu,
\end{equation*}
where $\delta_\nu := \min\{1,\frac{2 n \Lmax \delta_1 \wmax}{\wmin}\} = \mO\left(n R^{-1} \ln(\varepsilon^{-1}) \right)$ and $\delta_{\nu} \in \Delnu$.\qed

    \subsection{Proof of Theorem~\ref{theoratioform}}
    \label{sec:appen_proof_theoratioLB}
    \theoratioform*

\begin{proof}

Take $\varepsilon = \barep / 21$ and $\tilde{L} = L^* 21^2 (\ln(21)+1)$ (where $L^*$ is indicated in Theorem~\ref{theo:Lottery_generic_approx}). We note that \mbox{$\tilde{L} \barep^{-2} \logbar  \ge L^* \varepsilon \logep $};\footnote{Note that $\varepsilon = \barep / 21$ and apply Lemma~\ref{lem:log_pre} to have that $\tilde{L} \barep^{-2} \logbar \ge L^* \left( \frac{21}{\barep} \right)^2 \ln \left( \frac{21}{\barep}\right)$; moreover, we recall that $\frac{\barep}{21} < \frac{1}{\e}$; therefore, $\ln\left(\frac{21}{\barep} \right)  = \ln \left( \frac{1}{\min\{\barep/21, 1/\e \}} \right) = \logep$.}
therefore, for any $n \ge \tilde{L} \barep^{-2} \logbar $, we have \mbox{$n \ge L^* \varepsilon \logep$} and thus, apply Theorem~\ref{theo:Lottery_generic_approx}-$(ii)$,  for any $R>0$, the $\IU$ strategy is an \mbox{$(8\delta_\mu+ 13\varepsilon) W$}-equilibrium of the game $\LB(\mu^R)$ (recall that $W:= \max \{W_A, W_B \}$). Similarly, the $\IU$ strategy is an $(8\delta_\nu + 13\varepsilon) W$-equilibrium of the game $\LB(\nu^R)$). 

We first consider the game $\LB(\mu^R)$. Now, by applying Lemma~\ref{lem:delta_mu_nu}, for any $\gam  \in \Sn$ and any {$R\ge \mO\left(\ln\left(  \frac{1}{\varepsilon} \right) \ln \left( \frac{\varepsilon}{n} +1 \right) \right) = \mO \left( \ln\left(  \frac{1}{\varepsilon} \right) \frac{n}{\varepsilon} \right)$}, we have $\delta_\mu \le \varepsilon$. Therefore, we deduce that for any \mbox{$n \ge \tilde{L} \barep^{-2} \logbar $}, \mbox{$R \ge \mO \left(  \frac{n}{\barep} \ln\left(  \frac{1}{\barep} \right)   \right) $}, the $\IU$ strategy is an \mbox{$21 \varepsilon W$}-equilibrium (i.e., \mbox{$\barep W$}-equilibrium) of the game~$LB(\mu^R)$.

Similarly, apply Lemma~\ref{lem:delta_mu_nu}, for any $\gam  \in \Sn$ and $R\ge \mO \left( \frac{n}{\barep}\ln\left(   \frac{1}{\barep} \right)     \right)$, we have $\delta_\nu \le \varepsilon$. Therefore, for any $n \ge \tilde{L} \barep^{-2} \logbar$, $R \ge \mO \left(\frac{n}{\barep} \ln\left(  \frac{1}{\barep} \right)     \right) $, the $\IU$ strategy is an \mbox{$21 \varepsilon W$}-equilibrium (i.e., \mbox{$\barep W$}-equilibrium) of the game~$LB(\nu^R)$. \qed

\end{proof}

\end{document}